\pdfoutput=1
\documentclass[a4paper]{amsart}
\usepackage[utf8]{inputenc}
\usepackage{amsmath}
\usepackage{amssymb}
\usepackage{calc}
\usepackage{float}
\usepackage{xcolor}
\usepackage{svg}
\usepackage{amsthm}
\usepackage{amsfonts}
\usepackage{enumitem}
\usepackage{ccicons}
\usepackage[ruled]{algorithm2e}
\usepackage{hyperref}
\hypersetup{
    colorlinks=true,
    linkcolor=blue,
    citecolor=darkgray
}

\usepackage[giveninits]{biblatex}
\AtEveryBibitem{
  \clearlist{language}
}

\usepackage{tikz}
\tikzset{every picture/.style={line width=0.75pt}}

\theoremstyle{theorem}
\newtheorem{theorem}{Theorem}[section]
\newtheorem{lemma}[theorem]{Lemma}
\newtheorem{proposition}[theorem]{Proposition}
\newtheorem{corollary}[theorem]{Corollary}

\theoremstyle{definition}
\newtheorem{definition}[theorem]{Definition}
\newtheorem{remark}[theorem]{Remark}
\newtheorem{example}[theorem]{Example}
\newtheorem*{excont1}{Example \ref{ex:lattice} continued}
\newtheorem*{excont2}{Example \ref{ex:bio} continued}

\renewcommand{\d}{\mathrm d}
\renewcommand{\i}{\mathrm{i}}
\newcommand{\diffop}{\mathcal{D}}
\newcommand{\diffops}[1]{#1[z]\langle d/dz\rangle}
\providecommand{\C}{\mathbb{C}}
\newcommand{\R}{\mathbb{R}}
\newcommand{\N}{\mathbb{N}}
\newcommand{\Q}{\mathbb{Q}}
\newcommand{\Qbar}{\overline{\mathbb{Q}}}
\newcommand{\K}{\mathbb{K}}
\newcommand{\Z}{\mathbb{Z}}

\newcommand{\mP}{\mathcal{P}}
\newcommand{\mB}{\mathcal{B}}
\newcommand{\mS}{\mathcal{S}}
\newcommand{\mL}{\mathcal{L}}
\newcommand{\mI}{\mathcal{I}}
\newcommand{\by}{\mathbf{y}}
\newcommand{\bc}{\mathbf{c}}
\newcommand{\bC}{\mathbf{C}}
\newcommand{\Rel}{\operatorname{Re}}
\renewcommand{\Im}{\operatorname{Im}}
\newcommand{\CBF}{\CC^{\bullet}}
\newcommand{\rhol}{\smash{\overline\Q}^{\mathrm{rhol}}}
\newcommand{\rholgamma}{\smash{\overline\Q}^{\mathrm{rhol}, \Gamma}}
\newcommand{\domsing}{\Xi^{\mathrm d}}
\newcommand{\domsingapprox}{\tilde\Xi^{\mathrm d}}
\newcommand{\anasing}{\Xi^{\mathrm a}}
\newcommand{\esing}{\Xi^{\mathrm e}}
\newcommand{\B}{\mathbf B}
\newcommand{\I}{\mathbf I}
\newcommand{\abs}[1]{\mathopen|#1\mathclose|}
\newcommand{\CC}{\mathbb{C}}

\allowdisplaybreaks[2]

\begin{document}

\title[Asymptotic Bounds for P-Recursive Sequences]{Computing Error Bounds for Asymptotic Expansions of Regular P-Recursive Sequences}

\thanks{%
\ccby\:
This work is licensed under a Creative Commons Attribution 4.0 International License
(\url{http://creativecommons.org/licenses/by/4.0/}).}

\author{Ruiwen Dong}
  \address{
    Ruiwen Dong,
    École polytechnique, Institut polytechnique de Paris,
    91200 Palaiseau,
    France}
  \curraddr{
    Department of Computer Science, University of Oxford,
    OX1 3QG,
    Oxford,
    United Kingdom}
  \email{ruiwen.dong@kellogg.ox.ac.uk}

\author{Stephen Melczer}
  \address{
    Stephen Melczer,
    Department of Combinatorics and Optimization, University of Waterloo,
    Waterloo, Ontario,
    Canada}
  \email{smelczer@uwaterloo.ca}

\author{Marc Mezzarobba}
  \address{
    Marc Mezzarobba,
    LIX, CNRS, École polytechnique, Institut polytechnique de Paris,
    91200 Palaiseau,
    France}
  \thanks{%
    MM's work is supported in part by ANR grants ANR-19-CE40-0018 DeRerumNatura
    and ANR-20-CE48-0014-02 NuSCAP. SM's work was supported by NSERC Discovery Grant RGPIN-2021-02382.}
  \email{marc@mezzarobba.net}

\thanks{
Author Accepted Manuscript.
For the purpose of Open Access, a CC-BY public copyright licence
has been applied by the authors to the present document.}

\begin{abstract}
  Over the last several decades, improvements in the fields of analytic combinatorics
  and computer algebra have made determining the asymptotic behaviour of sequences
  satisfying linear recurrence relations with polynomial coefficients largely a 
  matter of routine, under assumptions that hold often in practice. The algorithms
  involved typically take a sequence, encoded by a recurrence
  relation and initial terms, and return the leading terms in an asymptotic expansion
  up to a big-O error term. Less studied, however, are effective techniques giving
  an explicit bound on asymptotic error terms. Among other things, such explicit
  bounds typically allow the user to automatically prove sequence positivity (an 
  active area of enumerative and algebraic combinatorics) by exhibiting an index 
  when positive leading asymptotic behaviour dominates any error terms.

  In this article, we present a practical algorithm for computing such asymptotic
  approximations with rigorous error bounds, under the assumption that the
  generating series of the sequence is a solution of a differential
  equation with regular (Fuchsian) dominant singularities.
  Our algorithm approximately follows the singularity analysis method of
  Flajolet and Odlyzko, except that all big-O terms involved in the derivation
  of the asymptotic expansion are replaced by explicit error terms.
  The computation of the error terms combines analytic bounds from the
  literature with effective techniques from rigorous numerics and computer algebra.
  We implement our algorithm in the SageMath computer algebra system and exhibit its
  use on a variety of applications (including our original motivating example,
  solution uniqueness in the Canham model for the shape of genus one
  biomembranes).
\end{abstract}

\maketitle

\section{Introduction}

\subsection{Context and motivation}

A sequence $(f_n) \in \K^{\N}$ is said to be \emph{P-recursive} over a field~$\K$ if it satisfies a linear recurrence relation
\begin{equation}
\label{eq:lrr_intro}
0 = c_r(n)f_{n+r} + \cdots + c_1(n)f_{n+1} + c_0(n)f_n
\end{equation}
with polynomial coefficients $c_j(n) \in \K[n]$. The sequence $(f_n)$ is P-recursive if and only if its \emph{generating series} (or \emph{generating function}) $f(z) = \sum_{n \geq 0}f_nz^n$ is \emph{D-finite} as a formal power series, meaning the series satisfies a formal linear differential equation
\begin{equation}
\label{eq:lde_intro}
0 = p_{q}(z) f^{(q)}(z) + \cdots + p_1(z) f'(z) + p_0(z) f(z)
\end{equation}
with polynomial coefficients $p_j(z) \in \K[z]$. A complex-valued function $f(z)$ that satisfies an equation of the form~\eqref{eq:lde_intro} with $\K = \C$ is also called \emph{D-finite}. Given sufficiently many initial terms, the sequence~$(f_n)$ is uniquely determined by either the linear recurrence relation \eqref{eq:lrr_intro} or the linear differential equation~\eqref{eq:lde_intro}.  Numerous sequences arising in combinatorics and the analysis of algorithms are P-recursive, while many elementary and special functions have D-finite power series.

\begin{example}[Lattice walks in $\N^2$]
\label{ex:lattice}
The \emph{kernel method}~\cite[Chapter 4]{Melczer2021}, a standard technique used to study lattice path families restricted to cones, implies that the generating function $f(z)$ for the number $f_n$ of lattice walks beginning at the origin, staying in $\N^2$, and taking $n$ steps in $S = \{(1, 0), (-1, 0), (0, 1), (0, -1)\}$ satisfies the linear differential equation
\begin{multline*}
  z^2(4z - 1)(4z + 1)f'''(z) + 2z(4z+1)(16z-3)f''(z) \\
  + 2(112z^2 + 14z - 3)f'(z) + 4(16z + 3)f(z) = 0.
\end{multline*}
Standard generating function manipulations then imply that $f_n$ satisfies the linear recurrence
\[
(n+4)(n+3)f_{n+2} - 4(2n+5)f_{n+1} - 16(n+1)(n+2)f_n = 0 \,,
\]
and is uniquely defined by this recurrence and the initial terms
\[
(f_0, f_1, f_2) = (1, 2, 6).
\]
Characterizing when the generating function of a lattice path model in $\N^2$ is D-finite has been an active corner of enumerative combinatorics in recent years; see~\cite[Chapter 4]{Melczer2021} for an overview of the techniques and results in this area, and~\cite{mohanty2014lattice} for a broad survey of lattice path applications.
(This example continued in Section~\ref{sec:contrib} below.)
\end{example}

Algorithms to compute asymptotic expansions of P-recursive sequences have a long history, including methods that have been implemented in computer algebra systems
\cite[\emph{e.g.},][]{%
Tournier1987,
Chabaud2002,
Salvy1989,
FlajoletSalvyZimmermann1991,
Zeilberger2008,
kauers2011mathematica}.
These algorithms take as input an encoding of $(f_n)$, typically a recursion it satisfies or an equation satisfied by its generating function, and return explicitly defined functions $A(n)$ and $B(n)$ such that $f_n = A(n) + O(B(n))$ as $n\rightarrow\infty$. Although this (usually) allows the user to determine dominant asymptotic behaviour of $f_n$, in some applications such a `big-O' error is not sufficient and an explicit bound on the difference between the sequence and its dominant asymptotic behavior is required.
The original application motivating the line of work presented here is the \emph{complete positivity problem}: given a linear recurrence relation and enough initial terms, determine if all terms in the corresponding P-recursive sequence are positive. 

\begin{example}[Canham Model for Biomembranes]\label{ex:bio}
The \emph{Canham model} is an influential energy-minimization model to predict the structure of biomembranes such as blood cells. Yu and Chen~\cite{YuChen2022} reduced the question of proving solution uniqueness of the model for genus one surfaces to a proof of positivity for all terms%
\footnote{A more direct proof has since been given by Bostan and Yurkevich~\cite{BostanYurkevich2022}.}
in the P-recursive sequence $(d_n)$ defined by the initial terms
\begin{align*}
(d_0, d_1, d_2, d_3, d_4, d_5, d_6) = (&72, 1932, 31248,{790101}/{2},{17208645}/{4}, \\
 &{338898609}/{8}, {1551478257}/{4})
\end{align*}
and an explicit order seven linear recurrence relation
$\sum_{i=0}^{7} r_i(n) d_{n+i} = 0$,
with $r_i(n) \in \Z[n]$ defined in \cite[Appendix]{MelczerMezzarobba2022}. 
Although standard algorithms show $d_n = A(n) + O(B(n))$ for a simple positive function $A(n)$ and asymptotically smaller $B(n)$, the unknown constant in the big-O term does not allow one to conclude positivity of $d_n$ for \emph{all} $n$. Instead, Melczer and Mezzarobba~\cite{MelczerMezzarobba2022} found an explicit constant $C>0$ such that $|d_n-A(n)| \leq C|B(n)|$ for all $n$. Once $C$ is known, it is possible to determine $n_0 \in \N$ sufficiently large such that $d_n \geq A(n) - C|B(n)| > 0$ for all $n >n_0$, then computationally check positivity of the finitely many remaining terms $d_0,\dots,d_{n_0}$ by computing them. 
(This example continued in Section~\ref{sec:contrib} below.)
\end{example}

Example~\ref{ex:bio} is an instance of a \textit{complete positivity problem} \textbf{(CPP)},
which asks whether all terms in a sequence encoded by a P-recursion 
and a set of initial values are positive. Such positivity problems are, in
general, extremely difficult: as noted by Ouaknine and 
Worrell~\cite{OuaknineWorrell2014}, 
proving decidability of the complete positivity problem even for \emph{C-finite sequences} (satisfying P-recursions with constant
coefficients) of \emph{order 6} would already entail major breakthroughs in the Diophantine approximation of transcendental numbers.
Furthermore, the famous \emph{Skolem problem}, which asks whether it is decidable to take a real P-recursive sequence $\{f_n\}_n$ and determine whether there exists $n\in\N$ with $f_n=0$, can be reduced to \textbf{CPP} since P-recursiveness of the real sequence $\{f_n\}_n$ implies P-recursiveness of the non-negative sequence~$\{f_n^2\}_n$. The Skolem problem (for C-finite sequences) 
has essentially been open for almost one hundred years.

Despite the difficulty in the general case, complete positivity can be determined
in many cases. Indeed, given a C-finite recurrence and initial values that
determine a unique solution $u_n$, one can compute a representation of~$u_n$ as a finite
linear combination of terms of the form~$\varphi^n n^k$ where $\varphi$~is an
algebraic number and $k$~is an integer,
and, if one of these terms dominates all others for large~$n$, explicitly find
an~$n_0$ starting from which $u_n$~has the same sign as that term. The difficulty
arises when two terms have exponential growths $\varphi_1, \varphi_2$ such that $\abs{\varphi_1} = \abs{\varphi_2}$ and $\varphi_1/\varphi_2$ is not a root of unity --- in this case it can be
hard (perhaps undecidable) to see how the sums of the algebraic powers 
involved interact as $n$ grows. Thankfully, it is a ``meta-principle''
that rational generating functions arising from combinatorial problems
always seem to lie the special class of \emph{$\mathbb{N}$-rational functions},
meaning (among other things) that their positivity can be decided
(see, for instance,~\cite{Bousquet-Melou2006}). 

Further difficulties can arise for P-recursive sequences, including
some that do occur for combinatorial generating functions. Unlike the
rational generating functions of C-finite sequences, which can be explicitly encoded
and manipulated, the D-finite generating functions of P-recursive sequences
are typically manipulated \emph{implicitly} through the differential
equations they satisfy. As discussed below, the singularities of a generating
function dictate the asymptotic behaviour of its coefficient sequence, and it
can be very hard (perhaps undecidable in general) to separate singularities
of a D-finite generating function from the singularities of other solutions to a
differential equation it satisfies (see Remark~\ref{rem:cert_sing} below).
To work around this difficulty, our algorithms allow the user to pass
as input a set of points which are known not to be singularities of 
a D-finite function of interest.

\subsection{Contributions}
\label{sec:contrib}

This paper generalizes the ad-hoc approach of~\cite{MelczerMezzarobba2022} for the Canham problem and extends it to a wide class of P-recursive sequences.

It is well-known to specialists that many of the methods used to obtain
asymptotic expansions of \mbox{P-recursive} sequences can, in principle, provide
computable error bounds.
However, the error bounds are far from explicit---in the best case, they are
expressed in terms of maxima of potentially complicated analytic functions over
certain domains, and buried in the proofs of results of a more asymptotic nature.

The main contribution of the present work is a \emph{practical} algorithm that,
taking as input any P-recursive sequence whose associated differential
equation has regular dominant singular points, computes an asymptotic
approximation of that sequence  along with an \emph{explicit} error bound.
In favorable cases, the approximation is a truncated asymptotic expansion (to
arbitrary order) of the sequence.

We provide a complete implementation of our algorithm in the SageMath
computer algebra system.
Before going further, we illustrate our methods, using this implementation, on
the two examples introduced above.

\begin{excont2}
Returning to the Canham model sequence $(d_n)$, our algorithm
provides a brief and almost automatic proof of its positivity.
Setting the parameters $n_0 = 50$ and $r_0 = 2$ in Algorithm~\ref{algo:main} below,
we produce the expansion
\begin{align*}
    d_n \in (3 - 2\sqrt{2})^{-n} n^3 \cdot \biggl( &
    \left[8.072 \pm 2.30 \cdot 10^{-4}\right] \log n
    + \left[1.371 \pm 8.94 \cdot 10^{-4}\right] \\    
    + & \left[50.51 \pm 1.98 \cdot 10^{-3}\right] \frac{\log n}{n}
    + \left[29.70 \pm 4.42 \cdot 10^{-3}\right] \frac{1}{n} \\
    + & \left[\pm 3.11 \cdot 10^3\right] \frac{\log n}{n^2} \biggr),
\end{align*}
where $[a\pm b]$ denotes a real constant%
\footnote{Technically our algorithm returns real and imaginary components of the coefficients appearing in this asymptotic expansion, however when it is clear a priori that the coefficients must be real we omit the imaginary parts (which are certified to be zero to several decimal places) from the outputs displayed in the text.}
certified to be in the interval $[a-b,a+b]$.
The first four constants that appear are the leading coefficients in an asymptotic 
expansion of $f_n$ and can be computed to any
desired precision~$\varepsilon$
(here they are displayed to approximately three decimal places).
The final term, with a large constant, is an explicit error bound.
It can easily be seen directly from this bound that $d_n > 0$ for all $n \geq 50$. Thus,
by computing all $d_n$ for $n < 50$ and verifying their positivity we conclude that
$(d_n)$ contains only positive terms.
\end{excont2}

\begin{excont1} \label{excont1}
Let $f(z)$ be the lattice path generating function introduced above.
Setting the parameters $r_0 = 3$ and $n_0=0$, our algorithm produces the
rigorous approximation
\begin{align*}
    f_n \in 4^{n} n^{-1} \cdot \biggl( &
    \left[1.27 \pm 3.44 \cdot 10^{-3}\right]
    + \left[-1.91 \pm 3.76 \cdot 10^{-3}\right] \frac{1}{n}   \\
    + & \left(\left[4.93 \pm 8.13 \cdot 10^{-3}\right] 
    + (-1)^n \left[0.318 \pm 6.18 \cdot 10^{-4}\right]\right) \frac{1}{n^2} \\
    + & \left[\pm 1.51 \cdot 10^3\right] \frac{\log^2 n}{n^3} \biggr)
\end{align*}
and determines that it is valid for $n \geq 9$.
Despite the oscillatory behaviour of the third term, 
the leading constants can still be computed to any desired accuracy.
By increasing the expansion order to~$r_0=6$, we obtain for instance that the
probability that a random walk in~$\Z^2$ starting at the origin has not left
the quarter plane after a million steps is equal to
$[1.27323763487919 \cdot 10^{-6} \pm 7 \cdot 10^{-21}]$.
In less than 4~seconds on a modern laptop we can compute a 20-term
approximation of~$f_n$ with constants certified to more than 1000~decimal places.
\end{excont1}

Our approach is based on the \emph{singularity analysis} method as developed by Flajolet and
Odlyzko~\cite{flajolet1990singularity,FlajoletSedgewick2009}.
Roughly speaking, in singularity analysis one estimates the $n$th term of a
convergent power series~$f(z)$ by representing it as a complex Cauchy integral.
The path of the Cauchy integral is deformed into a union of small circular arcs around
singularities of $f(z)$ closest to the origin (\emph{dominant} singularities), arcs of a
big circle containing all dominant singularities in its interior, 
and straight lines connecting these circles.
One computes asymptotic expansions of the analytic continuation of~$f$ in
the neighborhood of the dominant singularities, then integrates the 
leading terms of these local expansions over the small
arcs close to the dominant singularities to compute dominant asymptotic terms.
Finally, one proves that the contributions of both the remainders of the
local expansions and the rest of the integration path are negligible for large
enough~$n$.
Our algorithm follows the same pattern, except that we show how to compute
explicit bounds on all the asymptotically negligible terms.
To do so, we leverage the representation of the series~$f(z)$ as a D-finite
function and make use of techniques for the rigorous numerical solution of
differential equations.

We limit ourselves here to D-finite functions because of their link to
P-recursive sequences, their ubiquity in combinatorics, and because 
this restriction causes all pieces of the analysis to fit together in a way that
provides a complete, implementable algorithm.
However, much of what we discuss actually applies to more general situations.
In particular, the procedure for computing asymptotic expansions with error
bounds of coefficients of algebro-logarithmic monomials
\[ (1-z)^{-\alpha} \log^k(1/(1-z)) \]
has independent interest.
Our main algorithm can also, in principle, be adapted to
other classes of differential equations with analytic coefficients, the main
requirements being that coefficients are given as computable series
expansions with suitable convergence bounds, and that singular
points, in addition to being regular, can be computed exactly (as elements of an effective field).

\subsection{Related work}

Singularity analysis, and more generally complex-analytic methods for
asymptotic enumeration, are a classical topic and the subject of abundant
literature. Good entry points to the theory include
the now-classic book by Flajolet and Sedgewick~\cite{FlajoletSedgewick2009}
and a survey of Odlyzko~\cite{Odlyzko1995}. The focus in such combinatorial
contexts is typically on obtaining asymptotic equivalents, or
asymptotic expansions with big-O error terms, as opposed to sharp error bounds
with explicit constants as one finds for example in work on special
functions~\cite[\emph{e.g.},][]{Olver1997}. More specifically, our
algorithm is based on the well-established method of
\emph{singularity analysis of linear differential equations}~\cite[Section VII.9.1]{FlajoletSedgewick2009}, with our main tools coming from or inspired by
works of Jungen~\cite{jungen1931series},
Flajolet--Puech~\cite{FlajoletPuech1986}, and
Flajolet--Odlyzko~\cite{flajolet1990singularity}.

Automating such asymptotic techniques using symbolic computation is not a new idea.
Already in the late 1980s, Salvy and collaborators~\cite{Salvy1989,FlajoletSalvyZimmermann1991}
developed and implemented algorithms to compute asymptotic expansions of the
coefficients of wide classes of generating series, typically given by
closed-form formulas.
In the case of series defined by functional equations, such as algebraic or
linear differential equations, one can still often determine the dominant
singularities and singular behaviour of a \emph{general} solution, but it is
typically difficult to pinpoint that of the particular solution of interest
using purely symbolic methods.
As noted by Flajolet and Puech~\cite[Section~5.4]{FlajoletPuech1986},
however, one can use numerical methods for this purpose.
The case of algebraic equations is detailed in Chabaud's
thesis~\cite[Part~III]{Chabaud2002}, while Julliot~\cite{Julliot2020} recently
developed a prototype implementation of the D-finite case.

Singularity analysis is not the only available method to determine the
asymptotic behaviour of P-recursive sequences.
In fact, as early as 1930 Birkhoff~\cite{Birkhoff1930} described the
construction of formal solutions of general linear difference equations with
formal asymptotic series as coefficients.
Implementations of this method~\cite{Zeilberger2008,kauers2011mathematica,KauersJaroschekJohansson2015}
are widely used as a heuristic way of obtaining asymptotic expansions of
P-recursive sequences.
As linking formal asymptotic solutions to actual solutions is already
difficult~\cite{birkhoff1933analytic,Immink1991,vdPutSinger1997},
it seems challenging (though probably possible in principle) to extract
error bounds from this general approach.

In the special case of a linear difference equation with \emph{polynomial}
coefficients, one can also produce a basis of analytic solutions with
well-understood asymptotic behaviour using Mellin transforms of solutions of
the associated differential equation, a technique going back to Pincherle in the
late nineteenth century
\cite[\emph{e.g.,}][]{Pincherle1892,Duval1983,immink1999relation}.
Barkatou~\cite{Barkatou1989} implemented an algorithm based on this idea for
computing a basis of asymptotic expansions of solutions of a given difference
equation.
Van der Hoeven~\cite{vanderhoeven:hal-03291372}, in concurrent work with ours,
uses a construction of this type to extend the approach
of~\cite{MelczerMezzarobba2022} to asymptotic expansions and positivity, and
study the computational complexity of evaluating P-recursive sequence to
moderate precision at large values of the index.
As the focus of his paper is on feasibility and complexity theorems rather than
detailed algorithms, the overlap with the present work is minimal.

Various algorithms based on sufficient conditions have been proposed to partially deal with the complete positivity problem, such as \cite{GerholdKauers2005,KauersPillwein2010,Pillwein2013}.
More recently there has been progress that focuses on special \mbox{P-recursive} sequences, notably low-order C-finite sequences~\cite{OuaknineWorrell2014}, as well as second order P-recursive sequences~\cite{KKLLMOWW2021,NeumannOuaknineWorrell2021}.

An earlier version of the present work also appeared in the first author's Masters
thesis~\cite{Dong2021}.

\subsection{Outline}

The remainder of this article starts in Section \ref{sec:prelim}, where 
we recall some definitions and facts related to
differential equations with polynomial coefficients, their numerical solution,
and complex ball (interval) arithmetic.
Sections \ref{sec:bounds}~to~\ref{sec:global} are dedicated to our algorithm
and its proof of correctness.
In Section~\ref{sec:bounds}, we decompose the Cauchy integral representing a
term of the sequence into several contributions that are then bounded
separately, and give an overview of the main algorithm.
Section~\ref{sec:explicit_part} deals with the contribution to the final bound
of the initial terms of local expansions at individual singularities.
In particular, we describe a subroutine for computing approximations with error
bounds of the coefficient of~$z^n$ in a series of the ``standard scale''
$(1-z/\rho)^{-\alpha} \log((1-z/\rho)^{-1})^k$
in which the local expansions are written.
Then, in Section~\ref{sec:local_error}, we explain how to bound the
contribution of the remainders of these local expansions.
In Section~\ref{sec:global}, we do the same for the error term associated to
the portion of the integration path away from the singularities, and conclude
the proof of correctness of the algorithm.
Finally, in Section~\ref{sec:implementation}, we discuss our implementation in
more detail with the support of additional examples.

\section{Preliminaries}\label{sec:prelim}

\subsection{Differential operators and singular points}

Let $\K \subset \CC$ be a number field and define the linear differential
operator
\begin{equation}\label{eq:ldop}
    \diffop = p_{q}(z) \frac{d^q}{d z^q} + \cdots + p_1(z) \frac{d}{d z} + p_0(z)
\end{equation}
with polynomial coefficients $p_j(z) \in \K[z]$ where $p_q \neq 0$.
We call~$q$ the \emph{order} of $\diffop$ and the linear differential equation
\[\diffop f := p_{q}(z) \frac{d^qf}{d z^q}(z) + \cdots + p_1(z) \frac{df}{d z}(z) + p_0(z)f(z) = 0  \]
the \emph{(homogeneous) D-finite equation defined by $\diffop$}. We also say that a series
or complex function $f(z)$ is a \emph{solution} of $\diffop$, or is \emph{annihilated} by $\diffop$, if it satisfies $\diffop f = 0$ (where defined, in the case of a complex function). Linear differential operators of the form~\eqref{eq:ldop} can be encoded as elements of the \emph{Weyl algebra} $\diffops{\K}$, which contains skew polynomials over $\K[z]$ in the indeterminate~$d/dz$, subject to the relation
$d/dz \cdot z = z \cdot d/dz + 1$.

A D-finite power series~$f(z)$ can be represented by an annihilating differential
operator~$\diffop$ and enough initial conditions to specify it as a unique solution.
If $f(z) = \sum_{n \geq 0}f_nz^n$ satisfies $\diffop f=0$ then extracting the coefficient of the general term $z^n$ in
\[ 0 = p_{q}(z) \frac{d^q}{d z^q}\sum_{n \geq 0}f_n z^n + \cdots+ p_0(z)\sum_{n \geq 0}f_n z^n \]
yields a linear recurrence relation
\begin{equation}
0 = c_r(n)f_{n+r} + \cdots + c_1(n)f_{n+1} + c_0(n)
\label{eq:Prec}
\end{equation}
for the sequence $(f_n)$ with polynomial coefficients $c_j(n) \in \K[n]$. If $M$ is the largest natural number root of $c_r(n)$, or zero if $c_r(n)$ has no natural number roots, then any solution $(f_n)_{n \in \N}$ to~\eqref{eq:Prec} is uniquely determined by the values of $f_0,f_1,\dots,f_{r+M}$. 

A point $\rho \in \CC$ is called
a \emph{singular point} of~$\diffop$ if $p_q(\rho) = 0$,
and an \emph{ordinary point} otherwise.
Cauchy's existence theorem for differential equations~\cite[Ch.~1.2]{Poole1960} implies that if $\rho$ is an ordinary point of~$\diffop$ then there exist $q$~linearly independent solutions $f_1(z),\dots,f_q(z)$ to $\diffop$ analytic in the disk $\{z : |z-\rho| < |\rho' - \rho|\}$, where $\rho'$ is the closest singular point of $\diffop$ to $\rho$. If some solution of~$\diffop$ is analytic on an open set with~$\rho$ on its boundary, but singular at $\rho$, then $\rho$ is a singular point of~$\diffop$.

\begin{definition}\label{def:singular_point}
A singular point~$\rho$ of~$\diffop$ is called
\begin{itemize}
    \item an \emph{apparent singularity} if there exist $q$ complex solutions $f_1(z),\dots,f_q(z)$ for $\diffop$ which are analytic at $z=\rho$ and linearly independent over $\CC$,
    \item a \emph{regular singular point} if, for all $j = 0,1,\ldots, q-1$ the order of the pole of $\frac{p_j(z)}{p_q(z)}$ at $z = \rho$ is at most $q - j$.
\end{itemize}
An apparent singularity is also a regular singular point.
We say that $\rho$ is \emph{at most} a regular singular point if it is an ordinary point or regular singular point, and let $\Xi = \{\rho : p_t(\rho)=0\}$ denote the set of all singular points.
\end{definition}

\begin{remark}
Suppose $f(z)$ is a solution of a D-finite equation with power series solution $f(z) = \sum_{n \geq 0}f_nz^n$ at the origin, where $(f_n)$ is an integer sequence such that $|f_n| \leq C^n$ for some $C>0$. A series of deep results due to André, the Chudnovsky brothers, and Katz combine to show that the differential operator corresponding to any minimal order D-finite equation satisfied by $f(z)$ has only regular singular points. Thus, it is very common to encounter D-finite equations with regular singularities in combinatorial applications. See Melczer~\cite[Section 2.4]{Melczer2021} for more details.
\end{remark}

To study the analytic solutions of D-finite equations near their singularities we need to allow for more general series expansions than usual power series.
Here, and everywhere in this article, the complex logarithm and the complex
power function of a non-integer exponent take their principal value, which is
defined on $\CC \setminus \{0\}$, analytic on $\CC \setminus \R_{\leq 0}$, and
continuous as $z$ approaches the negative real line from above in the complex plane.
It will be convenient to express the local behavior of solutions at nonzero
singular points in terms of the functions
$z \mapsto (1-z/\rho)^{\nu}$ with $\nu \in \CC$
and
$z \mapsto \log \bigl((1-z/\rho)^{-1}\bigr)$,
both analytic on the complex plane with the ray from~$\rho$ to~$\infty$ removed.

\begin{proposition}[Solution basis at regular singular points]
\label{prop:singFuschs}
At any regular singular point $\rho$ of $\diffop$ the D-finite equation defined by $\diffop$ admits a $\CC$-basis of solutions $(y_{\rho,1}(z), \ldots, y_{\rho,q}(z))$ with
\begin{equation}\label{eq:sol_basis}
  \everymath{\displaystyle}
  \begin{cases}
    y_{0,j}(z) = z^{\nu_{j}}
      \sum_{i = 0}^{\infty} \sum_{k = 0}^{\kappa_{j}}
      d_{i,k, j} z^{i} \log^{k}(z), \\
    y_{\rho,j}(z) = (1-z/\rho)^{\nu_{j}}
      \sum_{i = 0}^{\infty}
      \sum_{k = 0}^{\kappa_{j}}
      d_{i,k, j} (1-z/\rho)^{i} \log^{k}\left(\frac{1}{1-z/\rho}\right),
      & \rho \neq 0,
  \end{cases}
\end{equation}
where
\begin{itemize}
  \item the $\nu_j$ are algebraic, the $\kappa_j$ are nonnegative integers, and the $d_{i,k, j}$ are elements of \, $\K(\rho, \nu_j)$,
  \item for each~$j$, at least one of the $d_{0, k, j}$ is nonzero, and
  $\nu_{j_1} = \nu_{j_2}$ implies
  $\min\{k : d_{0, k, j_1} \neq 0\} \neq \min\{k : d_{0, k, j_2} \neq 0\}$
  (the basis is in ``triangular form'').
\end{itemize}
These series solutions converge on
$B(\rho, r) \setminus [\rho, (1+r) \rho)$,
the open disk around $\rho$ with radius
$r = \sup \{R : B(\rho, R) \cap \Xi = \emptyset\}$,
slit along a radius.
\end{proposition}

See Poole~\cite[Chapter 5]{Poole1960} for a proof of Proposition~\ref{prop:singFuschs}. As noted above, if $z = \rho$ is an ordinary point then there is a basis of power series solutions also satisfying~\eqref{eq:sol_basis}, with $\nu_j \in \N$ and $\kappa_j = 0$ for $j = 1, \ldots, t$.

We fix once and for all a basis~$(y_{\rho,j})_j$ of the form~\eqref{eq:sol_basis} for each regular singular point~$\rho$.
In what follows, we will write $y_j(z)$ instead of $y_{\rho,j}(z)$ when $\rho$ is clearly indicated by the context.

\subsection{Numerical approximations}

Our method for deriving explicit asymptotic bounds on the coefficient sequence $f_n$ involves numerically approximating certain constants associated with the solutions of $\diffop$. First, we introduce a class of numbers that suffices to represent all values that we will encounter.

\begin{definition}[Holonomic constants]
A number $\alpha \in \CC$ is said to be a
\emph{(regular) singular holonomic constant}~\cite{FlajoletVallee2000,Hoeven2021},
or
\emph{D-finite number}\footnote{Some of the cited definitions allow one to take the limit of~$f$ at a regular singular point.
The equivalence of these definitions follows
from~\cite[Corollary~B.4]{Hoeven2021}.
Note that the statement of that corollary contains a typo:
the inclusion should read
$\mathbb K^{\mathrm{rhol}} \subseteq \mathbb K^{\mathrm{hola}}$.
}%
~\cite{HuangKauers2018}, if $\alpha = f(1)$ for some solution
$f \in \overline\Q[[z]]$ to a linear differential operator $\diffop$ having at most a regular singular point at the origin and no other singular point in the closed unit disk.
\end{definition}

We write $\rhol$ for the class of regular singular holonomic constants.
Computationally, an element $\alpha\in\rhol$ is represented by an operator
$\diffop \in \diffops{\Qbar}$ and enough initial conditions at the origin to
define a unique solution $f$ of $\diffop$ with $\alpha=f(1)$.

Let $\rholgamma$ denote the $\overline\Q$-algebra generated by
\begin{equation} \label{eq:rholgamma}
\rhol \cup \left\{\Gamma(z)^{-1} : z\in \overline{\Q}\right\} \cup \left\{\gamma^{(j)}(z) : z\in \overline{\Q}, j \in \N\right\},
\end{equation}
where $\Gamma(z)$ denotes the gamma function and 
$\gamma^{(j)}(z) = \frac{d^{j+1}}{dz^{j+1}} \log \Gamma(z)$ denotes the polygamma function of order $j$.
An element of $\rholgamma$ is represented as a polynomial expression in the elements of the generating set~\eqref{eq:rholgamma}.

The following proposition shows that it is possible to efficiently approximate elements of $\rholgamma$ rigorously to any prescribed accuracy. 

\begin{proposition}[Computing in $\rholgamma$]
Let~$E$ be a fixed polynomial expression in the elements of the set~\eqref{eq:rholgamma}.
As $n$ tends to infinity, one can compute the value of~$E$ to precision
$\varepsilon = 2^{-n}$ in time 
\[ O(M(n \log n) \log n), \]
where $M(n)$ is the time needed to multiply two $n$-digit numbers.
\end{proposition}

\begin{proof}
The value of an element in $\rhol$ can be computed with an error bounded
by~$\varepsilon$ in $O(M(n \log n) \log n)$ operations by solving the
corresponding differential equation using a Taylor method where
partial sums of Taylor series are computed by binary
splitting~\cite{Hoeven2001}.
The Gamma function can be evaluated to precision $\varepsilon$ at any
fixed~$z \in \overline\Q$ in time $O(M(n \log n) \log n)$ using the
strategy mentioned in~\cite[§1, last paragraph]{Brent1976}.
Combining this method with fast evaluation of the logarithm and automatic differentiation allows one to evaluate $\psi^{(j)}(z)$ for any fixed $j \in \N$ 
and~$z \in \overline\Q$  in the same complexity.
(See also Karatsuba~\cite{Karatsuba1998} for a more detailed discussion
of the evaluation of
the Hurwitz zeta function
$\zeta(j, z) = (-1)^{j} \psi^{(j-1)}(z)/(j-1)!$
based on similar ideas, and an explicit complexity result.)
Thus, the value of any element of the set~\eqref{eq:rholgamma} can be computed in the claimed time complexity.
For a fixed expression~$E$, one only needs a bounded number of extra `guard digits' to recover the value of~$E$ to precision~$\varepsilon$,
so one can increase the precision of intermediate computations until the final result is accurate without the asymptotic running time being affected.
Adding and multiplying together the intermediate results to recover the value of~$E$ takes $O(M(n))$ operations.
\end{proof}

\begin{remark}
Although it is possible to evaluate elements of $\rhol$ to arbitrary
precision, there is no known zero test for its elements, nor, \emph{a fortiori}, for elements of
$\rholgamma$.
This can be problematic when certifying singularities of solutions to D-finite equations (see Remark~\ref{rem:cert_sing} below) and subsequently when proving positivity. Fortunately, in most applications
all constants that one needs to test turn out to be nonzero,
and thus arbitrary precision evaluation suffices to prove that this is the case.
\end{remark}

In order to manipulate and perform arithmetic operations on bounds,
we use complex ball arithmetic.

\begin{definition}\label{def:CBF}
Let $\CBF$ denote the set of \emph{complex rectangles} of the form
\[
\I(a,b,\varepsilon_a, \varepsilon_b)
= [a-\varepsilon_a, a+\varepsilon_a] + \i \cdot [b-\varepsilon_b, b+\varepsilon_b],
\]
where $a, b \in \R$ and $\varepsilon_a, \varepsilon_b \in \R_{\geq 0}$ are real
numbers.
We call the elements of~$\CBF$ \emph{balls}.
The ball~$\I(a,b,\varepsilon_a, \varepsilon_b)$ is \emph{exact}
if $\varepsilon_a = \varepsilon_b$ = 0.
Addition, subtraction and multiplication of balls are performed
following the standard rule of interval arithmetic:
$I \ast I'$ is a rectangle containing
$\{z \ast z' : z \in I, z' \in I'\}$,
where $\ast$ denotes addition, subtraction or multiplication.
We do not require $I \ast I'$ to be the smallest such rectangle, but
do assume that, for fixed $a, b, 'a, b'$, the diameter of
\[
  \I(a,b,\varepsilon_a, \varepsilon_b) \ast
  \I(a',b',\varepsilon_{a'}, \varepsilon_{b'})
\]
tends to zero when
$\varepsilon_a, \varepsilon_b, \varepsilon_{a'}, \varepsilon_{b'}$
all tend to zero.
In particular, $I \ast I'$ is exact if both $I$~and~$I'$ are,
and we often identify $\I(\Rel(z), \operatorname{Im}(z), 0,0)$ with the complex
number~$z$.
\end{definition}

For theoretical purposes, it is convenient in this definition to allow
$a$~and~$b$ to be arbitrary real numbers.
In practice, however,
$a, b, \varepsilon_a,$ and $\varepsilon_b$
need to be machine-representable numbers,
so that not all complex numbers can be represented by exact balls.
When we say that a ball manipulated by an algorithm is exact,
this may not hold true in an actual implementation using finite-precision
arithmetic.
Nevertheless, the balls we manipulate in this article are defined by computable
real numbers (in fact, elements of~$\rholgamma$), so that quantities
represented by ``exact'' balls can at least be approximated to arbitrary
precision.

The set $\CBF$ is not a ring, despite its resemblance to one,
yet the usual ring operations are well-defined over $\CBF$. In fact, we can
define polynomial ``rings'' $\CBF[\mathbf{x}]$
(where $\mathbf{x}$ denotes a vector of variables):
the elements of $\CBF[\mathbf{x}]$ have the form
\[
  \I\left(\sum_{\alpha} a_{\alpha} \mathbf{x}^{\alpha}, \sum_{\alpha} b_{\alpha} \mathbf{x}^{\alpha}, \sum_{\alpha} \varepsilon_{a, \alpha} \mathbf{x}^{\alpha}, \sum_{\alpha} \varepsilon_{b, \alpha} \mathbf{x}^{\alpha}\right) \triangleq \sum_{\alpha}\I(a_{\alpha}, b_{\alpha}, \varepsilon_{a, \alpha}, \varepsilon_{b, \alpha}) \cdot \mathbf{x}^{\alpha},
\]
and arithmetic operations are defined in the natural way.

We denote by $\B(f, r)$ the ball $\I(\Rel(f), \operatorname{Im}(f), r, r)$ for $f \in \CC$ and $r \in \mathbb{R}$, or for $f \in \CC[\mathbf{x}]$ and $r \in \mathbb{R}[\mathbf{x}]$.
Here, if $f \in \CC[\mathbf{x}]$ then $\Rel(f)$ means, by abuse of notation, $\sum_{\alpha}\Rel(f_{\alpha})\mathbf{x}^{\alpha}$ where $f = \sum_{\alpha}f_{\alpha}\mathbf{x}^{\alpha}$, and similarly for $\operatorname{Im}(f)$.
We denote
$\B(\I(a, b, \varepsilon_a, \varepsilon_b), r) = \I(a, b, \varepsilon_a + r, \varepsilon_b + r)$,
both for $a, b \in \CBF$ and for $a, b \in \CBF[\mathbf{x}]$.

\subsection{Connection coefficients}
\label{sec:connection}

If $f(z)$ is a solution of $\diffop$ which is analytic at $\rho$, and $\Gamma$ is any piecewise linear curve in $\CC$ starting at~$\rho$ and avoiding the elements of $\Xi$, then it is possible to define a unique analytic continuation of~$f$ along $\Gamma$.
If $U$ is any simply connected open set in $\CC$ which contains $\rho$ and does not contain any element of $\Xi$, then this process defines a (unique) analytic continuation of $f$ to $U$. The following domain will be particularly useful for our considerations.

\begin{definition}[Multi-slit disk $\Delta$]
\label{def:slitdisk}
The \emph{multi-slit disk} $\Delta$ defined by $\diffop$ is the set
\[
\Delta = \Delta_{\diffop}
= \CC - \left(\bigcup_{\zeta \in \Xi\setminus\{0\}} \left\{z : \frac{z}{\zeta} \in [1, \infty) \right\} \cup \left( \Xi \cap \{0\}\right) \right)
\]
obtained by removing the rays from the non-zero singular points of $\diffop$ to infinity from $\CC$, and removing zero if it is a singular point of $\diffop$.
\end{definition}

In our applications, we start with knowledge of the series expansion of a function at the origin and want to determine properties of this function near other points in the complex plane.
The first step to doing this is expressing the function of interest in terms of the basis of solutions provided by Proposition~\ref{prop:singFuschs}. 

\begin{proposition}[Computing coefficients]
\label{prop:coeff_zero}
Suppose that $\rho$ is at most a regular singular point of $\diffop$, let $\bigl(y_1(z),\dots,y_q(z)\bigr)$ be a basis of solutions at $z=\rho$ for $\diffop$ as provided by Proposition~\ref{prop:singFuschs}, and let $f(z) = \sum_{n \geq 0}f_n(1-z/\rho)^n$ be a power series solution of $\diffop$.
Given enough coefficients $f_n$ we can compute $c_{1}, \ldots, c_{q} \in \K(f_0, \ldots, f_M)$ (for large enough $M$) such that 
\begin{equation}
f(z) = c_{1} y_{1}(z) + \cdots + c_{q} y_{q}(z)
\label{eq:fandc}
\end{equation}
in a neighbourhood of $\rho$ in $\Delta_{\diffop}$.
\end{proposition}

Given a representation of $f(z)$ in terms of a basis of solutions near one point, we represent $f(z)$ near another point using analytic continuation.

\begin{definition}[Connection matrix]
\label{def:connection_matrix}
Let $\rho_1$ and $\rho_2$ be at most regular singular points for $\diffop$, and let $\by_{\rho_1} = (y_{\rho_1,1}, \ldots, y_{\rho_1,q})$ and $\by_{\rho_2} = (y_{\rho_2,1}, \ldots, y_{\rho_2,q})$ be bases of solutions for $\diffop$ in terms of series at $\rho_1$ and $\rho_2$ of the form provided by Proposition~\ref{prop:singFuschs}.
The \emph{connection matrix} for $\diffop$ defined by
\begin{itemize}
    \item a polygonal path $\rho_1 \to \rho_2$ linking $\rho_1$ to $\rho_2$ that lies, except for its endpoints, entirely within the domain $\Delta_{\diffop}$, and
    \item the bases of solutions $\by_{\rho_1}$ and $\by_{\rho_2}$
\end{itemize}
is the change of basis matrix $\bC_{\rho_1 \rightarrow \rho_2}(\by_{\rho_1}, \by_{\rho_2}) \in \operatorname{GL}_q(\CC)$ satisfying
\[
\overline{\by_{\rho_1}} = \by_{\rho_2} \bC_{\rho_1 \rightarrow \rho_2}(\by_{\rho_1}, \by_{\rho_2}),
\]
on some interval $[\zeta, \rho_2)$ contained in the last edge of the path,
where the entries of $\overline{\by_{\rho_1}}$ consist of the analytic continuation of the entries of $\by_{\rho_1}$ along the path%
\footnote{%
When $\rho_1$ or $\rho_2$ is the origin, the general convention that the complex logarithm is continuous ``from above'' on its branch cut applies.
For example, the path $0 \to -1$ defines the same connection matrix as an infinitesimally close path lying entirely in the upper half-plane.
In general, this matrix differs from the one defined by a similar path in the lower half-plane.}.

For convenience we usually assume that the bases of solutions $\by_{\rho_1}$ and $\by_{\rho_2}$ are fixed, and write $\bC_{\rho_1 \rightarrow \rho_2}(\by_{\rho_1},\by_{\rho_2}) = \bC_{\rho_1 \rightarrow \rho_2}$.
\end{definition}

It follows from the closure properties of D-finite functions that 
the entries of $\mathbf{C}_{\rho_1 \rightarrow \rho_2}$ are in $\rhol$
(see \cite[Proposition B.3]{Hoeven2021} and ~\cite[Theorem~19]{HuangKauers2018} for details).
Efficient numerical methods are available~\cite{Hoeven2001,Mezzarobba2011,mezzarobba2019truncation} that compute $\bC_{\rho_1 \to \rho_2}$ to arbitrary precision and with rigorous error bounds given the operator~$\diffop$ and the path $\rho_1 \to \rho_2$.
If $\bc_{0} = (c_{0,1}, \ldots, c_{0,q})^T$ and $\bc_{\rho} = (c_{\rho,1}, \ldots, c_{\rho,q})^T$ denote the vectors of coefficients appearing when representing $f$ in~\eqref{eq:fandc} using bases of solutions $\by_0$ and $\by_{\rho}$ then
\begin{equation}\label{eq:connection_coeff}
\bc_{\rho} = \bC_{0 \rightarrow \rho} \bc_0.
\end{equation}

If $f(z)$ is analytic at $z=0$ then, because of the triangular form of the solution basis that we have chosen,  
all non-zero entries $c_{0,j}$ of $\bc_0$ correspond to solutions $y_{0, j}$ that are analytic at $z=0$, thus analytic in $\Delta_{\diffop}^0 = \Delta_{\diffop} \cup \{0\}$.
Since $\Delta_{\diffop}^0$ is simply connected, the analytic continuation of $y_{0, j}$ to $\rho$ within the domain $\Delta_{\diffop}^0$ does not depend on the choice of the path.
The expression $\bC_{0 \rightarrow \rho} \bc_0$ in \eqref{eq:connection_coeff} is therefore solely dependent on $\rho$, and independent of the choice of path $0 \rightarrow \rho$ in $\bC_{0 \rightarrow \rho}$.

Using Proposition \ref{prop:coeff_zero} and~\eqref{eq:connection_coeff} we can thus compute, for any $\rho$ and to any precision, constants $c_{\rho,1}, \ldots, c_{\rho,q}$ such that
\[
f(z) = c_{\rho,1} y_{\rho, 1}(z) + \ldots + c_{\rho,q} y_{\rho, q}(z).
\]
In what follows, we will write $c_j$ instead of $c_{\rho,j}$ when $\rho$ is clearly indicated by the context.

\begin{remark}[Certifying singularity]\label{rem:cert_sing}
Given $\diffop$ and initial terms of $f$ at $z=0$ it can be difficult to verify when $f$ is analytic at a singular point $z = \rho \in \Xi$. This is mainly due to a lack of an exact zero test for elements of~$\rhol$: when $f$ is singular this can be detected by computing connection coefficients to a sufficiently high accuracy, however when $f$ is not singular we can show only that it is a linear combination of basis elements whose singular terms have coefficients that are zero to any given number of decimal places. However, there are a few cases where non-singularity can be rigorously verified:
\begin{itemize}
    \item Apparent singularities, that is, singular points where all solutions are analytic.
    This is equivalent to the existence of a full basis of formal power series solutions, so testing this reduces to linear algebra in~$\K[\rho]$.
    \item When $f$ is algebraic, meaning there exists a bivariate polynomial $P(z,y) \in \K[x,y]$ such that $P(z,f(z))=0$, the singularities of $f(z)$ can be determined using this algebraic relation. See Chabaud~\cite{Chabaud2002} or Flajolet and Sedgewick~\cite[Chapter VII. 7]{FlajoletSedgewick2009} for details.
\end{itemize}

There are other, more sporadic approaches that can be used to rule out singularities. For instance, if one is studying the generating function $g(z) = \sum_{n \geq 0}f_nz^n$ of a combinatorial class and combinatorial arguments can be used to bound the exponential growth of $f_n$ then this may give a meaningful bound on the singularities of $g$. When $g(z)$ can be represented as the \emph{diagonal} of a multivariate rational function, techniques from the field of analytic combinatorics in several variables can be used to determine asymptotic information about $f_n$, and thus analytic information about $g(z)$; see Melczer~\cite[Chapters 5 and 9]{Melczer2021}.
\end{remark}

\section{Algorithm Overview}\label{sec:bounds}

We now fix a power series solution
\[
f(z) = f_0 + f_1z + f_2 z^2 + \ldots
\]
to the differential operator $\diffop$ from~\eqref{eq:ldop} that converges in a disk around the origin. Our goal is to take $f(z)$, encoded by $\diffop$ and enough initial terms to uniquely specify it among the solutions of $\diffop$, and express $f_n$ as a linear combination of explicit functions of~$n$ plus an
explicit error term (see Theorem \ref{thm:mainBound} for the exact form). 
In many circumstances, the first terms of the sum correspond to a truncated asymptotic
expansion of~$f_n$ as $n \to \infty$, and the expression can be made
arbitrarily tight relative to~$|f_n|$.
However, as discussed in the introduction above, interference of terms with the same order of magnitude can make the bounds either too weak or too complicated to provide any useful information.

The singularities of $f(z)$ that are closest to the origin are called the \emph{dominant singularities of~$f$}.
Because the Cauchy existence theorem implies that these singularities are singular points of $\diffop$, we let $\domsing \subseteq \Xi$ denote the set of dominant singularities of $f$.
Our results assume the following.

\theoremstyle{acmdefinition}
\newtheorem*{runningassumptions}{Global Assumptions}
\begin{runningassumptions}
We assume that $f$ has at least one singularity on $\CC \setminus \{0\}$ and that all dominant singularities of $f$ are regular singular points of~$\diffop$.
\end{runningassumptions}

The first of these assumptions is a matter of convenience;
starting from an arbitrary D-finite series, one can usually reduce to it by
means of a formal Borel or Laplace transform.
The second assumption, however, is more restrictive but is crucial for our method.

The goal of this section is to describe an algorithm that computes the aforementioned bound for $f_n$.
Algorithm \ref{algo:main} provides an overview of the construction.
In addition to the function~$f$ and an expansion order~$r_0$,
Algorithm \ref{algo:main} takes as input a lower bound~$n_0$ for the desired
validity range $\{ n \geq N_0 \}$ of the output.
The reason for distinguishing between $n_0$ and $N_0$ is that the constant
factor in the result depends heavily on $N_0$.
If one only needs a bound valid for large~$n$, one can input a large $n_0$
in order to make the bound tighter.
The algorithm also accepts a subset~$\anasing$ of the singular points
of~$\diffop$ where $f$~is known to be analytic.
One can always assume that $\anasing$ contains the apparent singularities
of~$\diffop$, as these are effectively computable.

\begin{algorithm}[p]
\caption{Asymptotic expansion with error bound}
\label{algo:main}
\begin{description}[itemindent=0pt-\widthof{(1)},leftmargin=2em]
\item[Input]
    A linear differential operator $\diffop \in \diffops{\K}$
    and sufficiently many initial coefficients $f_0, f_1, \ldots$ to uniquely determine a series $f(z) = f_0 + f_1 z + \ldots$ satisfying $\diffop f = 0$.
    Integers $n_0$ and $r_0$.
    A set of points $\anasing$ where $f$ is known to be analytic.
\item[Output]
    An integer $N_0$ and an estimation for $f_n$ with explicit error bounds of the form \eqref{eq:estimate_form_main}, valid for all $n \geq N_0$.
\end{description}
\begin{enumerate}[leftmargin=2em]
    \item Compute the set~$\Xi$ of singular points of $\diffop$.
    Let $\domsingapprox$ be the set of elements of minimal modulus of
    $\Xi \setminus \anasing$.
    \item\label{step:sol_bas} Compute the structure of a local solution basis of $\diffop$ at $z = 0$ as in Proposition~\ref{prop:singFuschs} and the coordinate vector $\bc_0 = (c_{0, 1}, \ldots, c_{0, q})$ of~$f$ in this basis as in Proposition \ref{prop:coeff_zero}.
    \item\label{step:R0N0} Compute $R_0$ according to Equation~\eqref{eq:R0_choice} and $N_0$ according to Equations~\eqref{eq:def_n1}~and~\eqref{eq:def_n0}.
    \item[$\blacktriangleright$] \emph{Contribution of each singularity}
    \item For each $\rho \in \domsingapprox$:
    \begin{enumerate}
        \item[$\blacktriangleright$] \emph{Singular expansions}
        \item \label{step:localbasis}
        Compute the structure of a local solution basis $(y_1(z), \ldots, y_q(z))$ of $\diffop$ at $z = \rho$ as in Proposition~\ref{prop:singFuschs}.
        Let $\nu_j$, $\kappa_j$ be the corresponding parameters appearing in Equation~\eqref{eq:sol_basis}.
        \item \label{step:connection} Let $\bC_{0 \rightarrow \rho}$ be the connection matrix for~$\diffop$ along a straight path $0 \rightarrow \rho$, making small detours to avoid other singular points if needed (see Definition \ref{def:connection_matrix}).
        Compute local coordinates $\bc_{\rho}$ at $z = \rho$ for~$f$ using $\bc_{\rho} = \bC_{0 \rightarrow \rho} \bc_0$.
        \item \label{step:foreachbasiselement}
        For $j = 1, \dots, q$:
        \begin{enumerate}
            \item \label{step:localparam}
            Let $r_j$ be as in Definition~\ref{def:landg}
            and $s = N_0/\max_j(\abs{\nu_j} + r_j + 1)$.
            \item[$\blacktriangleright$]
            \emph{Explicit part (Section~\ref{sec:explicit_part})}
            \item \label{step:localbasisexplicit}
            Compute the coefficients $d_{i,k,j}$ of~$\ell_j$ in the decomposition
            $y_j(z) = \ell_j(z) + g_j(z)$
            given by Equations~\eqref{eq:lgexpansion}--\eqref{eq:local_error_term}.
            Save these coefficients for later reuse in Step~\eqref{step:bound_I}.
            \item \label{step:coeffasy}
            For $i = 0, \dots, r_j-1$, call Algorithm~\ref{algo:coeffasy}
            with $\alpha:=-\nu_j-i$, $K:=\kappa_j$, $r:=r_j-i$, $n_0:=N_0$
            to compute bounds of the form
            $n^{\alpha-1} e_k(n^{-1}, \log n)$
            on $[z^n](1-z)^{\nu_j+i} \log^k\left(\frac{1}{1-z}\right)$
            for $0 \leq k \leq K$.
            \item \label{step:contribution_explicit}
            Deduce a bound for $[z^n] \ell_{\rho, j}(z)$ using Equation~\eqref{eq:sum_l_decomp}.
            \item[$\blacktriangleright$] \emph{Local error term (Section~\ref{sec:local_error})}
            \item \label{step:analyticcase}
            If $\nu_j + r_j \in \Z_{\geq 0}$ and $\kappa_j = 0$, continue with
            the next loop iteration
            (setting to zero the bounds otherwise computed by the next two steps).
            \item \label{step:tailbound}
            Compute bounds $b_0, \ldots, b_{\kappa_j}$ satisfying Equation~\eqref{eq:def_b_i}
            using \cite[Algorithm~6.11]{mezzarobba2019truncation}
            and Equation~\eqref{eq:choice_b_i}.
            Define the polynomial $B$ as in Equation~\eqref{eq:def_B}.
            \item \label{step:contribution_local_error}
            Deduce a bound for
            $\frac{1}{2\pi \i} \int_{\mathcal{S}_{\rho}(n)} z^{-n-1} g_{\rho}(z) \,dz$
            using Proposition~\ref{prop:bound_S},
            and a bound for
            $\lim_{\varphi \to 0} \frac{1}{2\pi \i} \int_{\mathcal{L}_{\rho}(n)} z^{-n-1} g_{\rho}(z) \,dz$
            using Corollary~\ref{cor:bound_L}
            (with $n_0 := N_0$, $r := r_j$ and $\nu := \nu_j$, and  $B$~from the previous step).
        \end{enumerate}
    \end{enumerate}
    \item[$\blacktriangleright$] \emph{Global error term and final bound (Section~\ref{sec:global})}
    \item \label{step:bound_I} Bound values on the circle $\{\abs{z} = R_0\}$ of
    $f(z) - \sum_{\rho \in \domsingapprox} \ell_{\rho}(z)$
    as described in Section~\ref{subsection:bound_I}.
    \item \label{step:combine}
    Combine the bounds produced in steps
    \ref{step:contribution_explicit},
    \ref{step:contribution_local_error}, and \ref{step:bound_I}
    into a bound for~$f_n$
    according to \eqref{eq:landg}, \eqref{eq:decomposition},
    and~\eqref{eq:bound_Ic}.
    Simplify it as described in Section~\ref{subsection:assemble}.
    \item Return $N_0$ and the simplified sum.
\end{enumerate}%
\end{algorithm}

\newcommand{\statemaintheorem}[1]{%
Let $M = \min_{\rho \in \Xi \setminus \anasing} |\rho|$ be the minimal modulus of a
singular point of~$\diffop$, excluding any point where~$f$ is known to be
analytic.
Algorithm~\ref{algo:main} (page~\pageref{algo:main}) computes an integer $N_0 \geq n_0$ and an estimate
\begin{equation}\label{#1}
    f_n = \sum_{\rho \in \domsing} \rho^{-n} n^{\gamma} \sum_{i = 0}^{m_{\rho}} \sum_{k = 0}^{\kappa} a_{\rho,i,k} \frac{\log^{k} n}{n^{\gamma_{\rho,i}}} + R(n),
    \quad
    |R(n)| \leq A \, M^{-n} n^{\Rel(\gamma)} \frac{\log^{\kappa} n}{n^{r_0}}
\end{equation}
for all $n \geq N_0$, where $0 \leq \Rel(\gamma_{\rho,0}) \leq \Rel(\gamma_{\rho,1}) \leq \ldots \leq \Rel(\gamma_{\rho,m_{\rho}}) < r_0$.
In this estimate, $\gamma$~and the~$\gamma_{\rho,i}$ are algebraic numbers,
the~$a_{\rho,i,l}$ belong to the class $\rholgamma$ and can be computed to arbitrary precision with rigorous error bounds,
and $A$~is a nonnegative real number.
}

\begin{theorem}
\label{thm:mainBound}
\statemaintheorem{eq:estimate_form_main}
\end{theorem}

In practice, the real and imaginary parts of all terms in~\eqref{eq:estimate_form_main} are represented by balls that contain them.
We can make the radii of balls corresponding to the asymptotic series coefficients for $f_n$ arbitrarily small by increasing numeric precision, however the radii of balls for the constants in the error bound for $|R(n)|$ are limited by the behaviour of the sequence.

Our derivation of these bounds will not need to deal explicitly with the
non-dominant singularities of $f$. However, we will need to take into account
the singular points of $\diffop$ that are closer to the origin than all
singularities of $f$.

\begin{definition}
A \emph{dominant singular point} of $\diffop$ for~$f$ is any singular point of $\diffop$ whose modulus is at most the modulus of the dominant singularities of $f$. We write $\esing$ for the set of dominant singular points of $\diffop$, and note that $\domsing \subseteq \esing$.
\end{definition}

Theorem~\ref{thm:mainBound} is proved in Section~\ref{sec:proofconclusion},
based on the bounds developed in Sections~\ref{sec:explicit_part}
to~\ref{sec:global} below.
An overview of the proof is as follows. First,
in the remainder of the present section, we express~$f_n$ as a Cauchy integral
and divide it into an ``explicit'' contribution of the leading asymptotic behavior
of~$f(z)$ at each of its dominant singularities, a ``local'' error term associated
to each of these singularities, and a ``global'' error term.
In Section~\ref{sec:explicit_part}, we give a subroutine for computing bounds on
the contribution of the explicit part.
Section~\ref{sec:local_error} deals with the local error terms.
Finally, in Section~\ref{sec:global}, we discuss how to combine these bounds,
incorporate the global error term, and trim them down to an expression
of the form given in Theorem~\ref{thm:mainBound}.

\subsection{\texorpdfstring{Expressing $f_n$ as a Cauchy integral}{Expressing fn as a Cauchy integral}}
\label{subsection:express_integral}

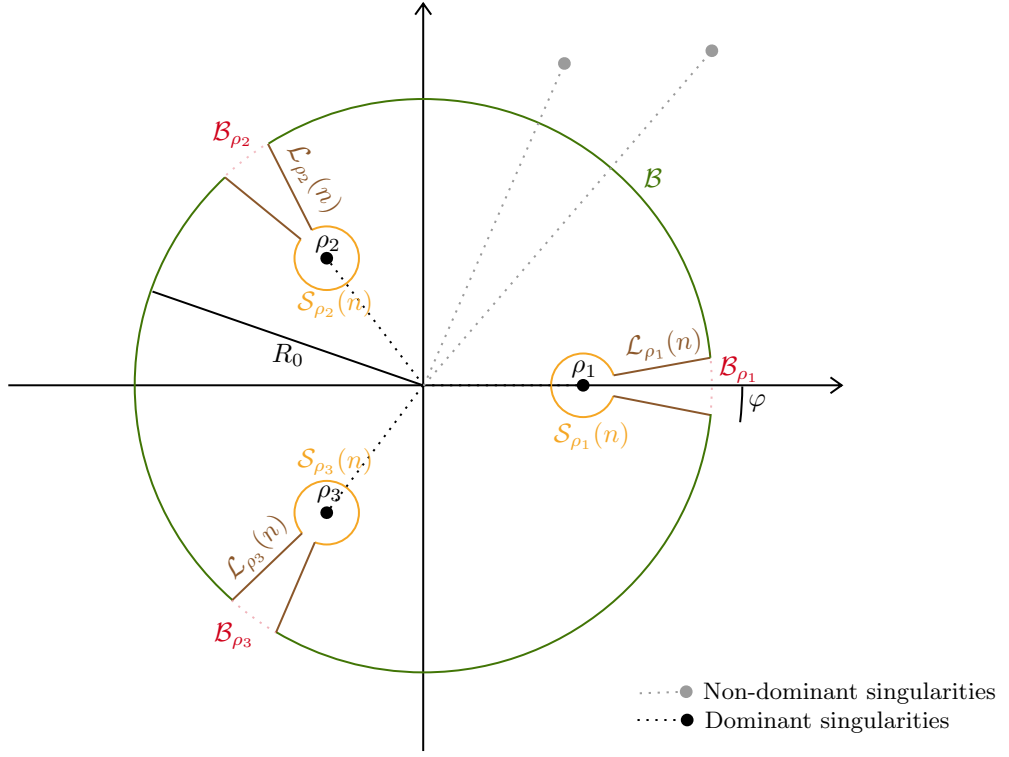
\begin{figure}
    \centering
    \begin{tikzpicture}[x=.6pt,y=.6pt,yscale=-1,xscale=1]

    \draw  (91.42,280.05) -- (611.42,280.05)(350.3,40.01) -- (350.3,510.01) (604.42,275.05) -- (611.42,280.05) -- (604.42,285.05) (345.3,47.01) -- (350.3,40.01) -- (355.3,47.01)  ;
    \draw  [draw opacity=0] (468.92,286.83) .. controls (466.14,294.54) and (458.77,300.05) .. (450.1,300.05) .. controls (439.05,300.05) and (430.1,291.1) .. (430.1,280.05) .. controls (430.1,269) and (439.05,260.05) .. (450.1,260.05) .. controls (458.96,260.05) and (466.47,265.81) .. (469.1,273.79) -- (450.1,280.05) -- cycle ; \draw  [color={rgb, 255:red, 245; green, 166; blue, 35 }  ,draw opacity=1 ] (468.92,286.83) .. controls (466.14,294.54) and (458.77,300.05) .. (450.1,300.05) .. controls (439.05,300.05) and (430.1,291.1) .. (430.1,280.05) .. controls (430.1,269) and (439.05,260.05) .. (450.1,260.05) .. controls (458.96,260.05) and (466.47,265.81) .. (469.1,273.79) ;
    \draw  [draw opacity=0] (280.79,182.42) .. controls (283.57,180.97) and (286.74,180.15) .. (290.1,180.15) .. controls (301.2,180.15) and (310.2,189.1) .. (310.2,200.15) .. controls (310.2,211.2) and (301.2,220.15) .. (290.1,220.15) .. controls (279,220.15) and (270,211.2) .. (270,200.15) .. controls (270,195.68) and (271.48,191.54) .. (273.97,188.21) -- (290.1,200.15) -- cycle ; \draw  [color={rgb, 255:red, 245; green, 166; blue, 35 }  ,draw opacity=1 ] (280.79,182.42) .. controls (283.57,180.97) and (286.74,180.15) .. (290.1,180.15) .. controls (301.2,180.15) and (310.2,189.1) .. (310.2,200.15) .. controls (310.2,211.2) and (301.2,220.15) .. (290.1,220.15) .. controls (279,220.15) and (270,211.2) .. (270,200.15) .. controls (270,195.68) and (271.48,191.54) .. (273.97,188.21) ;
    \draw  [draw opacity=0] (274.74,373.05) .. controls (271.78,369.57) and (270,365.07) .. (270,360.15) .. controls (270,349.1) and (279,340.15) .. (290.1,340.15) .. controls (301.2,340.15) and (310.2,349.1) .. (310.2,360.15) .. controls (310.2,371.2) and (301.2,380.15) .. (290.1,380.15) .. controls (287.49,380.15) and (284.99,379.65) .. (282.7,378.75) -- (290.1,360.15) -- cycle ; \draw  [color={rgb, 255:red, 245; green, 166; blue, 35 }  ,draw opacity=1 ] (274.74,373.05) .. controls (271.78,369.57) and (270,365.07) .. (270,360.15) .. controls (270,349.1) and (279,340.15) .. (290.1,340.15) .. controls (301.2,340.15) and (310.2,349.1) .. (310.2,360.15) .. controls (310.2,371.2) and (301.2,380.15) .. (290.1,380.15) .. controls (287.49,380.15) and (284.99,379.65) .. (282.7,378.75) ;
    \draw [color={rgb, 255:red, 139; green, 87; blue, 42 }  ,draw opacity=1 ]   (529.48,262.73) -- (469.1,273.79) ;
    \draw [color={rgb, 255:red, 139; green, 87; blue, 42 }  ,draw opacity=1 ]   (529.73,298.81) -- (468.92,286.83) ;
    \draw  [draw opacity=0] (253.81,128.07) .. controls (281.7,110.33) and (314.8,100.05) .. (350.3,100.05) .. controls (443.87,100.05) and (520.76,171.45) .. (529.48,262.73) -- (350.3,280.05) -- cycle ; \draw  [color={rgb, 255:red, 65; green, 117; blue, 5 }  ,draw opacity=1 ] (253.81,128.07) .. controls (281.7,110.33) and (314.8,100.05) .. (350.3,100.05) .. controls (443.87,100.05) and (520.76,171.45) .. (529.48,262.73) ;
    \draw  [draw opacity=0] (529.72,298.14) .. controls (520.75,389.32) and (443.98,460.55) .. (350.6,460.55) .. controls (316.97,460.55) and (285.5,451.31) .. (258.57,435.23) -- (350.6,280.23) -- cycle ; \draw  [color={rgb, 255:red, 65; green, 117; blue, 5 }  ,draw opacity=1 ] (529.72,298.14) .. controls (520.75,389.32) and (443.98,460.55) .. (350.6,460.55) .. controls (316.97,460.55) and (285.5,451.31) .. (258.57,435.23) ;
    \draw  [draw opacity=0] (231.24,415.05) .. controls (193.87,382.07) and (170.3,333.81) .. (170.3,280.05) .. controls (170.3,228.55) and (191.93,182.1) .. (226.6,149.29) -- (350.3,280.05) -- cycle ; \draw  [color={rgb, 255:red, 65; green, 117; blue, 5 }  ,draw opacity=1 ] (231.24,415.05) .. controls (193.87,382.07) and (170.3,333.81) .. (170.3,280.05) .. controls (170.3,228.55) and (191.93,182.1) .. (226.6,149.29) ;
    \draw [use as bounding box,color={rgb, 255:red, 139; green, 87; blue, 42 }  ,draw opacity=1 ]   (226.6,149.29) -- (273.97,188.21) ;
    \draw [color={rgb, 255:red, 139; green, 87; blue, 42 }  ,draw opacity=1 ]   (253.42,128.27) -- (280.79,182.42) ;
    \draw [color={rgb, 255:red, 139; green, 87; blue, 42 }  ,draw opacity=1 ]   (274.74,373.05) -- (231.24,415.05) ;
    \draw [color={rgb, 255:red, 139; green, 87; blue, 42 }  ,draw opacity=1 ]   (282.7,378.75) -- (258.57,435.23) ;
    \draw    (181.4,221.1) -- (350.6,280.23) ;
    \draw  [dash pattern={on 0.84pt off 2.51pt}]  (350.6,280.23) -- (290.1,200.15) ;
    \draw [shift={(290.1,200.15)}, rotate = 232.93] [color={rgb, 255:red, 0; green, 0; blue, 0 }  ][fill={rgb, 255:red, 0; green, 0; blue, 0 }  ][line width=0.75]      (0, 0) circle [x radius= 3.35, y radius= 3.35]   ;
    \draw  [dash pattern={on 0.84pt off 2.51pt}]  (350.3,280.05) -- (450.1,280.05) ;
    \draw [shift={(450.1,280.05)}, rotate = 0] [color={rgb, 255:red, 0; green, 0; blue, 0 }  ][fill={rgb, 255:red, 0; green, 0; blue, 0 }  ][line width=0.75]      (0, 0) circle [x radius= 3.35, y radius= 3.35]   ;
    \draw  [dash pattern={on 0.84pt off 2.51pt}]  (350.2,280.15) -- (290.1,360.15) ;
    \draw [shift={(290.1,360.15)}, rotate = 126.92] [color={rgb, 255:red, 0; green, 0; blue, 0 }  ][fill={rgb, 255:red, 0; green, 0; blue, 0 }  ][line width=0.75]      (0, 0) circle [x radius= 3.35, y radius= 3.35]   ;
    \draw [color={rgb, 255:red, 155; green, 155; blue, 155 }  ,draw opacity=1 ] [dash pattern={on 0.84pt off 2.51pt}]  (350.67,280.39) -- (438.4,77.7) ;
    \draw [shift={(438.4,77.7)}, rotate = 293.41] [color={rgb, 255:red, 155; green, 155; blue, 155 }  ,draw opacity=1 ][fill={rgb, 255:red, 155; green, 155; blue, 155 }  ,fill opacity=1 ][line width=0.75]      (0, 0) circle [x radius= 3.35, y radius= 3.35]   ;
    \draw [color={rgb, 255:red, 155; green, 155; blue, 155 }  ,draw opacity=1 ] [dash pattern={on 0.84pt off 2.51pt}]  (350.6,280.23) -- (530.4,69.7) ;
    \draw [shift={(530.4,69.7)}, rotate = 310.5] [color={rgb, 255:red, 155; green, 155; blue, 155 }  ,draw opacity=1 ][fill={rgb, 255:red, 155; green, 155; blue, 155 }  ,fill opacity=1 ][line width=0.75]      (0, 0) circle [x radius= 3.35, y radius= 3.35]   ;
    \draw  [draw opacity=0][dash pattern={on 0.84pt off 2.51pt}] (227.4,147.83) .. controls (235.49,140.35) and (244.27,133.6) .. (253.62,127.69) -- (349.9,280.4) -- cycle ; \draw  [color={rgb, 255:red, 208; green, 2; blue, 27 }  ,draw opacity=0.26 ][dash pattern={on 0.84pt off 2.51pt}] (227.4,147.83) .. controls (235.49,140.35) and (244.27,133.6) .. (253.62,127.69) ;
    \draw  [draw opacity=0][dash pattern={on 0.84pt off 2.51pt}] (258.87,435.71) .. controls (248.7,429.7) and (239.17,422.71) .. (230.41,414.89) -- (350.6,280.23) -- cycle ; \draw  [color={rgb, 255:red, 208; green, 2; blue, 27 }  ,draw opacity=0.3 ][dash pattern={on 0.84pt off 2.51pt}] (258.87,435.71) .. controls (248.7,429.7) and (239.17,422.71) .. (230.41,414.89) ;
    \draw  [draw opacity=0][dash pattern={on 0.84pt off 2.51pt}] (529.18,259.26) .. controls (529.99,266.19) and (530.4,273.24) .. (530.4,280.39) .. controls (530.4,286.52) and (530.1,292.58) .. (529.5,298.55) -- (350.67,280.39) -- cycle ; \draw  [color={rgb, 255:red, 208; green, 2; blue, 27 }  ,draw opacity=0.29 ][dash pattern={on 0.84pt off 2.51pt}] (529.18,259.26) .. controls (529.99,266.19) and (530.4,273.24) .. (530.4,280.39) .. controls (530.4,286.52) and (530.1,292.58) .. (529.5,298.55) ;
    \draw  [draw opacity=0] (549.4,280.81) .. controls (549.36,288.4) and (548.9,295.9) .. (548.04,303.27) -- (349.4,279.9) -- cycle ; \draw   (549.4,280.81) .. controls (549.36,288.4) and (548.9,295.9) .. (548.04,303.27) ;
    \draw  [dash pattern={on 0.84pt off 2.51pt}]  (484,490.5) -- (515,490.5) ;
    \draw [shift={(515,490.5)}, rotate = 0] [color={rgb, 255:red, 0; green, 0; blue, 0 }  ][fill={rgb, 255:red, 0; green, 0; blue, 0 }  ][line width=0.75]      (0, 0) circle [x radius= 3.35, y radius= 3.35]   ;
    \draw [color={rgb, 255:red, 155; green, 155; blue, 155 }  ,draw opacity=1 ] [dash pattern={on 0.84pt off 2.51pt}]  (485.4,473.7) -- (514.4,472.7) ;
    \draw [shift={(514.4,472.7)}, rotate = 358.03] [color={rgb, 255:red, 155; green, 155; blue, 155 }  ,draw opacity=1 ][fill={rgb, 255:red, 155; green, 155; blue, 155 }  ,fill opacity=1 ][line width=0.75]      (0, 0) circle [x radius= 3.35, y radius= 3.35]   ;

    \draw (254,252) node [anchor=north west][inner sep=0.75pt]   [align=left] {$\displaystyle R_{0}$};
    \draw (441.06,262.94) node [anchor=north west][inner sep=0.75pt]  [rotate=-0.34] [align=left] {$\rho _{1}$};
    \draw (281.06,183.94) node [anchor=north west][inner sep=0.75pt]  [rotate=-0.34] [align=left] {$\rho _{2}$};
    \draw (282.06,340.94) node [anchor=north west][inner sep=0.75pt]  [rotate=-0.34] [align=left] {$\rho _{3}$};
    \draw (430,301.4) node [anchor=north west][inner sep=0.75pt]  [color={rgb, 255:red, 245; green, 166; blue, 35 }  ,opacity=1 ]  {$\mathcal{S}_{\rho _{1}}( n)$};
    \draw (473.13,246.97) node [anchor=north west][inner sep=0.75pt]  [color={rgb, 255:red, 139; green, 87; blue, 42 }  ,opacity=1 ,rotate=-353.99]  {$\mathcal{L}_{\rho _{1}}( n)$};
    \draw (486,142.4) node [anchor=north west][inner sep=0.75pt]  [color={rgb, 255:red, 65; green, 117; blue, 5 }  ,opacity=1 ]  {$\mathcal{B}$};
    \draw (269.94,218.55) node [anchor=north west][inner sep=0.75pt]  [color={rgb, 255:red, 245; green, 166; blue, 35 }  ,opacity=1 ,rotate=-359.64]  {$\mathcal{S}_{\rho _{2}}( n)$};
    \draw (270.07,317.22) node [anchor=north west][inner sep=0.75pt]  [color={rgb, 255:red, 245; green, 166; blue, 35 }  ,opacity=1 ,rotate=-0.43]  {$\mathcal{S}_{\rho _{3}}( n)$};
    \draw (276.38,122.14) node [anchor=north west][inner sep=0.75pt]  [color={rgb, 255:red, 139; green, 87; blue, 42 }  ,opacity=1 ,rotate=-57.38]  {$\mathcal{L}_{\rho _{2}}( n)$};
    \draw (220.79,394.36) node [anchor=north west][inner sep=0.75pt]  [color={rgb, 255:red, 139; green, 87; blue, 42 }  ,opacity=1 ,rotate=-312.89]  {$\mathcal{L}_{\rho _{3}}( n)$};
    \draw (217,112.4) node [anchor=north west][inner sep=0.75pt]  [color={rgb, 255:red, 208; green, 2; blue, 27 }  ,opacity=1 ]  {$\mathcal{B}_{\rho _{2}}$};
    \draw (533,260.4) node [anchor=north west][inner sep=0.75pt]  [color={rgb, 255:red, 208; green, 2; blue, 27 }  ,opacity=1 ]  {$\mathcal{B}_{\rho _{1}}$};
    \draw (218,427.4) node [anchor=north west][inner sep=0.75pt]  [color={rgb, 255:red, 208; green, 2; blue, 27 }  ,opacity=1 ]  {$\mathcal{B}_{\rho _{3}}$};
    \draw (551.4,284.21) node [anchor=north west][inner sep=0.75pt]    {$\varphi $};
    \draw (523.2,465) node [anchor=north west][inner sep=0.75pt]  [font=\small] [align=left] {Non-dominant singularities};
    \draw (524.2,483) node [anchor=north west][inner sep=0.75pt]  [font=\small] [align=left] {Dominant singularities};

    \end{tikzpicture}
    \caption{The domain of integration $\mP(n)$.}
    \label{fig:int_path}
\end{figure}

Recall that, as a solution of~$\diffop$ analytic at the origin, the function~$f(z)$ can be analytically continued to the domain~$\Delta_\diffop$.
Following the transfer method of
Flajolet and Odlyzko~\cite{flajolet1990singularity}, we express $f_n$ as a Cauchy integral
\[ f_n = \frac{1}{2\pi \i}\int_{\mP(n)} \frac{f(z)}{z^{n+1}} dz \]
over a (counter-clockwise oriented) simple closed path $\mP(n)$ depending on $n$, sitting completely inside $\Delta_{\diffop}$.
See Figure~\ref{fig:int_path} for an illustration.

In order to specify $\mP(n)$ we pick constants $R_0, R_1 \in \Q$ with
\[
  0 < R_1 < \min_{\substack{\rho_1 \in \domsing \\ \rho_2 \in \Xi \;\;}} |\rho_1 - \rho_2|,
  \qquad
  M < R_0 < M + R_1,
  \qquad
  R_0 < \min_{\rho \in \Xi \setminus \esing} |\rho|,
\]
where $M$ is the common modulus of the elements of $\domsing$.
We set
\begin{equation}\label{eq:def_n1}
    N_1 = \left\lceil \frac{2 \max_{\rho \in \domsing} |\rho|}{\min_{\rho_1 \in \domsing, \rho_2 \in \Xi} |\rho_1 - \rho_2|} \right\rceil
\end{equation}
if $|\Xi| \geq 2$ and $N_1 = 0$ otherwise.
When $n > N_1$ and $\varphi > 0$ is sufficiently small, we define a path $\mP(n)$ consisting of
\begin{itemize}
    \item Arcs $\mB$ of a \emph{big circle} of radius $R_0$ centered at $0$,
    \item Arcs $\mS_{\rho}(n)$ of \emph{small circles} of radius $|\rho|/n$ centered at each $\rho \in \domsing$, and
    \item Pairs of \emph{line segments} $\mL_{\rho}(n)$ connecting the arcs of the big and small circles, supported on lines passing through $\rho$ at angles $\pm \varphi$ with the ray from $0$ to $\rho$.
\end{itemize}
For $\rho\in\domsing$ we furthermore define $\mB_\rho$ to be the arc of the big circle between the ends of $\mL_{\rho}$. 

The choice of the constants $N_1$, $R_0$, and $R_1$ guarantees that $\mP(n)$ is a simple closed path and that its interior does not contain any singularity of~$f(z)$ (recall that $f$ is analytic at $z = 0$).
The value of~$\varphi$ does not play a large role since we will let $\varphi \to
0$ in what follows.

For concreteness in the algorithm, we again let $M$ be the modulus of the dominant singularities and take
\begin{equation} \label{eq:R0_choice}
R_0 \simeq \min \left\{ \frac M8 + \frac78 \min_{\rho \in \Xi \setminus \esing} |\rho|, \;\; M + \frac34\min_{\substack{\rho_1 \in \domsing\\ \rho_2 \in \Xi \;\;}} |\rho_1 - \rho_2|\right\},
\quad
R_1 \simeq R_0 - M,
\end{equation}
rounded down to the numeric working precision.
Empirically, these values provide good bounds, the reason being roughly as follows.
On one hand, for any $\rho \in \domsing$, the disk $\{|z - \rho| \leq R_1\}$ stays at sufficient distance from other dominant singular points, so that when estimating integrals on $\mathcal{S}_{\rho}(n)$ and $\mathcal{L}_{\rho}(n)$, a bound \eqref{eq:def_b_i} of reasonable size can be produced using the algorithm from \cite{mezzarobba2019truncation}.
On the other hand, the big circle $\mathcal{B}$ also stays at sufficient distance from non-dominant singularities, so that when estimating integrals on $\mathcal{B}$ a bound of reasonable size can also be produced.

\subsection{Decomposing the integral}
\label{sec:decompose}

Now fix a dominant singularity $\rho \in \domsing$, let $\bigl(y_{1}(z), \ldots, y_{q}(z)\bigr)$ be a basis of solutions at $\rho$ specified as in Proposition~\ref{prop:singFuschs}, and let $c_1,\dots,c_q$ be the constants such that
\[
f(z) = c_1 y_{1}(z) + \cdots + c_q y_{q}(z).
\]
Since $\rho$ is a regular singular point we may select each function $y_i(z)$ to have an expansion of the form~\eqref{eq:sol_basis}, and split off any number $r_j \in \N_{>0}$ of leading terms to obtain a decomposition
\begin{equation}
\label{eq:lgexpansion}
y_j(z) = \ell_{\rho,j}(z) + g_{\rho,j}(z)
\end{equation}
with
\begin{align}
  \label{eq:explicit_part}
  \ell_{\rho, j}(z)
    &= (1 - z/\rho)^{\nu_j}
       \sum_{i=0}^{r_j - 1}
        (1 - z/\rho)^{i} \sum_{k = 0}^{\kappa} d_{i, k, j} \log ^k \left(\frac{1}{1 - z/\rho}\right) \\
\intertext{and}
  \label{eq:local_error_term}
  g_{\rho, j}(z) &= (1 - z/\rho)^{\nu_j + r_j} \sum_{k = 0}^{\kappa} h_{j,k}(z) \log ^k\left(\frac{1}{1 - z/\rho}\right).
\end{align}
In this expression, all constants $\nu_j,d_{i,k,j}$ are computable and the functions $h_{j,k}(z)$ are analytic at $z = \rho$. We view the finite series $\ell_{\rho,j}$ as the (explicit) \emph{leading terms} of this expansion, and $g_{\rho,j}$ as the (implicitly defined) \emph{error term} obtained when approximating $y_j$ by these leading terms.

\begin{remark} \label{rk:expomodZ}
By replacing elements $c_j y_{j}(z)$ having the same exponent $\nu_i$ modulo $\Z$ with the sum of these elements, we can suppose all $\nu_j$ different from each other modulo $\Z$.
This is not done in Algorithm~\ref{algo:main} to keep the pseudocode simple, but accelerates the algorithm by decreasing the number of elements in the solution basis considered.
\end{remark}

\begin{definition} \label{def:landg}
Let $\lambda = \min_{j \in \{ 1, \dots, q\}} \Rel \nu_j$ be the minimal real part of the leading exponents of the basis elements, excluding basis elements that are analytic at~$\rho$.
If $r_0>0$ is the desired expansion order in the statement of Theorem~\ref{thm:mainBound} we take
$r_j = r_0 + \left \lceil \lambda - \Rel \nu_j\right \rceil - 1$
in the expansions~\eqref{eq:lgexpansion} and let
\begin{equation}
\ell_\rho(z) = c_{1} \ell_{\rho, 1}(z) + \cdots + c_{q} \ell_{\rho, q}(z) \quad\text{ and }\quad
g_\rho(z) = c_{1} g_{\rho, 1}(z) + \cdots + c_{q} g_{\rho, q}(z).
\label{eq:landg}
\end{equation}
\end{definition}

\begin{proposition} \label{prop:decomposition}
  One has
  \begin{equation}
  \label{eq:decomposition}
  [z^{n}]f(z) =
    \sum_{\rho \in \domsing} \left(
      [z^n] \ell_{\rho}(z)
      + \frac1{2\pi \i} \int_{\mS_{\rho}(n) + \mL_{\rho}(n)} \frac{g_{\rho}(z)}{z^{n+1}} \, dz
    \right)
    + \mI_{\mB} \,,
  \end{equation}
  where
  \begin{equation} \label{eq:bound_Ic}
    \lim_{\varphi \to 0} \left|\mI_{\mB}\right|
    \leq \frac{1}{R_0^{n}} \cdot \max_{|z| = R_0} \Bigl|f(z) - \sum_{\rho \in \domsing} \ell_{\rho}(z)\Bigr|.
  \end{equation}
\end{proposition}

\begin{proof}
We begin by writing
\begin{align*}
    & (2\pi \i)[z^{n}]f(z) \\
    & = \int_{\mP(n)} \frac{f(z)}{z^{n+1}} \, dz \\
    &= \sum_{\rho \in \domsing} \int_{\mS_{\rho}(n)} \frac{f(z)}{z^{n+1}} \, dz + \sum_{\rho \in \domsing} \int_{\mL_{\rho}(n)} \frac{f(z)}{z^{n+1}} \, dz + \int_{\mB} \frac{f(z)}{z^{n+1}} \, dz \\
    & = \sum_{\rho \in \domsing} \int_{\mS_{\rho}(n)} \frac{\ell_{\rho}(z) + g_{\rho}(z)}{z^{n+1}} \, dz + \sum_{\rho \in \domsing} \int_{\mL_{\rho}(n)} \frac{\ell_{\rho}(z) + g_{\rho}(z)}{z^{n+1}} \, dz + \int_{\mB} \frac{f(z)}{z^{n+1}} \, dz.
\end{align*}
Because $\mP(n)$ can be contracted to a small circle $|z|=\varepsilon$ around the origin without crossing any singularities of the $\ell_\rho(z)$, for any $\rho \in \domsing$ we can express
\begin{multline*}
\int_{\mS_{\rho}(n)} \frac{\ell_{\rho}(z)}{z^{n+1}} \, dz + \int_{\mL_{\rho}(n)} \frac{\ell_{\rho}(z)}{z^{n+1}} \, dz
= \\
[z^n] \ell_\rho
- \int_{\mB} \frac{\ell_{\rho}(z)}{z^{n+1}} \, dz
- \sum_{\substack{\rho' \in \domsing \\ \rho' \neq \rho}} \left(
    \int_{\mS_{\rho'}(n)} \frac{\ell_{\rho}(z)}{z^{n+1}} \, dz
    + \int_{\mL_{\rho'}(n)} \frac{\ell_{\rho}(z)}{z^{n+1}} \, dz
\right)
\end{multline*}
where $[z^n]\ell_{\rho}(z) = \int_{|z| = \varepsilon} \frac{\ell_{\rho}(z)}{z^{n+1}} \, dz$ denotes the degree $n$ coefficient of the series expansion of~$\ell_{\rho}(z)$.
Thus, summing over all $\rho \in \domsing$ and performing some algebraic manipulation gives
\begin{equation*}
[z^{n}]f(z) = \frac1{2\pi \i} \sum_{\rho \in \domsing} \left(
[z^n] \ell_\rho
+ \int_{\mS_{\rho}(n) + \mL_{\rho}(n)} \frac{g_{\rho}(z)}{z^{n+1}} \, dz
\right) + \mI_{\mB},
\end{equation*}
where
\[ 2 \pi \i \mI_{\mB} = \int_{\mB} \frac{f(z) - \sum_{\rho \in \domsing} \ell_{\rho}(z)}{z^{n+1}} \, dz - \sum_{\rho \in \domsing} \int_{\mS_{\rho}(n) + \mL_{\rho}(n)} \frac{\sum_{\rho' \in \domsing \setminus \{\rho\}} \ell_{\rho'}(z)}{z^{n+1}} \, dz. \]
Because $\ell_{\rho'}(z)$ is holomorphic in an open region containing the area enclosed by $\mB_{\rho}, \mS_{\rho}(n),$ and $\mL_{\rho}(n)$ when $\rho' \neq \rho$, we have
\begin{align*}
    2 \pi \left|\mI_{\mB}\right| 
    & = \left|\int_{\mB} \frac{f(z) - \sum_{\rho \in \domsing} \ell_{\rho}(z)}{z^{n+1}} \, dz - \sum_{\rho \in \domsing} \int_{\mB_{\rho}} \frac{\sum_{\rho' \in \domsing \setminus \{\rho\}} \ell_{\rho'}(z)}{z^{n+1}} \, dz\right| \\
    & \leq \int_{|z| = R_0} \frac{\left|f(z) - \sum_{\rho \in \domsing} \ell_{\rho}(z)\right|}{|z|^{n+1}} \, dz
    + \left| \sum_{\rho \in \domsing} \int_{\mB_{\rho}} \frac{\sum_{\rho' \in \domsing \setminus \{\rho\}} \ell_{\rho'}(z)}{z^{n+1}} \, dz\right|.
\end{align*}
Since the second term tends to zero as $\varphi \to 0$, we obtain the bound~\eqref{eq:bound_Ic}.
\end{proof}

Our goal is now to bound the components of~\eqref{eq:decomposition}.
The next three sections deal respectively with~$[z^n] \ell_\rho$, followed by
$\int_{\mS_{\rho}(n) + \mL_{\rho}(n)} \frac{g_{\rho}(z)}{z^{n+1}}$,
and finally $\lim_{\varphi \to 0} \mI_{\mB}$.

\section{Contribution of a Singularity: The Explicit Part}
\label{sec:explicit_part}
\label{subsection:bound_l}

Expanding the explicit leading part $\ell_\rho(z)$ of the local expansion
of~$f(z)$ at a singular point~$\rho$ of~$\diffop$ using its definition in
Equations~\eqref{eq:local_error_term} and~\eqref{eq:landg} gives
$[z^n]\ell_{\rho}(z) = \sum_{j=1}^q c_j \, [z^n] \ell_{\rho, j}$
where
\begin{align}\label{eq:sum_l_decomp}
    [z^n]\ell_{\rho, j}(z)
    & =
    \sum_{i=0}^{r_j-1}\sum_{k=0}^{\kappa_j} d_{i,k,j} [z^n] (1 - z/\rho)^{\nu_j+i} \log^k\left(\frac{1}{1 - z/\rho}\right) \nonumber \\
    & = 
    \sum_{i=0}^{r_j-1}\sum_{k=0}^{\kappa_j} d_{i,k,j} \rho^{-n} [z^n](1-z)^{\nu_j + i} \log^k\left(\frac{1}{1-z}\right).
\end{align}
Thus, to obtain an asymptotic expansion (and corresponding error
bound) of~$[z^n] \ell_\rho$ it suffices to compute asymptotic expansions with error bounds for
the coefficient extraction
$[z^n](1-z)^{-\alpha} \log^k\left(\frac{1}{1-z}\right)$
for any
$k\in\N$ and $\alpha = -\nu_{j}, \ldots, -\nu_{j} - r_{j} + 1$.

More precisely, given a complex number~$\alpha$ and nonnegative integers
$k$~and~$r$, we aim to compute an expression $e(1/n, \log n)$ such that
\begin{equation}\label{eq:coeffexpansion}
  [z^n](1-z)^{-\alpha} \log^k\left(\frac{1}{1-z}\right) \in
  e(1/n, \log n) =
  n^{\alpha - 1}
    \sum_{i=0}^{r+\delta} \sum_{\ell=0}^{k} e_{i,\ell} n^{-i} \log^\ell n
\end{equation}
where the coefficients $e_{i,\ell}$ for $0 \leq i \leq r-1$ are exact,
and the $e_{i,\ell}$ for $i \geq r$ are complex balls.
We ask that~\eqref{eq:coeffexpansion} hold as soon as
$n > \max(n_0, s \, |\alpha|)$
for some fixed $s > 2$.

Flajolet and Odlyzko's proof~\cite[Theorem~3A]{flajolet1990singularity}
of the asymptotic analogue of~\eqref{eq:coeffexpansion}
yields as a byproduct a simple and efficient algorithm for computing the
coefficients~$e_{i,\ell}$~\cite[Section~IV.2]{Salvy1991}, and
implicitly contains a bound for the remainder of the asymptotic expansion.
Unfortunately, we were not able to obtain satisfactory numeric bounds for the
error terms using this approach.
Instead, we take a less direct route, closer to the older method of
Jungen~\cite{jungen1931series}, that allows us to invoke sharp bounds
readily available in the literature for subexpressions involving gamma and
polygamma functions%
\footnote{It remains an interesting question to see if a more careful analysis of
the error terms in Flajolet and Odlyzko's expansion could yield bounds suitable
for our purposes.
}.
The starting point is the following observation~\cite[Note~VI.7]{FlajoletSedgewick2009}.
\begin{proposition} \label{prop:Jungen}
  For all $\alpha \in \CC$ and $k \in \N$, one has
  \begin{equation}\label{eq:Jungen}
  [z^n](1-z)^{-\alpha} \log^k\left(\frac{1}{1-z}\right)
  = \frac{\d^k}{\d \alpha^k} \binom{n + \alpha - 1}{n}
  = \left. \frac{\d^k}{\d t^k} \frac{\Gamma(n+t)}{\Gamma(t) \Gamma(n+1)} \right|_{t=\alpha}
  \end{equation}
  where, in the case $\alpha \in \Z_{\leq 0}$, the evaluation is to be understood as a limit.
\end{proposition}


Fix $\alpha \in \CC$.
By taking logarithmic derivatives, we write
\begin{equation}\label{eq:three_parts1}
  \begin{aligned}
    \frac{\d^k}{\d\alpha^k} \frac{\Gamma(n+\alpha)}{\Gamma(\alpha) \Gamma(n+1)}
    = H(n, k)
      \cdot \frac{\Gamma(n+\alpha)}{\Gamma(n+1)}
  \end{aligned}
\end{equation}
for some factor $H(n, k)$, and let
\[
  G(n) = n^{1-\alpha} \, \frac{\Gamma(n+\alpha)}{\Gamma(n+1)}
\]
so that
\begin{equation}\label{eq:three_parts2}
  \frac{\d^k}{\d\alpha^k} \frac{\Gamma(n+\alpha)}{\Gamma(\alpha) \Gamma(n+1)}
  = n^{\alpha-1} \, G(n) \, H(n, k).
\end{equation}
We now proceed to bound each of the factors $G(n)$ and $H(n, k)$ with
an expression of the same shape as the right-hand side
of~\eqref{eq:coeffexpansion}, after replacing $n^{\alpha-1}$ by $n^0=1$.

The case $\alpha \in \Z_{\leq 0}$ is special.
When $\alpha \in \Z_{\leq 0}$ and $k = 0$, the factor~$H(n, k)$ reduces to
$1/\Gamma(\alpha)$ and vanishes.
The factor~$G(n)$ has a simple pole
when $\alpha \in \Z_{\leq 0}$ and $0 \leq n \leq -\alpha$;
however, this subcase does not occur in our setting
due to the assumption that $n > s |\alpha|$.
When $n > -\alpha$, the factor~$G(n)$ remains finite,
and hence the whole expression vanishes
when $\alpha \in \Z_{\leq 0}$ with $k = 0$.
Finally, as we will see, $H(n, k)$ takes a finite value when
$\alpha \in \Z_{\leq 0}$ with $k \geq 1$ and $n > -\alpha$
despite the apparent presence of a factor~$1/\Gamma(\alpha)$.

\begin{lemma}\label{lem:varchange}
  For all $\alpha \in \CC$ and $n \in \N$ with $n > s |\alpha|$ for $s \geq 1$, one has
  \begin{align}
    \label{eq:ineq_var_change_inverse}
    \left(n+\alpha \right)^{-1}
      &= \sum_{j=0}^{r-2} (-\alpha)^{j} n^{-j-1}
        + R^{(1)}, &
    |R^{(1)}|
      &\leq \frac{|\alpha|^{r-1}}{1 - 1/s} n^{-r} \\
    \label{eq:ineq_var_change_log}
    \log \left( n+\alpha \right)
      &= \log n - \sum_{j=1}^{r-1} \frac{(-\alpha)^{j}}{j} n^{-j}
        + R^{(2)}, &
    |R^{(2)}|
      &\leq \log(1 + 1/s) \, |\alpha|^r n^{-r}.
  \end{align}
\end{lemma}

\begin{proof}
  The bounds result directly from the Taylor expansions of $(1+z)^{-1}$ and
  $\log(1+z)$.
\end{proof}

\subsection{The Gamma Ratio}

Up to a convenient normalization factor, $G(n)$ is a quotient of the form
$\Gamma(n+\alpha)/\Gamma(n+\beta)$.
Erdélyi~\cite{TricomiErdelyi1951} gave an explicit asymptotic expansion of this
quotient in terms of the
\emph{generalized Bernoulli numbers} $B_{2j}^{(2\sigma)}(\sigma)$
defined by
\[
  \left( \frac{t}{e^t - 1} \right)^{2\sigma} e^{\sigma t}
  = \sum_{j = 0}^{\infty} \frac{t^{2j}}{(2j)!} B_{2j}^{(2\sigma)}(\sigma).
\]
In order to bound the remainder term, we use the following result of Frenzen
(stated here in the special case~$\beta = 0$).

\begin{theorem}[Frenzen~\cite{frenzen1992error}] \label{thm:Frenzen}
Let $\sigma$ be any complex number, let $w$ be a complex number such that
\begin{equation}\label{eq:condition_angle}
\mathopen|\arg(w)| < \pi/2 \quad \text{ and } \quad \Rel(w) > |\operatorname{Im}(\sigma)|,
\end{equation}
and let $\eta$ be a positive integer.
Then one has
\begin{equation}
\frac{\Gamma(w + \sigma)}{\Gamma(w - \sigma + 1)} = \sum_{j=0}^{\eta-1} \frac{\Gamma(1 - 2\sigma + 2j)}{\Gamma(1 - 2\sigma) (2j)!} B_{2j}^{(2\sigma)}(\sigma) w^{2\sigma - 1 - 2j} + R^{\mathrm{Fr}}_\eta(w, \sigma),
\end{equation}
where
\begin{equation}
|R^{\mathrm{Fr}}_\eta(w, \sigma)| \leq \frac{\Gamma(1 - \Rel(2\sigma) + 2\eta)}{|\Gamma(1 - 2\sigma)| (2\eta)!} \left|B_{2\eta}^{(|2\sigma|)}(|\sigma|)\right| \cdot \Rel(w)^{\Rel(2 \sigma) - 1 - 2\eta}.
\end{equation}
\end{theorem}

To apply this theorem, we further decompose $G(n)$ as $G_1(n) \, G_2(n)$ with
\begin{equation*}
  G_1(n) = \left(n + \frac{\alpha}{2}\right)^{1 - \alpha} \frac{\Gamma(n+\alpha)}{\Gamma(n+1)} 
  \qquad\text{and}\qquad
  G_2(n) = \left( 1 + \frac{\alpha}{2n} \right)^{\alpha - 1}.
\end{equation*}

\begin{corollary}\label{cor:ratio_gamma}
Let $\alpha$ be any complex number,
and let $\eta$ be a positive integer.
For any integer~$n > s |\alpha|$ (where $s \geq 2$), one has
\begin{equation}\label{bound:ratio_gamma_main}
    G_1(n) =
    \sum_{j=0}^{\eta-1} \frac{\Gamma(1 - \alpha + 2j)}{\Gamma(1 - \alpha) (2j)!} B_{2j}^{(\alpha)}(\alpha/2) \left(n + \frac{\alpha}{2}\right)^{- 2j} + R^{(G_1)},
\end{equation}
with
\begin{equation}\label{bound:ratio_gamma_error}
\begin{split}
|R^{(G_1)}| \leq
\frac{\Gamma(1 - \Rel(\alpha) + 2\eta)}{|\Gamma(1 - \alpha)| (2\eta)!} &\cdot
\left|B_{2\eta}^{(|\alpha|)}\left(\left|\frac{\alpha}{2}\right|\right)\right|
e^{|\operatorname{Im}(\alpha)| \arcsin \frac1{2s}} \\
&\cdot \left( \frac{s + 1/2}{s - 1/2} \right)^{\max\{0, 2\eta + 1 - \Rel(\alpha)\}}
\left| n+\frac{\alpha}{2} \right|^{- 2\eta}.
\end{split}
\end{equation}
\end{corollary}

\begin{proof}
Letting $w = n + \alpha/2$ and $\sigma = \alpha/2$, one can verify that the conditions \eqref{eq:condition_angle} are met when $n > |\alpha|$, and 
the asymptotic expansion becomes
\begin{equation}\label{eq:asy1}
\frac{\Gamma(n + \alpha)}{\Gamma(n + 1)} = \sum_{j=0}^{\eta-1} \frac{\Gamma(1 - \alpha + 2j)}{\Gamma(1 - \alpha) (2j)!} B_{2j}^{(\alpha)}(\alpha/2) \left(n + \frac{\alpha}{2}\right)^{\alpha - 1 - 2j} + R^{\mathrm{Fr}}_\eta\left(n + \frac\alpha2, \sigma\right).
\end{equation}
Multiplying by $\left( n+\frac{\alpha}{2} \right)^{1 - \alpha}$ then
gives~\eqref{bound:ratio_gamma_main}, with
$R^{(G_1)} = R^{\mathrm{Fr}}_\eta(n+\alpha/2, \sigma)$ and
\[
\left|R^{(G_1)}\right| \leq
  \frac{\Gamma(1 - \Rel(\alpha) + 2\eta)}{|\Gamma(1 - \alpha)| (2\eta)!} \left|B_{2\eta}^{(|\alpha|)}\left(\left|\frac{\alpha}{2}\right|\right)\right|
  \cdot \Rel(n + \alpha/2)^{\Rel(\alpha) - 1 - 2\eta}
  \left| \left(n+\frac{\alpha}{2}\right) ^{1 - \alpha} \right|.
\]
Writing
$| (n+\alpha/2) ^{1 - \alpha} |
= |n + \alpha/2|^{1 - \Rel \alpha} \cdot e^{- \arg(n + \alpha/2) \operatorname{Im} \alpha}$
and observing that the assumption on~$n$ implies $\mathopen|\arg(n + \alpha/2)| \leq \arcsin\frac1{2s}$,
we have
\[
  \Rel(n + \alpha/2)^{\Rel(\alpha) - 1 - 2\eta}
  \left| \left(n+\frac{\alpha}{2}\right) ^{1 - \alpha} \right|
  \leq
    \begin{aligned}[t]
      e^{|\operatorname{Im}(\alpha)| \arcsin\frac1{2s}}
      \cdot \left| n+\frac{\alpha}{2} \right|^{- 2\eta} \\
      \cdot \left( \frac{|n+\alpha /2|}{\Rel(n + \alpha/2)} \right)^{2\eta + 1 - \Rel(\alpha)}.
    \end{aligned}
\]
Since $\Rel(n+\alpha/2) \geq n - |\alpha/2|$, 
the last factor satisfies
\[
  \left( \frac{|n+\alpha /2|}{\Rel(n + \alpha/2)} \right)^{2\eta + 1 - \Rel(\alpha)}
  \leq  \left( \frac{s + 1/2}{s - 1/2} \right)^{\max\{0, 2\eta + 1 - \Rel(\alpha)\}},
\]
leading to~\eqref{bound:ratio_gamma_error}.
\end{proof}

This corollary, applied with
$\eta = \left \lceil{r/2}\right \rceil$,
provides us with an explicit error bound for an asymptotic expansion of
$n^{\alpha-1} G_1(n)$.
However, the expansion is in descending powers of $n + \alpha/2$,
while our algorithm aims to compute expansions in descending powers of~$n$.
To obtain such an expression, we substitute for $(n + \alpha/2)^{-1}$ its
expansion in powers of~$1/n$ with the error term given by
Lemma~\ref{lem:varchange} (applied to~$s' = 2s$).

The asymptotic expansion of the normalization factor~$G_2(n)$ is
convergent and reduces to the binomial series.

\begin{lemma} \label{lem:power_main}
  Let $\alpha \in \CC$ and let $r$~be a nonnegative integer.
  For all $n \in \N$ such that $n \geq s |\alpha|$ (with $s \geq 2$), one has
  \begin{equation}\label{bound:power_main}
    G_2(n)
    = \left( 1+\frac{\alpha}{2n} \right)^{\alpha-1}
    = \sum_{j=0}^{r-1} \binom{\alpha - 1}{j}\left(\frac{\alpha}{2}\right)^{j} n^{-j} + R^{(G_2)}(n)
  \end{equation}
  with
  \begin{equation}\label{bound:power_error}
    |R^{(G_2)}(n)| \leq
      n^{- r}
      \cdot \frac{|\alpha|^{r}}{1 - 1/s}
      \cdot \max\left\{ \left(3/2\right)^{\Rel(\alpha) - 1}, \left(1/2\right)^{\Rel(\alpha) - 1} \right\}
      \cdot e^{|\operatorname{Im}(\alpha)|/2}.
  \end{equation}
\end{lemma}

\begin{proof}
  Let
  $\varphi(z) = (1+z)^{\alpha - 1}$
  and consider the Taylor expansion
  \[
    \varphi(z)
    = \sum_{j=0}^{\infty} \binom{\alpha-1}{j} z^j
    = \sum_{j=0}^{r-1} \binom{\alpha-1}{j} z^j + \tilde R(z).
  \]
  Cauchy's estimate on the disk $|z| \leq 1/2$ gives
  $\bigl| \binom{\alpha-1}{j} \bigr| \leq 2^j M$
  where
  \begin{align*}
    M = \max_{|z| = 1/2} \bigl|\varphi(z)\bigr|
    &= \left|1+z\right|^{\Rel(\alpha) - 1} \cdot e^{- \arg(1 + z) \cdot \operatorname{Im}(\alpha)} \\
    &\leq \max\left\{ \left(3/2\right)^{\Rel(\alpha) - 1}, \left(1/2\right)^{\Rel(\alpha) - 1} \right\}
      \cdot e^{|\operatorname{Im}(\alpha)|/2},
  \end{align*}
  and hence
  \[ |\tilde R(z)| \leq \frac{M \, (2 |z|)^r}{1 - 2 |z|}. \]
  Our assumptions on $n$~and~$s$ ensure that $n \neq 0$ and
  $\alpha/(2n) \leq 1/s < 1/2$.
  The bound~\eqref{bound:power_error} follows by substituting $\alpha/(2n)$
  for~$z$.
\end{proof}

\begin{remark} \label{rk:ratio_gamma_exact}
  In the special case $\alpha \in \Z_{\geq 1}$,
  the normalized gamma ratio $G(n)$ reduces to a polynomial in~$n^{-1}$ of
  degree~$\alpha - 1$.
  Computing this polynomial directly when $r \geq \alpha$ instead of using the
  previous results avoids introducing a spurious error term.
  When $\alpha = 0$, one has $G(n) = 1$, but the error terms
  \eqref{bound:ratio_gamma_error} (for $\eta \geq 1$)
  and~\eqref{bound:power_error} already vanish naturally.
\end{remark}

\subsection{Derivatives}
\label{sec:derivatives}

We now turn to the last factor in~\eqref{eq:three_parts2},
\begin{equation}\label{eq:H}
  H(n, k) =
    \frac{1}{\Gamma(n+\alpha)}
    \frac{\d^k}{\d\alpha^k} \frac{\Gamma(n+\alpha)}{\Gamma(\alpha)}.
\end{equation}
It can be checked that
the product $\Gamma(\alpha) \, H(n, k)$ for any given~$k$ can be written as
a polynomial (with integer coefficients) in the differences
\[
  \psi^{(m)}(n+\alpha) - \psi^{(m)}(\alpha), \qquad
  0 \leq m \leq k - 1,
\]
where $\psi^{(m)}$ is the \emph{polygamma function} defined by
\[
  \psi^{(m)}(z) = \frac{\d^{m+1}}{\d z^{m+1}} \log \Gamma(z).
\]
Therefore, to obtain an asymptotic enclosure of~$H(n, k)$
it suffices in principle (at least when $\alpha \notin \Z_{\leq 0}$)
to have expansions with error bounds of the various
$\psi^{(m)} (n + \alpha)$.
For this we rely on the following theorem.

\begin{theorem}[{Nemes~\cite{Nemes2017}}] \label{thm:Nemes}
  Let~$z$ be a nonzero complex number with
  $\mathopen|\arg z\mathclose| \leq \pi/4$,
  and let $m \geq 0$~and~$\eta \geq 1$ be integers.
  Then
  {\small
  \begin{equation}\label{bound:logder_main}
    \psi^{(m)}(z) =
    \everymath{\displaystyle}
    \begin{cases}
    \log z - \frac{1}{2z} - \sum_{j=1}^{\eta-1} \frac{b_{2j}}{2j z^{2j}}
    + R^{\psi,0}_\eta(z)
    &\text{if } m = 0\\
    \begin{aligned}
    (-1)^{m+1} &\left( \frac{(m-1)!}{z^m} + \frac{m!}{2z^{m+1}} + \sum_{j=1}^{\eta-1} \frac{b_{2j}}{(2j)!} \frac{(2j + m - 1)!}{z^{2j+m}} \right) \\
    &+ R^{\psi,m}_\eta(z)
    \end{aligned}
    &\text{if } m \geq 1,
    \end{cases}
  \end{equation}
  }
  where the $b_j$~are the Bernoulli numbers and the error term satisfies
  \begin{equation}\label{bound:logder_error}
    |R^{\psi, m}_\eta(z)|
    \leq (2\eta + m -1)! \frac{\abs{b_{2\eta}}}{(2\eta)!} |z|^{-m-2\eta}.
  \end{equation}
\end{theorem}

\begin{proof}
  The second and third paragraph of Section~4 in~\cite{Nemes2017} give
  \[
    R^{\psi, m}_\eta(z) = (-1)^{m+1} m! z^{-m} R_\eta(m+1, z)
    \qquad
    (m \geq 0, \eta \geq 1, \Rel z > 0)
  \]
  in terms of the quantity~$R_\eta(\mu, z)$ defined
  by~\cite[Equation~(2.1)]{Nemes2017}.
  When $\mathopen|\arg z| \leq \pi/4$, one has
  \[
    |R_\eta(\mu, z)| \leq
      \frac{|b_{2\eta}|}{(2\eta)!}
      \frac{\Gamma(2\eta+\mu-1)}{\Gamma(\mu)}
      |z|^{-2\eta}
  \]
  by~\cite[Equation~(5.1)]{Nemes2017}.
\end{proof}

The expression of $H(n, k)$ in terms of
$\psi^{(m)}(n+\alpha) - \psi^{(m)}(\alpha)$
mentioned above is complicated,
and computing it symbolically before substituting in the expansions from
Theorem~\ref{thm:Nemes} would be rather inefficient.
The following variants using generating series, however, lead to a
reasonably simple algorithm for computing $H(n,k)$.
The main advantage
is that we avoid building the full expression of~$H(n,k)$ as a polynomial
in the~$\psi^{(m)}(n+\alpha)$ (which would contain, in general, many occurrences
of each~$\psi^{(m)}$) and minimize the algebraic operations to be
performed on the expressions in~$n$ that result from substituting for
the~$\psi^{(m)}$ their asymptotic expansions.
In addition, we can compute all~$H(n,k)$ for $0 \leq k \leq K$ at once.

\begin{proposition}\label{prop:genH}
  Assume $\alpha \in \CC \setminus \Z_{\leq 0}$.
  Then, for all $n \geq 0$ and $K \geq k$, one has
  \begin{equation} \label{eq:genH1}
    H(n,k)
    = \frac{k!}{\Gamma(\alpha)}
      [\varepsilon^k] \exp \left(
          \sum_{m=0}^{K-1}
            \frac{\psi^{(m)}(n+\alpha) - \psi^{(m)}(\alpha)}{(m+1)!}
            \varepsilon^{m+1}
        \right).
  \end{equation}
  Now assume $\alpha \in \CC \setminus \Z_{> 0}$.
  Then, for all $n \geq |\alpha|$ and $K \geq k$, one has
  \begin{multline} \label{eq:genH2}
    H(n, k)
    = k! \Gamma(1{-}\alpha)
      [\varepsilon^k]
        \Biggl(
          \frac{\sin\bigl(\pi(\alpha+\varepsilon)\bigr)}{\pi} \\
          \cdot \exp \left(
              \sum_{m=0}^{K-2}
                \frac{\psi^{(m)}(n+\alpha) + (-1)^{m+1} \psi^{(m)}(1{-}\alpha)}
                     {(m+1)!}
                \varepsilon^{m+1}
            \right)
        \Biggr).
  \end{multline}
\end{proposition}

For $\alpha \in \Z_{\leq 0}$ (the main case of interest),
the factor $\sin\bigl(\pi(\alpha+\varepsilon)\bigr)$ in~\eqref{eq:genH2}
reduces to $(-1)^\alpha \sin(\pi\varepsilon)$.
In particular, as already noted, one then has $H(n,0) = 0$.

\begin{proof}
  In the case $\alpha \in \CC \setminus \Z_{\leq 0}$, the function
  $z \mapsto \log\bigl(\Gamma(n+z)/\Gamma(z)\bigr)$
  is analytic at~$\alpha$, with Taylor expansion
  \[
    \log \frac{\Gamma(n+\alpha+\varepsilon)}{\Gamma(\alpha+\varepsilon)}
    = \log \frac{\Gamma(n+\alpha)}{\Gamma(\alpha)}
      + \sum_{m=0}^{\infty}
          \frac{\psi^{(m)}(n+\alpha) - \psi^{(m)}(\alpha)}{(m+1)!}
          \varepsilon^{m+1}.
  \]
  Therefore, one has
  \[
    \frac{\Gamma(\alpha)}{\Gamma(n+\alpha)}
    \frac{\Gamma(n+\alpha+\varepsilon)}{\Gamma(\alpha+\varepsilon)}
    = \exp \left(
          \sum_{m=0}^{\infty}
            \frac{\psi^{(m)}(n+\alpha) - \psi^{(m)}(\alpha)}{(m+1)!}
            \varepsilon^{m+1}
        \right).
  \]
  The coefficient of~$\varepsilon^k$ in this series only depends on the
  coefficients of~$\varepsilon^m$ with $m \leq k$ in the sum.
  Comparing with the definition~\eqref{eq:H} of $H(n,k)$
  gives~\eqref{eq:genH1}.

  Now consider the case $\alpha \in \CC \setminus \Z_{> 0}$.
  Euler's reflection formula gives
  \[
    \frac{\Gamma(n+\alpha+\varepsilon)}{\Gamma(\alpha+\varepsilon)}
    = \Gamma(n+\alpha+\varepsilon)
      \Gamma(1-\alpha-\varepsilon)
      \frac{\sin\bigl(\pi(\alpha+\varepsilon)\bigr)}{\pi},
  \]
  where the product
  $\Gamma(n+\alpha+\varepsilon) \Gamma(1-\alpha-\varepsilon)$
  does not vanish at~$\varepsilon = 0$.
  By the same reasoning as above, for small $\varepsilon \in \CC$ one has
  \begin{multline*}
    \frac{1}{\Gamma(n+\alpha)\Gamma(1-\alpha)}
    \frac{\Gamma(n+\alpha+\varepsilon)}{\Gamma(\alpha+\varepsilon)} \\
    = \frac{\sin\bigl(\pi(\alpha+\varepsilon)\bigr)}{\pi}
      \exp \left(
          \sum_{m=0}^{\infty}
            \frac{\psi^{(m)}(n+\alpha) + (-1)^{m+1} \psi^{(m)}(1{-}\alpha)}
                  {(m+1)!}
            \varepsilon^{m+1}
        \right),
  \end{multline*}
  and \eqref{eq:genH1}~follows.
\end{proof}

In order to compute bounds on $H(n,k)$, we replace
each occurrence of~$\psi^{(m)}(n+\alpha)$ in
\eqref{eq:genH1}~or~\eqref{eq:genH2}
by the corresponding expression from Theorem~\ref{thm:Nemes},
then expand the result as a power series in~$\varepsilon$
and extract the coefficient of~$\varepsilon^k$.
Doing so yields an expression in
$\CBF[(n+\alpha)^{-1}, \log(n+\alpha)]$
whose evaluation at any sufficiently large~$n$ contains~$H(n,k)$.
As in the previous subsection, we replace
$\log(n+\alpha)$ and $(n+\alpha)^{-1}$ by
their expansions in powers of~$1/n$ given by Lemma~\ref{lem:varchange}.

\subsection{Algorithm}

\begin{algorithm}[p]
\caption{Coefficient bounds for algebraic-logarithmic monomials}
\label{algo:coeffasy}
\begin{description}[itemindent=0pt-\widthof{(1)},leftmargin=2em]
\item[Input]
  An ``algebraic'' singularity order~$\alpha \in \Qbar$,
  a logarithmic singularity order~$K \in \N$,
  an expansion order $r \in \N$,
  parameters $s \geq 2$ and $n_0 \in \N$.
\item[Output]
  A polynomial $e \in \CBF[n^{-1}, w, \varepsilon]$ satisfying~\eqref{eq:form_bound}.
\end{description}
\begin{enumerate}[leftmargin=2em]
  \item \label{step:ratio_gamma}
    If $\alpha \in \Z_{\geq 1}$ and $r \geq \alpha$:
  \begin{enumerate}
    \item \label{step:ratio_gamma_exact}
      Compute $g(n^{-1}) = \prod_{j=1}^{\alpha} (1 + j n^{-1}) \in \CBF[n^{-1}]$.
  \end{enumerate}
  Else:
  \begin{enumerate}[resume]
    \item \label{step:ratio_gamma_main}
      Using~\eqref{bound:ratio_gamma_main} with
      $\eta = \lceil r/2 \rceil$,
      compute a polynomial $\hat g_1(u) \in \CBF[u]$ of degree $\leq 2\eta$
      such that
      $G_1(n) \in \hat g_1\bigl((n+\alpha/2)^{-1}\bigr)$
      for all $n > s \alpha$.
    \item \label{step:ratio_gamma_subs}
      Make the substitution
      \[
        u =
        \sum_{j=0}^{r-2} (-\alpha/2)^j n^{-j-1}
        + \B(0, \delta_1) n^{-r}
        \quad\text{where }
        \delta_1 = \frac{|\alpha/2|^{r-1}}{1-1/(2s)}
      \]
      in~$\hat g_1(u)$.
      \label{step:coeffasy:trim-n}
      Trim the result to degree~$\leq r$ in $n^{-1}$
      by replacing every occurrence of~$n^{-r-j}$ with $j > 0$ by
      $\B(0, n_0^{-j}) n^{-r}$,
      resulting in a polynomial $g_1(n^{-1}) \in \CBF[n^{-1}]$.
    \item \label{step:power_main}
      Compute~$g_2(n^{-1}) \in \CBF[n^{-1}]$ such that~$G_2(n) \in g_2(n)$ for all
      $n > s \mathopen| \alpha \mathclose|$ using~\eqref{bound:power_main}.
    \item Set $g(n^{-1}) = g_1(n^{-1}) g_2(n^{-1})$.
  \end{enumerate}
  \item \label{step:coeffasy:generic}
    If $\alpha \notin \Z_{\leq 0}$:
  \begin{enumerate}
    \item \label{step:coeffasy:series-start}
      Using~\eqref{bound:logder_main} with $\eta = \max(1, \lceil r/2 \rceil)$,
      compute polynomials
      $p_0, \dots, p_{K-1} \in \CBF[v]$
      of degree~$r$
      such that, for $\mathopen|\arg v \mathclose| \leq \pi/4$,
      \begin{equation}\label{eq:varNemes} \tag{$\ast$}
        \begin{cases}
          \psi^{(0)}(v^{-1}) - \log(v^{-1}) - \psi^{(0)}(\alpha)
            \in p_0(v) \\[+2mm]
          \displaystyle
          \frac{\psi^{(m)}(v^{-1}) - \psi^{(m)}(\alpha)}
               {(m+1)!}
            \in p_m(v) &
          \text{if } m \geq 1. \\
        \end{cases}
      \end{equation}
      Set $p(v, \varepsilon) = \sum_{m=0}^{K-1} p_{m}(v) \varepsilon^{m+1}$.
    \item \label{step:coeffasy:expofasy}
      Compute the truncated series expansion
      \[
        \Gamma(\alpha)^{-1} \exp\bigl(p(v, \varepsilon)\bigr)
          = \hat h_1(v, \varepsilon)
          + O(\varepsilon^{K+1}),
      \]
      resulting in a polynomial
      $\hat h_1(v, \varepsilon) \in \CBF[v, \varepsilon]$ of degree $\leq K$ in~$\varepsilon$.
      Trim the polynomial $\hat h_1(v, \varepsilon)$ to degree $\leq r$ in~$v$ by replacing
      every occurrence of~$v^{r+j}$ with $j > 0$ by
      $\B(0, (n_0 - |\alpha|)^{-j}) v^r$.
  \end{enumerate}
  Else:
  \begin{enumerate}[resume]
    \item \label{step:coeffasy:negint}
      Perform steps
      \eqref{step:coeffasy:series-start} and \eqref{step:coeffasy:expofasy}
      with $p(v, \varepsilon)$ and $\hat h_1(v, \varepsilon)$ replaced, respectively, by
      \begin{align*}
        p(v, \varepsilon) &=
          \sum_{m=0}^{K-2}
            \frac{\psi^{(m)}(v^{-1}) + (-1)^{m+1} \psi^{(m)}(1-\alpha)}
                {(m+1)!}
            \varepsilon^{m+1} \\
          \text{and}\qquad \hat h_1(v, \varepsilon) &=
          (-1)^\alpha
          \Gamma(1-\alpha)
          \exp\bigl(p(v, \varepsilon)\bigr)
          \sin(\pi\varepsilon)/\pi.
      \end{align*}
  \end{enumerate}
  \item \label{step:coeffasy:subs2}
    In $\hat h_1(v, \varepsilon)$,
    substitute for~$v$ the enclosure of $(n+\alpha)^{-1}$
    (a polynomial in~$n^{-1}$ with ball coefficients)
    given by~\eqref{eq:ineq_var_change_inverse}.
    Trim the result to degree~$\leq r$ in~$n^{-1}$
    as in step~\eqref{step:coeffasy:trim-n}.
    Call the result
    $h_1(n^{-1}, \varepsilon) \in \CBF[n^{-1}, \varepsilon]$.
  \item \label{step:exp}
    Let $\delta_2 = \log(1+1/s) |\alpha|^r$ and
    $q(n^{-1}) = \sum_{j=1}^{r-1} \frac{(-\alpha)^j}{j!} n^{-j}
                 + \B(0, \delta_2) n^{-r}$
    Compute the truncated series expansions
    \begin{align*}
      \exp \bigl( q(n^{-1}) \varepsilon \bigr)
        &= h_2(n^{-1}, \varepsilon)
          + O(\varepsilon^{K+1}), \\
      \exp(w \varepsilon) &= h_3(w, \varepsilon) + O(\varepsilon^{K+1}).
    \end{align*}
  \item
    Compute
    $
      g(n^{-1}) \,
      h_1(n^{-1}, \varepsilon) \,
      h_2(n^{-1}, \varepsilon) \,
      h_3(w, \varepsilon)
    $,
    truncated to degree~$K$ in~$\varepsilon$.
    Multiply the coefficient of~$\varepsilon^k$ by~$k!$ for each~$k$.
    Trim the result to degree~$\leq r$ in~$n^{-1}$
    as in step~\eqref{step:coeffasy:trim-n} and return it.
\end{enumerate}
\end{algorithm}

Algorithm~\ref{algo:coeffasy} (page~\pageref{algo:coeffasy}) summarizes the steps that result from our
previous discussion for bounding the coefficient of~$z^n$ for large~$n$ in the
series
$(1-z)^{-\alpha} \log^k(1/(1-z))$.

On several occasions, the algorithm needs to \emph{trim}
a given~$p \in \CBF[z]$ to degree~$d$; that is, compute
a polynomial~$\tilde p$ of degree at most~$d$ such that
$p(z) \subseteq \tilde p(z)$
for all $z$~in a certain range of interest.

\begin{proposition} \label{prop:coeffasy}
  Given $\alpha$, $K$, $r$, $s$, and~$n_0$
  with $n_0 > s \mathopen|\alpha\mathclose|$,
  Algorithm~\ref{algo:coeffasy} computes a polynomial
  \[
    e(n^{-1}, w, \varepsilon)
    = e_0(n^{-1}, w) + \dots + e_K(n^{-1}, w) \varepsilon^K
  \]
  such that, for all $k \leq K$ and $n \geq n_0$,
  \begin{equation}\label{eq:form_bound}
  [z^n](1-z)^{-\alpha} \log^k\left(\frac{1}{1-z}\right)
  \in n^{\alpha-1} \cdot e_k(n^{-1}, \log n).
  \end{equation}
  The polynomial $e_k(n^{-1}, w)$ has degree at
  most $r$ with respect to~$n^{-1}$ and degree at most~$k$ with respect to~$w$, and its
  coefficients of degree less than~$r$ in~$n^{-1}$ are elements of~$\rholgamma$.
\end{proposition}

\begin{proof}
  Fix $n \geq n_0 > s |\alpha|$ and $k \leq K$.
  Step~\eqref{step:ratio_gamma_exact} implements Remark~\ref{rk:ratio_gamma_exact} and computes $G(n) = g(n^{-1})$ exactly.
  In the general case, the polynomial~$\hat g_1$ computed at
  step~\eqref{step:ratio_gamma_main} satisfies
  $G_1(n) \in \hat g_1\bigl((n+\alpha/2)^{-1}\bigr)$
  by Corollary~\ref{cor:ratio_gamma}.
  By Lemma~\ref{lem:varchange} applied with $\alpha$~replaced by~$\alpha/2$
  and $s$~replaced by~$2s$, the result of the substitution at
  step~\eqref{step:ratio_gamma_subs}, evaluated at~$n$, also contains~$G_1(n)$.
  The same property holds for $g_1(n^{-1})$ since, for all $c \in \CC$,
  one has $c n^{-r-j} \in c n^{-r} \B(0, n^{-j})$.
  The fact that $G_2(n) \in g_2(n)$ after step~\eqref{step:power_main} comes directly
  from Lemma~\ref{lem:power_main}.
  In both cases, we have $G(n) \in g(n^{-1})$ after step~\eqref{step:ratio_gamma}.

  Let us now show that
  $H(n,k) \in [\varepsilon^k] \bigl(h_1(n^{-1}, \varepsilon)
  h_2(n^{-1}, \log n, \varepsilon) \bigr)$.
  Assume first that $\alpha \in \CC \setminus \Z_{\leq 0}$.
  Theorem~\ref{thm:Nemes} proves that it is possible to compute
  polynomials~$p_k$ satisfying the condition~\eqref{eq:varNemes}
  appearing at step~\eqref{step:coeffasy:series-start}.
  Thus, for any $v \neq 0$ with $\abs{\arg v} \leq \pi/4$,
  we have after step~\eqref{step:coeffasy:series-start}
  \[
    [\varepsilon^k]
      \exp \left(
          \sum_{m=0}^{K-1}
            \frac{\psi^{(m)}(v^{-1}) - \psi^{(m)}(\alpha)}{(m+1)!}
            \varepsilon^{m+1}
        \right)
    \in
    [\varepsilon^k]
      \left(
        v^{-\varepsilon}
        \exp \left(
            \sum_{m=0}^{K-1} p_m(v) \varepsilon^{m+1}
          \right)
      \right),
  \]
  and therefore, using Proposition~\ref{prop:genH},
  \[
    \frac{H(n,k)}{k!}
    \in [\varepsilon^k] \bigl( v^{-\varepsilon} \hat h_1(v, k) \bigr)
  \]
  at step~\eqref{step:coeffasy:expofasy} before $\hat h_1(v, \varepsilon)$ is
  trimmed to degree $\leq r$.
  Now, since $n > s |\alpha|$ we have
  $\mathopen|\arg (n + \alpha)^{-1} \mathclose| \leq \pi/4$
  and
  $|(n+\alpha)^{-1}| \leq (n-|\alpha|)^{-1}$,
  so that
  \[
    H(n,k) \in k! \, [\varepsilon^k] \bigl(
        (n+\alpha)^{\varepsilon}
        \hat h_1((n+\alpha)^{-1}, \varepsilon)
      \bigr)
  \]
  after step~\eqref{step:coeffasy:expofasy}.
  A similar argument shows that the same conclusion holds after
  step~\eqref{step:coeffasy:negint} in the case $\alpha \in \Z_{\leq 0}$.
  Lemma~\ref{lem:varchange} turns this into
  \[
    H(n,k) \in k! \, [\varepsilon^k] \bigl(
        (n+\alpha)^{\varepsilon}
        h_1(n^{-1}, \varepsilon)
      \bigr).
  \]
  Using the series expansion of~$v^{-\varepsilon}$ with respect to~$\varepsilon$, and
  Lemma~\ref{lem:varchange} again, we have
  \[
    [\varepsilon^k] (n+\alpha)^\varepsilon \in
    [\varepsilon^k] h_2(n^{-1}, \varepsilon) h_3(\log n, \varepsilon)
  \]
  for all $k \leq K$
  after step~\eqref{step:exp}.

  Summing up, after step~\eqref{step:exp}, we have
  \[
    G(n) \in g(n^{-1}) \quad \text{and} \quad
    H(n,k) \in k! \, [\varepsilon^k] \bigl(h_1(n^{-1}) h_2(n^{-1}, \varepsilon) h_3(\log n, \varepsilon) \bigr).
  \]
  The bound~\eqref{eq:form_bound} then follows
  from \eqref{eq:Jungen}~and~\eqref{eq:three_parts2}.

  Due to the final trimming step, the polynomial~$e(n^{-1}, w, \varepsilon)$
  returned by the algorithm has degree at most~$r$ in~$n^{-1}$.
  It is clear from the formula in step~\eqref{step:exp} that the degree
  in~$w$ of $[\varepsilon^k] \, h_3(w, \varepsilon)$ is at most~$k$.
  Thus the same bound holds for
  $[\varepsilon^k] \, e(n^{-1}, w, \varepsilon)$.

  Finally, the only non-exact balls manipulated by the algorithm
  are the explicit $\B(0, \dots)$ appearing in the pseudo-code,
  those implicit in the remainder terms of
  \eqref{bound:ratio_gamma_main},
  \eqref{eq:ineq_var_change_inverse},
  \eqref{bound:power_main}, and
  \eqref{bound:logder_main},
  and balls created from these by subsequent algebraic operations.
  All other numeric values belong to the field extension
  of~$\Q$ generated by~$\alpha$ and the $\Gamma(z)$ and $\gamma^{(m)}(z)$ for
  $z \in \Q(\alpha)$ and $m \geq 0$.
  Since only the coefficient of degree~$r$ of the final result depends on the
  ``wide'' balls $\B(0,\dots)$ listed above,
  the coefficients of degree less than~$r$ belong to this field extension, and in
  particular to~$\rholgamma$.
\end{proof}

\begin{remark} \label{rk:coeffasy_exact}
  The output of Algorithm~\ref{algo:coeffasy} is actually more precise than
  Proposition~\ref{prop:coeffasy} suggests in special cases where some terms
  of the asymptotic expansion of
  $[z^n] (1-z)^{-\alpha} \log^k(1/(1-z))$
  as $n \to \infty$ vanish.

  Firstly, when~$\alpha \in \Z_{\leq 0}$, the degree in~$w$ of $e_k(n^{-1}, w)$ is at
  most~$k-1$ instead of~$k$.
  Indeed, the polynomial $h_1(n, \varepsilon)$ vanishes at~$\varepsilon=0$,
  so that one has
  $\deg_w [\varepsilon^k] \, e_k(n^{-1}, w, \varepsilon) \leq k-1$.
  In particular, $e_0(n^{-1}, w)$ is exactly zero, as could be expected since
  $(1-z)^{-\alpha}$ is a polynomial in~$z$.

  Secondly, in the case $\alpha \in \Z_{\geq 0}$, the polynomial~$g$ has
  degree~$\alpha-1 < r$, i.e., contains no error term, when $r \geq \alpha$
  (\emph{cf.} Remark~\ref{rk:ratio_gamma_exact}).
  Additionally, the coefficient of highest degree in~$w$ in the product $h_1 h_2 h_3$ does not depend on~$n$, because
  $[\varepsilon^0] h_1(n^{-1}, \varepsilon)$
  and
  $[\varepsilon^0] h_2(n^{-1}, \varepsilon)$
  are constants.
  Thus, for each~$k$, the leading coefficient of~$e_k(n^{-1}, w)$ viewed as a
  polynomial in~$w$ is a polynomial in~$n^{-1}$ of degree at most $\alpha - 1$.
  In particular, $e_0(n^{-1}, w)$ has degree at most~$\alpha - 1$ in~$n^{-1}$
  (and does not depend on~$w$).
  This reflects the fact that $[z^n] (1-z)^{-\alpha}$ is a polynomial in~$n$.

  The previous two observations combine when $\alpha = 0$:
  one then has
  \[
    [z^n] \log^k(1/(1-z))
    \sim c \log(n)^{k-1} n^{-1} + d_2(\log(n)) n^{-2} + d_3(\log(n)) n^{-3} + \cdots
  \]
  for polynomials~$d_i$ of degree at most~$k - 2$, and the approximations
  computed by Algorithm~\ref{algo:coeffasy} reveal this form.
\end{remark}

\section{Contribution of a Singularity: The Local Error Term}
\label{sec:local_error}
\label{subsection:bound_g}

\begin{figure}
    \centering
    \begin{tikzpicture}[x=0.5pt,y=0.5pt,yscale=-1,xscale=1]

    \draw[use as bounding box]  (90.42,270.04) -- (610.42,270.04)(349.3,30) -- (349.3,500) (603.42,265.04) -- (610.42,270.04) -- (603.42,275.04) (344.3,37) -- (349.3,30) -- (354.3,37)  ;
    \draw  [draw opacity=0] (468.05,276.44) .. controls (465.39,284.35) and (457.91,290.04) .. (449.1,290.04) .. controls (438.05,290.04) and (429.1,281.09) .. (429.1,270.04) .. controls (429.1,258.99) and (438.05,250.04) .. (449.1,250.04) .. controls (458.09,250.04) and (465.7,255.98) .. (468.22,264.14) -- (449.1,270.04) -- cycle ; \draw  [color={rgb, 255:red, 245; green, 166; blue, 35 }  ,draw opacity=1 ] (468.05,276.44) .. controls (465.39,284.35) and (457.91,290.04) .. (449.1,290.04) .. controls (438.05,290.04) and (429.1,281.09) .. (429.1,270.04) .. controls (429.1,258.99) and (438.05,250.04) .. (449.1,250.04) .. controls (458.09,250.04) and (465.7,255.98) .. (468.22,264.14) ;
    \draw  [draw opacity=0] (279.84,172.38) .. controls (282.61,170.95) and (285.76,170.14) .. (289.1,170.14) .. controls (300.2,170.14) and (309.2,179.09) .. (309.2,190.14) .. controls (309.2,201.19) and (300.2,210.14) .. (289.1,210.14) .. controls (278,210.14) and (269,201.19) .. (269,190.14) .. controls (269,185.57) and (270.54,181.35) .. (273.14,177.98) -- (289.1,190.14) -- cycle ; \draw  [color={rgb, 255:red, 245; green, 166; blue, 35 }  ,draw opacity=1 ] (279.84,172.38) .. controls (282.61,170.95) and (285.76,170.14) .. (289.1,170.14) .. controls (300.2,170.14) and (309.2,179.09) .. (309.2,190.14) .. controls (309.2,201.19) and (300.2,210.14) .. (289.1,210.14) .. controls (278,210.14) and (269,201.19) .. (269,190.14) .. controls (269,185.57) and (270.54,181.35) .. (273.14,177.98) ;
    \draw  [draw opacity=0] (273.65,362.93) .. controls (270.75,359.47) and (269,355.01) .. (269,350.14) .. controls (269,339.09) and (278,330.14) .. (289.1,330.14) .. controls (300.2,330.14) and (309.2,339.09) .. (309.2,350.14) .. controls (309.2,361.19) and (300.2,370.14) .. (289.1,370.14) .. controls (286.54,370.14) and (284.09,369.66) .. (281.84,368.79) -- (289.1,350.14) -- cycle ; \draw  [color={rgb, 255:red, 245; green, 166; blue, 35 }  ,draw opacity=1 ] (273.65,362.93) .. controls (270.75,359.47) and (269,355.01) .. (269,350.14) .. controls (269,339.09) and (278,330.14) .. (289.1,330.14) .. controls (300.2,330.14) and (309.2,339.09) .. (309.2,350.14) .. controls (309.2,361.19) and (300.2,370.14) .. (289.1,370.14) .. controls (286.54,370.14) and (284.09,369.66) .. (281.84,368.79) ;
    \draw [color={rgb, 255:red, 139; green, 87; blue, 42 }  ,draw opacity=1 ]   (528.48,252.72) -- (468.22,264.14) ;
    \draw [color={rgb, 255:red, 139; green, 87; blue, 42 }  ,draw opacity=1 ]   (528.73,288.8) -- (468.05,276.44) ;
    \draw  [draw opacity=0] (252.81,118.06) .. controls (280.7,100.32) and (313.8,90.04) .. (349.3,90.04) .. controls (442.87,90.04) and (519.76,161.44) .. (528.48,252.72) -- (349.3,270.04) -- cycle ; \draw  [color={rgb, 255:red, 65; green, 117; blue, 5 }  ,draw opacity=1 ] (252.81,118.06) .. controls (280.7,100.32) and (313.8,90.04) .. (349.3,90.04) .. controls (442.87,90.04) and (519.76,161.44) .. (528.48,252.72) ;
    \draw  [draw opacity=0] (528.72,288.13) .. controls (519.75,379.31) and (442.98,450.54) .. (349.6,450.54) .. controls (315.97,450.54) and (284.5,441.3) .. (257.57,425.22) -- (349.6,270.22) -- cycle ; \draw  [color={rgb, 255:red, 65; green, 117; blue, 5 }  ,draw opacity=1 ] (528.72,288.13) .. controls (519.75,379.31) and (442.98,450.54) .. (349.6,450.54) .. controls (315.97,450.54) and (284.5,441.3) .. (257.57,425.22) ;
    \draw  [draw opacity=0] (230.24,405.04) .. controls (192.87,372.06) and (169.3,323.8) .. (169.3,270.04) .. controls (169.3,218.54) and (190.93,172.09) .. (225.6,139.28) -- (349.3,270.04) -- cycle ; \draw  [color={rgb, 255:red, 65; green, 117; blue, 5 }  ,draw opacity=1 ] (230.24,405.04) .. controls (192.87,372.06) and (169.3,323.8) .. (169.3,270.04) .. controls (169.3,218.54) and (190.93,172.09) .. (225.6,139.28) ;
    \draw [color={rgb, 255:red, 139; green, 87; blue, 42 }  ,draw opacity=1 ]   (225.6,139.28) -- (273.14,177.98) ;
    \draw [color={rgb, 255:red, 139; green, 87; blue, 42 }  ,draw opacity=1 ]   (252.81,118.06) -- (279.84,172.38) ;
    \draw [color={rgb, 255:red, 139; green, 87; blue, 42 }  ,draw opacity=1 ]   (273.65,362.93) -- (230.24,405.04) ;
    \draw [color={rgb, 255:red, 139; green, 87; blue, 42 }  ,draw opacity=1 ]   (281.84,368.79) -- (257.57,425.22) ;
    \draw    (180.4,211.09) -- (349.6,270.22) ;
    \draw  [dash pattern={on 0.84pt off 2.51pt}]  (349.6,270.22) -- (289.1,190.14) ;
    \draw [shift={(289.1,190.14)}, rotate = 232.93] [color={rgb, 255:red, 0; green, 0; blue, 0 }  ][fill={rgb, 255:red, 0; green, 0; blue, 0 }  ][line width=0.75]      (0, 0) circle [x radius= 3.35, y radius= 3.35]   ;
    \draw  [dash pattern={on 0.84pt off 2.51pt}]  (349.3,270.04) -- (449.1,270.04) ;
    \draw [shift={(449.1,270.04)}, rotate = 0] [color={rgb, 255:red, 0; green, 0; blue, 0 }  ][fill={rgb, 255:red, 0; green, 0; blue, 0 }  ][line width=0.75]      (0, 0) circle [x radius= 3.35, y radius= 3.35]   ;
    \draw  [dash pattern={on 0.84pt off 2.51pt}]  (349.2,270.14) -- (289.1,350.14) ;
    \draw [shift={(289.1,350.14)}, rotate = 126.92] [color={rgb, 255:red, 0; green, 0; blue, 0 }  ][fill={rgb, 255:red, 0; green, 0; blue, 0 }  ][line width=0.75]      (0, 0) circle [x radius= 3.35, y radius= 3.35]   ;
    \draw  [draw opacity=0] (548.4,270.8) .. controls (548.36,278.39) and (547.9,285.89) .. (547.04,293.26) -- (348.4,269.89) -- cycle ; \draw   (548.4,270.8) .. controls (548.36,278.39) and (547.9,285.89) .. (547.04,293.26) ;
    \draw  [dash pattern={on 0.84pt off 2.51pt}]  (349.6,270.22) -- (380.8,204.4) ;
    \draw [shift={(380.8,204.4)}, rotate = 340.36] [color={rgb, 255:red, 0; green, 0; blue, 0 }  ][line width=0.75]    (-5.59,0) -- (5.59,0)(0,5.59) -- (0,-5.59)   ;
    \draw  [color={rgb, 255:red, 0; green, 0; blue, 0 }  ,draw opacity=0 ][fill={rgb, 255:red, 0; green, 0; blue, 0 }  ,fill opacity=0.2 ] (367.21,270.04) .. controls (367.21,224.81) and (403.87,188.15) .. (449.1,188.15) .. controls (494.33,188.15) and (530.99,224.81) .. (530.99,270.04) .. controls (530.99,315.27) and (494.33,351.93) .. (449.1,351.93) .. controls (403.87,351.93) and (367.21,315.27) .. (367.21,270.04) -- cycle ;
    \draw  [dash pattern={on 0.84pt off 2.51pt}]  (348.4,269.89) -- (349.2,270.14) ;
    \draw [shift={(349.2,270.14)}, rotate = 62.35] [color={rgb, 255:red, 0; green, 0; blue, 0 }  ][line width=0.75]    (-5.59,0) -- (5.59,0)(0,5.59) -- (0,-5.59)   ;
    \draw  [dash pattern={on 0.84pt off 2.51pt}]  (463.8,477.5) -- (494.8,477.5) ;
    \draw [shift={(494.8,477.5)}, rotate = 45] [color={rgb, 255:red, 0; green, 0; blue, 0 }  ][line width=0.75]    (-5.59,0) -- (5.59,0)(0,5.59) -- (0,-5.59)   ;
    \draw  [dash pattern={on 0.84pt off 2.51pt}]  (463.8,455.5) -- (494.8,455.5) ;
    \draw [shift={(494.8,455.5)}, rotate = 0] [color={rgb, 255:red, 0; green, 0; blue, 0 }  ][fill={rgb, 255:red, 0; green, 0; blue, 0 }  ][line width=0.75]      (0, 0) circle [x radius= 3.35, y radius= 3.35]   ;
    \draw    (406.9,339.75) -- (449.1,270.04) ;
    \draw  [dash pattern={on 0.84pt off 2.51pt}]  (349.6,270.22) -- (362.4,310.9) ;
    \draw [shift={(362.4,310.9)}, rotate = 117.54] [color={rgb, 255:red, 0; green, 0; blue, 0 }  ][line width=0.75]    (-5.59,0) -- (5.59,0)(0,5.59) -- (0,-5.59)   ;

    \draw (253,241.99) node [anchor=north west][inner sep=0.75pt]   [align=left] {$\displaystyle R_{0}$};
    \draw (440.06,250.93) node [anchor=north west][inner sep=0.75pt]  [rotate=-0.34] [align=left] {$\displaystyle \rho _{1}$};
    \draw (280.06,171.93) node [anchor=north west][inner sep=0.75pt]  [rotate=-0.34] [align=left] {$\displaystyle \rho _{2}$};
    \draw (277.06,329.93) node [anchor=north west][inner sep=0.75pt]  [rotate=-0.34] [align=left] {$\displaystyle \rho _{3}$};
    \draw (485,132.39) node [anchor=north west][inner sep=0.75pt]  [color={rgb, 255:red, 65; green, 117; blue, 5 }  ,opacity=1 ]  {$\mathcal{B}$};
    \draw (550.4,274.2) node [anchor=north west][inner sep=0.75pt]    {$\varphi $};
    \draw (503,470) node [anchor=north west][inner sep=0.75pt]  [font=\small] [align=left] {Other dominant singular points};
    \draw (504,447) node [anchor=north west][inner sep=0.75pt]  [font=\small] [align=left] {Dominant singularities of $\displaystyle f$};
    \draw (426.9,305.4) node [anchor=north west][inner sep=0.75pt]    {$R_{1}$};

    \end{tikzpicture}
    \caption{Choice of $R_1$}
    \label{fig:R_1}
\end{figure}
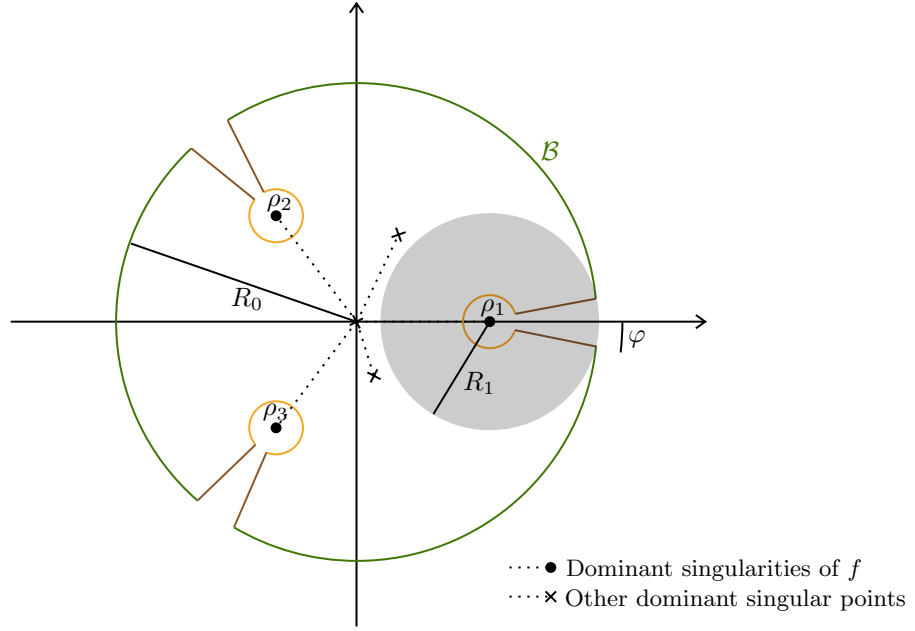

We now examine the term
\[
    \frac{1}{2\pi \i} \int_{\mathcal{S}_{\rho}(n) + \mathcal{L}_{\rho}(n)} \frac{g_{\rho}(z)}{z^{n+1}} \,dz = \frac{1}{2\pi \i} \sum_{j=1}^q c_j \int_{\mathcal{S}_{\rho}(n) + \mathcal{L}_{\rho}(n)} \frac{g_{\rho, j}(z)}{z^{n+1}}
\]
in the decomposition~\eqref{eq:decomposition}, focusing on the summand of index~$j$.
Recall from~\eqref{eq:local_error_term} that we have
\[
g_{\rho, j}(z) = (1 - z/\rho)^{\nu + r} \sum_{k = 0}^{\kappa} h_{k}(z) \log^k \left(\frac{1}{1 - z/\rho}\right),
\]
where $\nu = \nu_{\rho, j}$, $r = r_{\rho, j}$, and $h_k(z) = h_{j, k}(z)$
is analytic at $z = \rho$.
We consider $\rho$ and $j$ to be fixed
in this section and omit them in notation whenever the context is clear.

The first step is to provide an upper bound on $|h_k(z)|$ for $k = 0, \ldots, \kappa$ that is valid on the paths $\mathcal{S}_{\rho}(n)$ and $\mathcal{L}_{\rho}(n)$.
Recall from Section~\ref{subsection:express_integral} that both $\mathcal{S}_{\rho}(n)$ and $\mathcal{L}_{\rho}(n)$ are contained in the disk $\{z : |z - \rho| < R_1\}$. 
Since $R_1 < \min_{\rho_1 \in \domsing, \rho_2 \in \Xi} |\rho_1 - \rho_2|$, it follows that
$\rho$ is the only singular point in the disk 
$\{ z : |z - \rho| < R_1\}$ (see Figure~\ref{fig:R_1}).
Thus it suffices to find real numbers $b_0, \ldots, b_{\kappa}$ such that
\begin{equation}\label{eq:def_b_i}
|h_k(z)| \leq b_k \qquad \text{for all } z \text{ with } |z - \rho| < R_1 \text{ and } k = 0, \ldots, \kappa.
\end{equation}
Our method for computing $b_0, \ldots, b_{\kappa}$ is a generalization of \cite[Lemma~3.5]{MelczerMezzarobba2022}.
We compute some initial terms of the Taylor expansion of each $h_{k}(z)$
(in addition to those already collected in~$\ell_{\rho,j}$)
with rigorous error bounds, until we can apply
\cite[Algorithm~6.11]{mezzarobba2019truncation} to obtain a bound on the tail.
(It is useful in practice to include a few more terms than strictly necessary
in the explicitly computed part in order to limit overestimation.)
The following result is a reformulation
of \cite[Proposition~6.12]{mezzarobba2019truncation},
in slightly weakened form to avoid introducing unnecessary notation.

\begin{proposition}\label{prop:diffopbound}
  Let~$\mathcal L$ denote the operator obtained from~$\diffop$ by the change
  of independent variable~$z = \rho + z$.
  In the notation of Proposition~\ref{prop:singFuschs},
  let~$E$ be the set of exponents~$\nu_{j'}$, for $1 \leq j' \leq q$,
  such that $\nu_{j'} - \nu_j \in \Z$.
  Let~$\lambda$ be the element of~$E$ of minimum real part,
  and let $\delta = \nu_j - \lambda$.

  Given $\mathcal L$, $\lambda$,
  an integer $N \geq \max \bigl(1, \max_{\nu' \in E} (\nu' - \lambda) \bigr)$,
  and the coefficients~$d_{i,k,j}$
  in~\eqref{eq:sol_basis} for $0 \leq i < N - \delta$,
  Algorithm~6.11 in \cite{mezzarobba2019truncation}
  computes two rational functions $G(z)$ and $H(z)$ admitting
  power series expansions at~$0$ with nonnegative coefficients
  such that
  \[
    \abs{d_{i,k,j}}
    \leq \frac{1}{k!}
      [z^{\delta + i}] \left(
        z^N \, G(z) \int_0^z H(w) \, dw
      \right)
  \]
  for all $i \geq N - \delta$ and $k \leq \kappa_j$.
\end{proposition}

As discussed in~\cite{mezzarobba2019truncation},
by running the algorithm at sufficient precision,
the radii of convergence of $G$~and~$H$ can be made arbitrarily close to the
distance from~$\rho$ to the nearest other singular point of~$\diffop$ while
keeping the coefficients of $G$~and~$H$ bounded.
In particular, the radii can be made larger than~$R_1$.
It follows that one can take
\begin{equation} \label{eq:choice_b_i}
  b_k = \frac{1}{k!} R_1^{N - \delta + 1} \, G(R_1) \, H(R_1)< \infty
\end{equation}
in~\eqref{eq:def_b_i}. See also \cite[Section~8.1]{mezzarobba2019truncation} 
for a slightly tighter bound.

These bounds on $|h_k(z)|$ allow us to bound the integrals of $g_{\rho, j}(z)$ over $\mathcal{S}_{\rho}(n)$ and $\mathcal{L}_{\rho}(n)$.
Define the polynomial
\begin{equation}\label{eq:def_B}
    B(z) = b_0 + b_1 z + \cdots + b_{\kappa} z^{\kappa}.
\end{equation}

\begin{proposition}\label{prop:bound_S}
For all $n \geq n_0$, 
\[
    \left| \frac{1}{2\pi \i} \int_{\mathcal{S}_{\rho}(n)} \frac{g_{\rho, j}(z)}{z^{n+1}} \,dz \right|
    \leq |\rho|^{-n} n^{-\Rel(\nu) - 1 -r}
      \cdot \frac
        {e^{\pi \abs{\Im \nu}}}
        {\left( 1 - 1/n_0 \right)^{n_0 + 1}}
      \cdot B(\pi + \log n).
\]
\end{proposition}

\begin{proof}
This is a generalization of \cite[Proposition 3.7] {MelczerMezzarobba2022}.
Parametrize $z \in \mathcal{S}_{\rho}(n)$ by $z = \rho + \rho e^{\i\theta}/n$.
Then $|z| \geq |\rho|(1 - 1/n)$ and we have
\begin{align*}
\left| \log \frac{1}{1 - z/\rho} \right|
  &= |\log (- e^{-\i\theta}) - \log n| \leq \pi + \log n, \\
\left| (1 - z/\rho)^{\nu + r} \right|
  &= \left| \left( \frac{- e^{\i \theta}}{n} \right)^{\nu + r} \right|
  \leq e^{\pi \abs{\Im \nu}} n^{-\Rel \nu - r}.
\end{align*}
Since $|z/\rho| < R_1$ for $z \in \mathcal{S}_{\rho}(n)$,
the function~$\abs{h_k}$ is bounded by~$b_k$ on $\mathcal{S}_{\rho}(n)$.
Therefore,
\[
  \abs{g_{\rho, j}(z)}
  = \left| (z-\rho)^{\nu+r} \right|
    \sum_{k=0}^\kappa
      \abs{h_k(z)} \left| \log^k \frac{1}{1-z/\rho} \right|
  \leq e^{\pi \abs{\Im \nu}} n^{-\Rel \nu - r} B(\pi + \log n)
\]
for all $z \in \mathcal{S}_{\rho}(n)$.
The previous inequalities combine to give
\begin{align*}
    \left| \frac{1}{2\pi \i} \int_{\mathcal{S}_{\rho}(n)} \frac{g_{\rho, j}(z)}{z^{n+1}} \,dz \right| 
    & \leq \frac{\operatorname{length}(\mathcal{S}_{\rho}(n))}{2 \pi}
      \frac{\sup_{z \in \mathcal{S}_{\rho}(n)} \left|g_{\rho, j}(z)\right|}
           {|z|^{n+1}} \\
    & \leq \frac{\rho}{n}
      \frac{n^{-\Rel \nu - r} \cdot e^{\pi \abs{\Im \nu}} \cdot
            B(\pi + \log n)}
           {|\rho|^{n+1}(1 - 1/n)^{n+1}} \\
    & \leq |\rho|^{-n} n^{-\Rel(\nu) - 1 - r}
      \cdot \frac{e^{\pi \abs{\Im \nu}}}{\left( 1 - 1/n_0 \right)^{n_0 + 1}}
      \cdot B(\pi + \log n).
    \qedhere
\end{align*}
\end{proof}

We now consider the integral over~$\mathcal L_\rho(n)$.
In order to treat the case $\Rel(\nu) + r > 0$, we use the following lemma.

\begin{lemma}\label{lemma:int_L}
If $\beta > 0$ and $s > 2$ then, for all $n > s\beta$ and $x > 0$, 
\begin{equation}
    x^{\beta} \leq \left( \frac{(s-2)e}{2s\beta} \right)^{\beta} \left( 1+\frac{x}{n} \right)^{n/2}.
\end{equation}
\end{lemma}
\begin{proof}
Let
\[
\psi(x) := x^{-\beta} \left( 1 +\frac{x}{n}\right)^{n/2}.
\]
Solving $(\log \psi(x))' = 0$ gives the minimum as $x = (2n\beta)/(n - 2\beta)$, so
\[
  \psi(x)
  \geq \psi\left(\frac{2n\beta}{n - 2\beta}\right)
  = \frac{\left( 1 + \frac{2\beta}{n - 2\beta} \right)^{n/2}}
         {\left(\frac{2n\beta}{n - 2\beta}\right)^{\beta}}.
\]
Using the inequality $\left(1-\frac{1}{m}\right)^m < e^{-1}$, and 
substituting $m=\frac{n}{2\beta}>1$, gives
\[
\left( 1 + \frac{2\beta}{n - 2\beta} \right)^{n/2} > e^{\beta},
\]
so $n >s\beta$ implies
\[
\left(\frac{2n\beta}{n-2\beta}\right)^{\beta} < \left(\frac{2s\beta}{s-2}\right)^{\beta}.
\]
Combining the bounds for the numerator and the denominator yields the desired inequality.
\end{proof}

\begin{proposition}\label{prop:bound_L}
For all $n \geq s|\nu + r|$ and small enough $\varphi$,
\[
    \left| \frac{1}{2\pi \i} \int_{\mathcal{L}_{\rho}(n)} \frac{g_{\rho, j}(z)}{z^{n+1}} \,dz \right|
    \leq |\rho|^{-n} n^{-\Rel(\nu) - 1 -r}
      \frac{C(\Rel(\nu) + r, \varphi)}{\pi}
      e^{\pi \abs{\Im \nu}}
      B(\pi + \log n),
\]
where
\[
C(\beta, \varphi) = 
\begin{cases}
\frac{1}{\cos \varphi} &: \beta \leq 0, \\
\frac{2}{\cos^{\beta + 1} \varphi} \left( \frac{(s-2)e}{2s\beta} \right)^{\beta} &: \beta > 0.
\end{cases}
\]
\end{proposition}

\begin{proof}
This is a generalization of \cite[Proposition 3.8] {MelczerMezzarobba2022}.
The integral over the upper part of $\mathcal{L}_{\rho}(n)$ equals
\begin{align*}
L_{+}(n) &= \frac{1}{2\pi \i} \int_{\rho(1+e^{\i\varphi}/n)}^{\rho(1+Ee^{\i\varphi})} \frac{g_{\rho, j}(z)}{z^{n+1}} \,dz \\
&= \frac{1}{2\pi \i} \sum_{k = 0}^{\kappa} \int_{\rho(1+e^{\i\varphi}/n)}^{\rho(1+Ee^{\i\varphi})} (1-z/\rho)^{\nu + r}\frac{h_{k}(z) \log ^k (\frac{1}{1 - z/\rho})}{z^{n+1}} \,dz
\end{align*}
for some $E < R_1$.
Substituting $z = \rho(1+e^{\i\varphi}t/n)$ yields,
when $\varphi$ is small enough,
\begin{align*}
    \left| L_{+}(n) \right| & = \frac{1}{2\pi} \left| \sum_{k = 0}^{\kappa} \int_{1}^{En} (-e^{\i\varphi}t/n)^{\nu+r} \frac{h_{k}(\rho(1+e^{\i\varphi}t/n)) \log^{k}(-e^{-\i\varphi} n/t)}{\rho^{n+1} (1+e^{\i\varphi}t/n)^{n+1}} \frac{\rho e^{\i \varphi}}{n} \,dt \right| \\
    & \leq \frac{|\rho|^{-n} n^{-\Rel(\nu)-r-1}e^{(\pi - \varphi) \Im \nu}}{2\pi} \\
    & \qquad \cdot \sum_{k = 0}^{\kappa} b_k \int_{1}^{\infty} \left| \i(\pi - \varphi) + \log(n/t) \right|^{k} \cdot t^{\Rel(\nu) + r} \left( 1 + \frac{t \cos \varphi}{n} \right)^{-n-1} \,dt \\
    & \leq |\rho|^{-n} n^{-\Rel(\nu) - 1}
      \cdot \frac{e^{\pi \abs{\Im \nu}}}{2\pi} \cdot n^{-r} 
    \cdot B(\pi + \log n) \\
    & \qquad \cdot \int_{1}^{\infty} t^{\Rel(\nu) + r} \left( 1 + \frac{t \cos \varphi}{n} \right)^{-n-1} \,dt.
\end{align*}

Let
\[
    I_n(\beta) = \int_{1}^{\infty} t^{\beta}
                \left( 1 + \frac{t \cos \varphi}{n} \right)^{-n-1} \,dt.
\]
When $\beta \leq 0$, as $n \geq 1$, we have
\[
    I_n(\beta)
    \leq \int_{1}^{\infty} \left( 1 + \frac{t \cos \varphi}{n} \right)^{-n-1} \,dt \\
    = \frac{1}{\cos \varphi} \left( 1 + \frac{\cos \varphi}{n} \right)^{-n} \\
    \leq \frac{1}{\cos \varphi}.
\]
On the other hand, when $\beta > 0$ then an application of
Lemma~\ref{lemma:int_L} with $x = t\cos \varphi$ implies
\begin{align*}
    I_n(\beta)
    & = \frac{1}{\cos^{\beta} \varphi}\int_{1}^{\infty} (t\cos \varphi)^{\beta} \left( 1 + \frac{t \cos \varphi}{n} \right)^{-n-1} \,dt \\
    & \leq \frac{1}{\cos^{\beta} \varphi} \left( \frac{(s-2)e}{2s\beta} \right)^{\beta} \int_{1}^{\infty} \left( 1 + \frac{t \cos \varphi}{n} \right)^{-n-1 + n/2} \,dt ,
\end{align*}
where
\[
    \int_{1}^{\infty} \left( 1 + \frac{t \cos \varphi}{n} \right)^{-n-1 + n/2} \,dt
    = \frac{2}{\cos \varphi} \left( 1 + \frac{\cos \varphi}{n} \right)^{-n/2}
    \leq \frac{2}{\cos \varphi}.
\]
In both cases we conclude that
$I_n(\beta) \leq C(\beta, \varphi)$,
and therefore
\[
  \abs{L_+(n)}
  \leq |\rho|^{-n} n^{-\Rel(\nu) - 1 -r}
    \frac{C(\Rel(\nu) + r, \varphi)}{2 \pi}
    e^{\pi \abs{\Im \nu}}
    B(\pi + \log n).
\]

The same reasoning applies to the integral over the other part of $\mathcal{L}_{\rho}(n)$, replacing $\varphi$ by $2\pi-\varphi$, and their sum yields the desired bound.
\end{proof}

Letting $\varphi \to 0$ in Proposition~\ref{prop:bound_L} gives the following.

\begin{corollary} \label{cor:bound_L}
For all $n \geq s|\nu + r|$,
\begin{multline}
    \lim_{\varphi \to 0} \left|
      \frac{1}{2\pi \i}
      \int_{\mathcal{L}_{\rho}(n)} \frac{g_{\rho, j}(z)}{z^{n+1}} \,dz \right| \\
    \leq |\rho|^{-n} n^{-\Rel(\nu) - 1 -r}
      \cdot \frac{C(\Rel(\nu) + r)}{\pi} e^{\pi \abs{\Im \nu}}
      \cdot B(\pi + \log n),
\end{multline}
where
\[
C(\beta) = 
\begin{cases}
1 &: \beta \leq 0, \\
2 \left( \frac{(s-2)e}{2s\beta} \right)^{\beta} &: \beta > 0.
\end{cases}
\]
\end{corollary}

\begin{remark} \label{rk:error_exact}
  The error term computed here may not be of the same order of magnitude as
  that from the previous section.
  In particular, our approach overestimates the order of
  magnitude of the error by a factor of $\log n$ when $\nu + r$ is a
  nonnegative integer and $\kappa \geq 1$.
  This is no significant limitation since one can always increase the expansion
  order by one unit to recover an error term of the ``correct'' form.
  It is useful, however, to treat the case where both
  $\nu + r \in \Z_{\geq 0}$ and $\kappa = 0$
  (where $g(z)$ is analytic at $\rho$)
  specially in order to avoid artificially introducing error terms in
  terminating expansions in powers of~$n$.
\end{remark}

\section{Computing The Global Error Term and Combining the Bounds}
\label{sec:global}

\subsection{The Global Error Term}
\label{subsection:bound_I}

\begin{figure}
    \centering
    \includegraphics[width=0.5\textwidth,height=0.5\textheight,keepaspectratio]{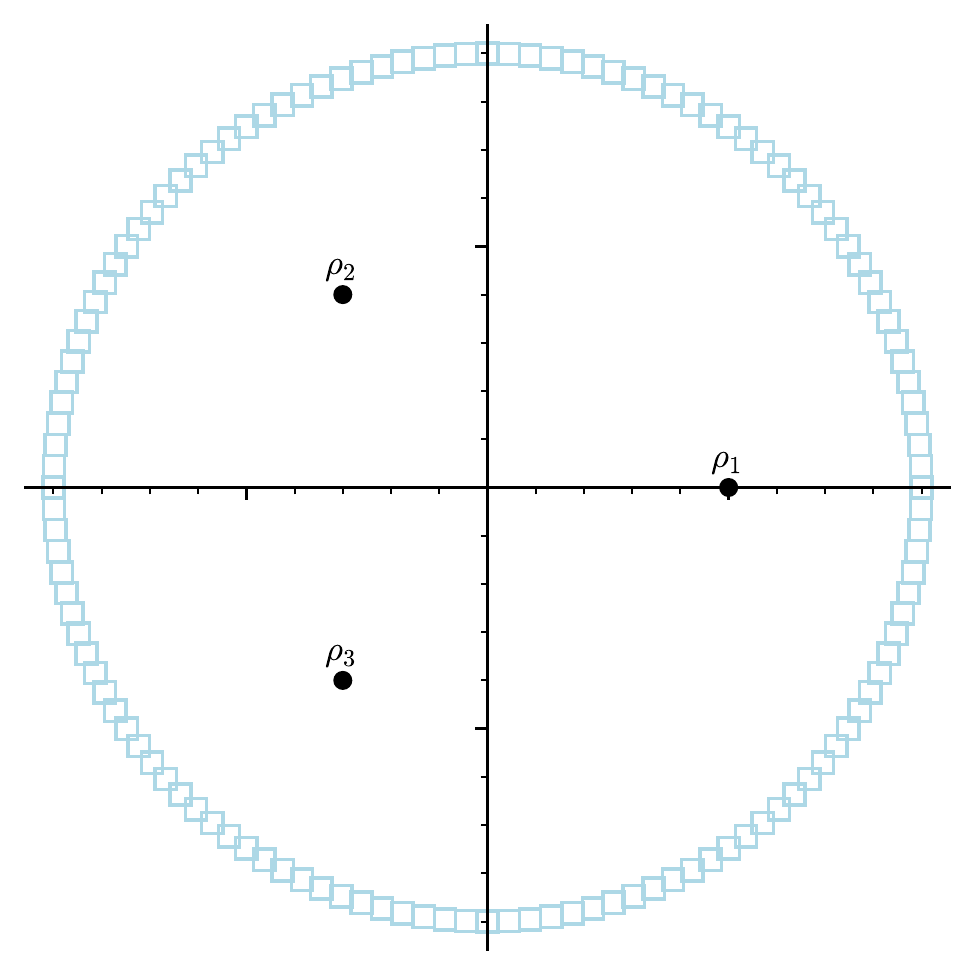}
    \caption{Covering the big circle with squares}
    \label{fig:cover_squares}
\end{figure}
\begin{figure}
    \centering
    \begin{tikzpicture}[x=0.5pt,y=0.5pt,yscale=-1,xscale=1]

    \draw[use as bounding box]  (81.62,259.04) -- (601.62,259.04)(340.5,19) -- (340.5,489) (594.62,254.04) -- (601.62,259.04) -- (594.62,264.04) (335.5,26) -- (340.5,19) -- (345.5,26)  ;
    \draw  [draw opacity=0] (459.12,265.81) .. controls (456.35,273.53) and (448.97,279.04) .. (440.3,279.04) .. controls (429.25,279.04) and (420.3,270.09) .. (420.3,259.04) .. controls (420.3,247.99) and (429.25,239.04) .. (440.3,239.04) .. controls (449.18,239.04) and (456.71,244.83) .. (459.32,252.84) -- (440.3,259.04) -- cycle ; \draw  [color={rgb, 255:red, 245; green, 166; blue, 35 }  ,draw opacity=1 ] (459.12,265.81) .. controls (456.35,273.53) and (448.97,279.04) .. (440.3,279.04) .. controls (429.25,279.04) and (420.3,270.09) .. (420.3,259.04) .. controls (420.3,247.99) and (429.25,239.04) .. (440.3,239.04) .. controls (449.18,239.04) and (456.71,244.83) .. (459.32,252.84) ;
    \draw  [draw opacity=0] (272.66,160.63) .. controls (275.02,159.67) and (277.6,159.14) .. (280.3,159.14) .. controls (291.4,159.14) and (300.4,168.09) .. (300.4,179.14) .. controls (300.4,190.19) and (291.4,199.14) .. (280.3,199.14) .. controls (269.2,199.14) and (260.2,190.19) .. (260.2,179.14) .. controls (260.2,174.11) and (262.06,169.52) .. (265.14,166.01) -- (280.3,179.14) -- cycle ; \draw  [color={rgb, 255:red, 245; green, 166; blue, 35 }  ,draw opacity=1 ] (272.66,160.63) .. controls (275.02,159.67) and (277.6,159.14) .. (280.3,159.14) .. controls (291.4,159.14) and (300.4,168.09) .. (300.4,179.14) .. controls (300.4,190.19) and (291.4,199.14) .. (280.3,199.14) .. controls (269.2,199.14) and (260.2,190.19) .. (260.2,179.14) .. controls (260.2,174.11) and (262.06,169.52) .. (265.14,166.01) ;
    \draw  [draw opacity=0] (265.39,352.55) .. controls (262.16,349) and (260.2,344.3) .. (260.2,339.14) .. controls (260.2,328.09) and (269.2,319.14) .. (280.3,319.14) .. controls (291.4,319.14) and (300.4,328.09) .. (300.4,339.14) .. controls (300.4,350.19) and (291.4,359.14) .. (280.3,359.14) .. controls (277.68,359.14) and (275.18,358.64) .. (272.89,357.74) -- (280.3,339.14) -- cycle ; \draw  [color={rgb, 255:red, 245; green, 166; blue, 35 }  ,draw opacity=1 ] (265.39,352.55) .. controls (262.16,349) and (260.2,344.3) .. (260.2,339.14) .. controls (260.2,328.09) and (269.2,319.14) .. (280.3,319.14) .. controls (291.4,319.14) and (300.4,328.09) .. (300.4,339.14) .. controls (300.4,350.19) and (291.4,359.14) .. (280.3,359.14) .. controls (277.68,359.14) and (275.18,358.64) .. (272.89,357.74) ;
    \draw [color={rgb, 255:red, 139; green, 87; blue, 42 }  ,draw opacity=1 ]   (519.68,241.72) -- (459.32,252.84) ;
    \draw [color={rgb, 255:red, 139; green, 87; blue, 42 }  ,draw opacity=1 ]   (519.93,277.8) -- (459.12,265.81) ;
    \draw  [draw opacity=0] (244.01,107.06) .. controls (271.9,89.32) and (305,79.04) .. (340.5,79.04) .. controls (434.07,79.04) and (510.96,150.44) .. (519.68,241.72) -- (340.5,259.04) -- cycle ; \draw  [color={rgb, 255:red, 65; green, 117; blue, 5 }  ,draw opacity=1 ] (244.01,107.06) .. controls (271.9,89.32) and (305,79.04) .. (340.5,79.04) .. controls (434.07,79.04) and (510.96,150.44) .. (519.68,241.72) ;
    \draw  [draw opacity=0] (519.92,277.13) .. controls (510.95,368.31) and (434.18,439.54) .. (340.8,439.54) .. controls (307.17,439.54) and (275.7,430.3) .. (248.77,414.22) -- (340.8,259.22) -- cycle ; \draw  [color={rgb, 255:red, 65; green, 117; blue, 5 }  ,draw opacity=1 ] (519.92,277.13) .. controls (510.95,368.31) and (434.18,439.54) .. (340.8,439.54) .. controls (307.17,439.54) and (275.7,430.3) .. (248.77,414.22) ;
    \draw  [draw opacity=0] (221.44,394.04) .. controls (184.07,361.06) and (160.5,312.8) .. (160.5,259.04) .. controls (160.5,207.54) and (182.13,161.09) .. (216.8,128.28) -- (340.5,259.04) -- cycle ; \draw  [color={rgb, 255:red, 65; green, 117; blue, 5 }  ,draw opacity=1 ] (221.44,394.04) .. controls (184.07,361.06) and (160.5,312.8) .. (160.5,259.04) .. controls (160.5,207.54) and (182.13,161.09) .. (216.8,128.28) ;
    \draw [color={rgb, 255:red, 139; green, 87; blue, 42 }  ,draw opacity=1 ]   (216.8,128.28) -- (265.14,166.01) ;
    \draw [color={rgb, 255:red, 139; green, 87; blue, 42 }  ,draw opacity=1 ]   (244.01,107.06) -- (272.66,160.63) ;
    \draw [color={rgb, 255:red, 139; green, 87; blue, 42 }  ,draw opacity=1 ]   (265.39,352.55) -- (221.44,394.04) ;
    \draw [color={rgb, 255:red, 139; green, 87; blue, 42 }  ,draw opacity=1 ]   (272.89,357.74) -- (248.77,414.22) ;
    \draw  [dash pattern={on 0.84pt off 2.51pt}]  (340.8,259.22) -- (280.3,179.14) ;
    \draw [shift={(280.3,179.14)}, rotate = 232.93] [color={rgb, 255:red, 0; green, 0; blue, 0 }  ][fill={rgb, 255:red, 0; green, 0; blue, 0 }  ][line width=0.75]      (0, 0) circle [x radius= 3.35, y radius= 3.35]   ;
    \draw  [dash pattern={on 0.84pt off 2.51pt}]  (340.5,259.04) -- (440.3,259.04) ;
    \draw [shift={(440.3,259.04)}, rotate = 0] [color={rgb, 255:red, 0; green, 0; blue, 0 }  ][fill={rgb, 255:red, 0; green, 0; blue, 0 }  ][line width=0.75]      (0, 0) circle [x radius= 3.35, y radius= 3.35]   ;
    \draw  [dash pattern={on 0.84pt off 2.51pt}]  (340.4,259.14) -- (280.3,339.14) ;
    \draw [shift={(280.3,339.14)}, rotate = 126.92] [color={rgb, 255:red, 0; green, 0; blue, 0 }  ][fill={rgb, 255:red, 0; green, 0; blue, 0 }  ][line width=0.75]      (0, 0) circle [x radius= 3.35, y radius= 3.35]   ;
    \draw  [dash pattern={on 0.84pt off 2.51pt}]  (340.8,259.22) -- (372.4,193.05) ;
    \draw [shift={(372.4,193.05)}, rotate = 340.53] [color={rgb, 255:red, 0; green, 0; blue, 0 }  ][line width=0.75]    (-5.59,0) -- (5.59,0)(0,5.59) -- (0,-5.59)   ;
    \draw  [dash pattern={on 0.84pt off 2.51pt}]  (339.6,258.89) -- (340.4,259.14) ;
    \draw [shift={(340.4,259.14)}, rotate = 62.35] [color={rgb, 255:red, 0; green, 0; blue, 0 }  ][line width=0.75]    (-5.59,0) -- (5.59,0)(0,5.59) -- (0,-5.59)   ;
    \draw  [dash pattern={on 0.84pt off 2.51pt}]  (340.8,259.22) -- (353.6,299.9) ;
    \draw [shift={(353.6,299.9)}, rotate = 117.54] [color={rgb, 255:red, 0; green, 0; blue, 0 }  ][line width=0.75]    (-5.59,0) -- (5.59,0)(0,5.59) -- (0,-5.59)   ;
    \draw    (340.8,259.22) -- (365.4,208.1) ;
    \draw    (365.4,208.1) -- (388.4,200.1) ;
    \draw    (388.4,200.1) -- (378.4,180.1) ;
    \draw    (378.4,180.1) -- (418.4,97.1) ;
    \draw [shift={(398.4,138.6)}, rotate = 115.73] [fill={rgb, 255:red, 0; green, 0; blue, 0 }  ][line width=0.08]  [draw opacity=0] (8.93,-4.29) -- (0,0) -- (8.93,4.29) -- cycle    ;
    \draw  [color={rgb, 255:red, 74; green, 144; blue, 226 }  ,draw opacity=1 ] (455,120) -- (464.3,120) -- (464.3,129.55) -- (455,129.55) -- cycle ;
    \draw    (418.4,97.1) .. controls (425.89,99.55) and (440.47,107.48) .. (457.51,122.82) ;
    \draw [shift={(459.65,124.78)}, rotate = 222.81] [fill={rgb, 255:red, 0; green, 0; blue, 0 }  ][line width=0.08]  [draw opacity=0] (8.93,-4.29) -- (0,0) -- (8.93,4.29) -- cycle    ;
    \draw    (461.6,442.5) -- (487.6,442.5) ;
    \draw [shift={(490.6,442.5)}, rotate = 180] [fill={rgb, 255:red, 0; green, 0; blue, 0 }  ][line width=0.08]  [draw opacity=0] (8.93,-4.29) -- (0,0) -- (8.93,4.29) -- cycle    ;
    \draw  [color={rgb, 255:red, 74; green, 144; blue, 226 }  ,draw opacity=1 ] (479,455) -- (487.6,455) -- (487.6,463.5) -- (479,463.5) -- cycle ;
    \draw  [dash pattern={on 0.84pt off 2.51pt}]  (460.8,475.85) -- (477.8,475.85) ;
    \draw [shift={(477.8,475.85)}, rotate = 45] [color={rgb, 255:red, 0; green, 0; blue, 0 }  ][line width=0.75]    (-5.59,0) -- (5.59,0)(0,5.59) -- (0,-5.59)   ;
    \draw    (418.4,97.1) .. controls (387.6,78.5) and (340.6,76.5) .. (307.6,82.5) ;
    \draw [shift={(363.82,80.57)}, rotate = 6.91] [fill={rgb, 255:red, 0; green, 0; blue, 0 }  ][line width=0.08]  [draw opacity=0] (8.93,-4.29) -- (0,0) -- (8.93,4.29) -- cycle    ;
    \draw  [color={rgb, 255:red, 74; green, 144; blue, 226 }  ,draw opacity=1 ] (302.95,77.72) -- (312.25,77.72) -- (312.25,87.28) -- (302.95,87.28) -- cycle ;
    \draw    (339.6,258.89) -- (429,416.3) ;
    \draw [shift={(384.3,337.59)}, rotate = 240.41] [fill={rgb, 255:red, 0; green, 0; blue, 0 }  ][line width=0.08]  [draw opacity=0] (8.93,-4.29) -- (0,0) -- (8.93,4.29) -- cycle    ;
    \draw    (429,416.3) .. controls (462,395.65) and (476,380.65) .. (489,362.15) ;
    \draw [shift={(462.16,392.33)}, rotate = 139.94] [fill={rgb, 255:red, 0; green, 0; blue, 0 }  ][line width=0.08]  [draw opacity=0] (8.93,-4.29) -- (0,0) -- (8.93,4.29) -- cycle    ;
    \draw  [color={rgb, 255:red, 74; green, 144; blue, 226 }  ,draw opacity=1 ] (484.35,357.38) -- (493.65,357.38) -- (493.65,366.93) -- (484.35,366.93) -- cycle ;
    \draw  [dash pattern={on 0.84pt off 2.51pt}]  (486.8,475.85) -- (505.3,475.85) ;
    \draw [shift={(505.3,475.85)}, rotate = 0] [color={rgb, 255:red, 0; green, 0; blue, 0 }  ][fill={rgb, 255:red, 0; green, 0; blue, 0 }  ][line width=0.75]      (0, 0) circle [x radius= 3.35, y radius= 3.35]   ;

    \draw (431.26,239.93) node [anchor=north west][inner sep=0.75pt]  [rotate=-0.34] [align=left] {$\displaystyle \rho _{1}$};
    \draw (270.26,159.93) node [anchor=north west][inner sep=0.75pt]  [rotate=-0.34] [align=left] {$\displaystyle \rho _{2}$};
    \draw (270.26,319.93) node [anchor=north west][inner sep=0.75pt]  [rotate=-0.34] [align=left] {$\displaystyle \rho _{3}$};
    \draw (516,435) node [anchor=north west][inner sep=0.75pt]  [font=\small] [align=left] {Analytic continuation paths};
    \draw (516,452) node [anchor=north west][inner sep=0.75pt]  [font=\small] [align=left] {Areas to be evaluated};
    \draw (506,168.4) node [anchor=north west][inner sep=0.75pt]  [color={rgb, 255:red, 65; green, 117; blue, 5 }  ,opacity=1 ]  {$\mathcal{B}$};
    \draw (516,468.5) node [anchor=north west][inner sep=0.75pt]  [font=\small] [align=left] {Dominant singular points};

    \end{tikzpicture}
    \caption{Analytic continuation path}
    \label{fig:ana_path}
\end{figure}
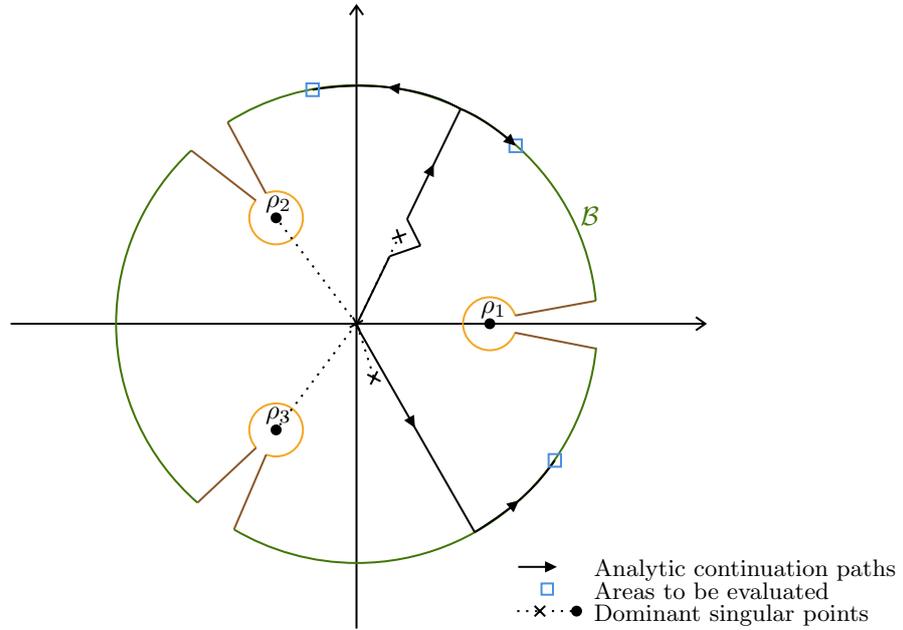

Inequality \eqref{eq:bound_Ic} implies that an upper bound for
$\lim _{\varphi \to 0} |\mI_{\mB}|$ can be determined by finding some $C_{\mB}>0$ such that
\begin{equation}%
  \label{eq:expr-big-circle}
  \Bigl|f(z) - \sum_{\rho \in \domsing} \ell_{\rho}(z)\Bigr| \leq C_{\mB}
  \text{ whenever } |z| = R_0.
\end{equation}
We cover the circle $\{|z| = R_0\}$ with small squares as illustrated
in Figure \ref{fig:cover_squares},
and compute approximations of
$f(z) - \sum_{\rho \in \domsing} \ell_{\rho}(z)$
on each square.
The size of the squares should be small enough so that they do not contain any singularity of $\diffop$.
For definiteness, we take squares of constant side length $\min_{\rho \in \Xi} |R_0 - |\rho||/5$ which guarantees that the squares are sufficiently far away from singular points, helping to obtain bounds of reasonable size.
(If $\min_{\rho \in \Xi} |R_0 - |\rho||/5$ is very small, it is better to use a non-uniform covering to limit the number of small squares to be considered.)

Since we have explicit expressions for $\ell_{\rho}(z)$,
we obtain enclosures of their ranges on each square by evaluations in ball
arithmetic.
It remains to bound~$f(z)$ on each of the small squares.
To do so we use a rigorous numerical solver for D-finite equations.
The procedure is a simpler variant of the one used above to compute the
connection matrices~$\bC_{0 \rightarrow \rho}$ and deduce bounds on
each~$g_\rho(z)$ in the neighborhood of~$\rho$.
We construct paths in~$\Delta_\diffop$ that connect the origin
to each component of~$\mathcal B$,
making small detours to avoid singular points within~$\Delta_\diffop$ if
necessary,
and then cover the corresponding component of~$\mathcal B$
(see Figure~\ref{fig:ana_path}).
By solving the differential equation~$\diffop f = 0$ along these paths,
we compute connection matrices from~$0$ to the centers of the small squares
covering~$\mathcal B$.
Finally, for each square we compute a few initial terms of the series
expansion of~$f$ at the center, evaluate them in interval arithmetic, and bound
the remainder of the series using Proposition~\ref{prop:diffopbound}.
This yields a bound of the form
\[
  |\mI_{\mB}| \leq C_{\mB} R_0^{-n},
\]
where $C_{\mB}$ is the computed bound satisfying~\eqref{eq:expr-big-circle}.

We note in passing that the rigorous numerical solver necessary for following the
paths and computing the connection matrices can itself be realized using
the same approach.
Essentially, one discretizes the integration path and, at each step
$z_n \to z_{n+1}$,
computes approximations of
$f(z_{n+1}), f'(z_{n+1}), \dots, f^{(q-1)}(z_{n+1})$
by summing the Taylor expansion of~$f$ at~$z_n$.
In the case of a D-finite equation, the coefficients of the Taylor series are
easily generated using the associated recurrence.
Computing a partial sum can be done in ball arithmetic,
hence the critical issue for obtaining a rigorous enclosure of the
solution is to bound the tails of each of the series,
for which we can use Proposition~\ref{prop:diffopbound} again
(see~\cite{mezzarobba2019truncation} for details).

\subsection{Combining the bounds}
\label{subsection:assemble}

Now that we have bounds for each component in \eqref{eq:decomposition}, it suffices to
combine and simplify them to match the result given in
Theorem~\ref{thm:mainBound}.

We begin by computing the constant $N_0$ in Theorem \ref{thm:mainBound}.
Firstly, in order for the integration path $\mP(n)$ to be closed for all $n
\geq N_0$, we require that $N_0 \geq N_1$ where $N_1$~is defined
in~\eqref{eq:def_n1}.
Secondly, we need to guarantee the existence of an $s > 2$ such that
$N_0 > s |\alpha|$ for all the exponents $\alpha$ to which we will apply the results of
Sections \ref{sec:explicit_part}~and~\ref{sec:local_error}
(Propositions \ref{prop:coeffasy}~and~\ref{prop:bound_L}).
These exponents are of the form $\nu_{\rho,i} + r$ for some $\rho$, $i$,
and $r \leq r_0$, where $r_0$~is the desired expansion order.
Finally, the statement of the theorem specifies that $N_0 \geq n_0$.
Thus, we set
\begin{equation}\label{eq:def_n0}
    N_2 = \bigl\lceil 2.1 (\max_{\rho, i} \abs{\nu_{\rho, i}} + r_0 + 1) \bigr\rceil
    \qquad\text{and}\qquad
    N_0 = \max\left\{n_0, N_1, N_2 \right\}.
\end{equation}

Recall from~\eqref{eq:decomposition} that we have
\begin{equation}\label{eq:fn_sum_assemble}
f_n = [z^{n}]f(z) = \sum_{\rho \in \domsing} [z^n] \ell_\rho  + \sum_{\rho \in \domsing} \frac{1}{2\pi \i} \int_{\mS_{\rho}(n) + \mL_{\rho}(n)} \frac{g_{\rho}(z)}{z^{n+1}} \, dz + \mI_{\mB},
\end{equation}
and from~\eqref{eq:estimate_form_main} that we are aiming for a bound of
the form
\begin{equation}\label{eq:main_bound_repeated}
    f_n = \sum_{\rho \in \domsing} \rho^{-n} n^{\gamma} \sum_{i = 0}^{m_{\rho}} \sum_{k = 0}^{\kappa} a_{\rho,i,k} \frac{\log^{k} n}{n^{\gamma_{\rho,i}}} + R(n),
    \qquad
    |R(n)| \leq A \, M^{-n} n^{\Rel(\gamma)} \frac{\log^{\kappa} n}{n^{r_0}}.
\end{equation}

In Sections \ref{subsection:bound_l} and \ref{subsection:bound_g} we computed a
rigorous estimate for the first two sums in~\eqref{eq:fn_sum_assemble}.
This estimate can be viewed as a sum of monomials
\begin{equation}\label{eq:form_bound_3}
    \B(a, \varepsilon) \cdot \rho^{-n} n^{\theta} \log^{k} n,
\end{equation}
where $a \in \rholgamma$ and the ball $\B(a, \varepsilon)$ is exact whenever $\Rel(\theta) > \Rel(\gamma) - r_0$.
We call the terms such that $\Rel(\theta) \leq \Rel(\gamma) - r_0$ \emph{error terms},
and those with $\Rel(\theta) < \Rel(\gamma) - r_0$ \emph{secondary error terms}.
In order to simplify the sum to the desired form, we need to identify~$\gamma$ and trim down any secondary error terms that may appear.
Let~$\kappa$ be the maximum value of the parameter~$k$ occurring among all error terms.

When $\Rel(\theta) < \Rel(\gamma) - r_0$, using the assumption $n \geq N_0$, we can replace
a term of the form
$\B(a, \varepsilon) \cdot \rho^{-n} n^{\theta} \log^{k} n$
with
\begin{equation}\label{eq:form_bound_4}
\B\left(0, (a + \varepsilon)\cdot N_0^{\Rel(\theta) - \Rel(\gamma) + r_0} \log^{k - \kappa} N_0\right) \cdot M^{-n} n^{\Rel(\gamma)} \frac{\log^{\kappa} n}{n^{r_0}},
\end{equation}
because
\[
  \B(a, \varepsilon) \cdot  n^{\theta - \Rel(\gamma) + r_0} \log^{k - \kappa} n 
  \subseteq
  \B\left(0, (a + \varepsilon)\cdot N_0^{\Rel(\theta) - \Rel(\gamma) + r_0} \log^{k - \kappa} N_0\right)
\]
when $n > N_0$.
Replacing all secondary error terms in the sum of \eqref{eq:form_bound_3} with \eqref{eq:form_bound_4} standardizes all the error terms to the form $\B(0, a') \cdot M^{-n} n^{\Rel(\gamma) - r_0} \log^{\kappa} n$.

\begin{remark}
  Remarks \ref{rk:coeffasy_exact}~and~\ref{rk:error_exact} imply that, in some
  cases, the sum contains no error terms at all.
  When this happens, and if the next singular points of the differential
  equation by increasing modulus are also regular, one can subtract the sum of
  the corresponding local expansions from~$f(z)$ and iterate the algorithm to
  improve the approximation of~$f_n$ with exponentially smaller terms.
  It can also make sense to do something similar when the constant in the
  combined error term has been verified to be very small but could not be
  checked to be exactly zero for lack of a zero-test for connection
  coefficients.
\end{remark}

In Section \ref{subsection:bound_I} we computed a bound for $\left|\mI_{\mB}\right|$ of the form $C_{\mB} R_0^{-n}$. 
When $n \geq N_0$,
letting $\beta = \Rel(\gamma) - r_0$,
we have
\begin{align*}
R_0^{-n} M^{n} n^{- \beta} \leq A = 
\begin{cases}
e^{\beta} \left( \frac{\beta}{\log(M/R_0)} \right)^{-\beta} &: \beta \leq N_0 \log \frac{M}{R_0}, \\
\left(\frac{M}{R_0}\right)^{N_0} N_0^{-\beta} &: \beta > N_0 \log \frac{M}{R_0}.
\end{cases}
\end{align*}
In both cases, we can absorb $\mI_{\mB}$ into an error term of the form
\[
  \B\left(0, \frac{C_{\mB} \, A}{\log^{\kappa} N_0} \right) \cdot M^{-n} n^{\Rel(\gamma) - r_0} \log^{\kappa} n.
\]

\subsection{Proof of Theorem~\ref{thm:mainBound}}
\label{sec:proofconclusion}

Finally, we recall the statement of Theorem~\ref{thm:mainBound} and conclude its proof.

\theoremstyle{acmplain}
\newtheorem*{mainthm}{Theorem~\ref{thm:mainBound}}
\begin{mainthm}
  \statemaintheorem{eq:estimate_form_main_2}
\end{mainthm}

\begin{proof}
The proof is a matter of checking that
Algorithm~\ref{algo:main} correctly implements the analysis from
the previous sections.
By Definition~\ref{def:landg} and Proposition~\ref{prop:decomposition}, we have
\begin{equation} \label{eq:decomposition_proof}
  [z^{n}]f(z) =
    \sum_{\rho \in \domsing} \sum_{j=1}^q c_{\rho,j} \left(
      [z^n] \ell_{\rho,j}(z)
      + \frac1{2\pi \i} \int_{\mS_{\rho}(n) + \mL_{\rho}(n)} \frac{g_{\rho,j}(z)}{z^{n+1}} \, dz
    \right)
    + \mI_{\mB}
\end{equation}
where
$c_{\rho,1}, \dots, c_{\rho,q}$ are the coordinates of~$f$ in the basis $(y_{\rho,j})_j$,
the functions $\ell_{\rho,j}$ and $g_{\rho,j}$ are defined by \eqref{eq:lgexpansion}--\eqref{eq:local_error_term} in terms of initial coefficients $d_{i,k,j}$ of the local expansion~\eqref{eq:sol_basis} of~$y_{\rho,j}$,
the paths $\mS_\rho(n)$ and $\mL_\rho(n)$ are defined in Section~\ref{sec:decompose} and implicitly depend on~$\varphi$,
and $\mI_{\mB}$ is a quantity, also depending on~$\varphi$, known to satisfy the inequality~\eqref{eq:bound_Ic}.

The algorithm essentially computes the terms of~\eqref{eq:decomposition_proof} one by one.
Fix $\rho$ and $j$ and consider the corresponding terms.

As discussed in Section~\ref{sec:connection}, since the path $0 \to \rho$ is contained in the domain~$\Delta_{\diffop}^0$ and $f$~is analytic on $\Delta_{\diffop}^0$, the coefficients~$\bc_{\rho}$ computed by steps \eqref{step:sol_bas} and \eqref{step:connection} agree with those appearing in~\eqref{eq:decomposition_proof}.
The parameters defining $\ell_{\rho,j}$ and $g_{\rho,j}$ are computed at steps
\eqref{step:localbasis},
\ref{step:localparam},
and~\ref{step:localbasisexplicit},
by direct application of the definitions.
In particular, after step~\ref{step:localbasisexplicit} at each loop iteration, we have (\emph{cf.}~\eqref{eq:sum_l_decomp})
\begin{equation} \label{eq:sum_l_decomp_bis}
  [z^n]\ell_{\rho, j}(z) =
  \sum_{i=0}^{r_j-1}\sum_{k=0}^{\kappa_j} d_{i,k,j} \rho^{-n} [z^n](1-z)^{\nu_j + i} \log^k\left(\frac{1}{1-z}\right)
\end{equation}
where all free variables on the right-hand side stand for computed values.
The choice of~$s$ at step~\ref{step:localparam} ensures that
$N_0 > s \abs{\alpha}$ at each call to Algorithm~\ref{algo:coeffasy};
hence, by Proposition~\ref{prop:coeffasy},
step~\ref{step:coeffasy} computes bounds that are valid for all~$n
\geq N_0$.
It follows that, for each $\rho$~and~$j$, step~\ref{step:contribution_explicit} yields a bound $E_{\rho,j}(n)$ for $[z^n] \ell_{\rho,j}$ also valid for all $n \geq N_0$.

Proposition~\ref{prop:coeffasy} also states that, for each $(i, k)$,
$e_k(n^{-1}, \log n)$ has degree at most~$\kappa_j$ with respect to $\log n$
and its coefficients of degree in~$n^{-1}$ less than~$r_j-i$ are elements of~$\rholgamma$.
The choice of~$r_j$ at step~\ref{step:localparam}, referring to Definition~\ref{def:landg}, ensures that we have
$\Rel \nu_j + r_j + 1 \geq \lambda + r_0$
where
$\lambda = \min_j \Rel \nu_j$.
This, combined with the algebraicity of $\rho$ and $d_{i,j,k}$, implies that the terms of $E_{\rho,j}(n)$ that are not contained in
$O\bigl( n^{-\lambda - r_0} \log(n)^{\kappa_j} \bigr)$
have coefficients belonging to~$\rholgamma$.

Turning to the local error term, let
\[
  \mathcal G_{\rho,j}(n) =
  \frac{1}{2\pi \i} \lim_{\varphi \to 0} \int_{\mS_{\rho}(n) + \mL_{\rho}(n)} \frac{g_{\rho}(z)}{z^{n+1}}.
\]
Step~\ref{step:analyticcase} implements Remark~\ref{rk:error_exact}:
when $n_j$ is a nonnegative integer and $\kappa_j = 0$,
one can see from~\eqref{eq:local_error_term} that
the function $z \mapsto g_\rho(z) z^{-n-1}$ is analytic at~$\rho$,
and hence on the disk
$\{|z - \rho| \leq R_1\}$
defined by~\eqref{eq:R0_choice}.
As the path $\mS_{\rho}(n) + \mL_{\rho}(n)$ tends to a closed contour contained in this disk as $\varphi \to 0$, we have
$\mathcal G_{\rho,j}(n) = 0$
in this case.
Otherwise steps \ref{step:tailbound}~and~\ref{step:contribution_local_error} are executed.
These steps are a direct application of Proposition~\ref{prop:bound_S} and Corollary~\ref{cor:bound_L}.
They yield a bound on
$\mathcal G_{\rho,j}(n)$
of the form
\[
  O\bigl( n^{\Rel \nu_j - 1 - r_j} \log(n)^{\kappa_j} \bigr)
  = O\bigl( n^{-\lambda - r_0} \log(n)^{\kappa_j} \bigr)
\]
and valid for all $n \geq \max(N_0, s \abs{\nu_j + r_j}) = N_0$.

Let $\mathcal C_\mB \in \mathbb R_{\geq 0}$ denote the bound computed at step~\eqref{step:bound_I}.
By Equation~\eqref{eq:decomposition_proof} we have
\[
  \left|
    [z^{n}] \Bigl(
      f(z) - \sum_{\rho \in \domsing} \sum_{j=1}^q
        c_{\rho,j} \ell_{\rho,j}(z)
    \Bigr)
  \right|
  \leq
    \sum_{\rho \in \domsing} \sum_{j=1}^q c_{\rho,j} \left|
      \frac1{2\pi \i} \int_{\mS_{\rho}(n) + \mL_{\rho}(n)} \frac{g_{\rho,j}(z)}{z^{n+1}} \, dz
    \right|
    + \abs{\mI_{\mB}}
\]
for all small enough~$\varphi>0$,
and Equation~\eqref{eq:bound_Ic} states that
$\lim_{\varphi \to 0} \abs{\mI_{\mB}} \leq R_0^{-n} \mathcal C_\mB$
for $R_0$ given by~\eqref{eq:R0_choice} (which agrees with the value computed at step~\eqref{step:R0N0}).
Therefore
\begin{equation}\label{eq:triangle}
  \left|
    [z^{n}] \left(
      f(z) - \sum_{\rho \in \domsing} \sum_{j=1}^q
        c_{\rho,j} \ell_{\rho,j}(z)
    \right)
  \right|
  \leq
    \sum_{\rho \in \domsing} \sum_{j=1}^q c_{\rho,j}
      \abs{\mathcal G_{\rho, j}(n)}
    + \mathcal C_\mB R_0^{-n}
\end{equation}
for all $n \geq N_0$.

Finally, step~\eqref{step:combine} ensures that the output represents a bound of
the form~\eqref{eq:estimate_form_main_2}.
More precisely, the simplification process sets~$\gamma = -\nu_j$ for one of the~$j$ such that $\Rel \nu_j = \lambda$, and $\kappa = \max_j \kappa_j$.
Combining the bounds $E_{\rho,j}(n)$ into the sum
$\sum_{\rho \in \domsing} \sum_{j=1}^q c_{\rho,j} E_{\rho,j}(n)$
yields an expression of the form~\eqref{eq:decomposition_proof}.
Since all contributions from steps
\ref{step:contribution_local_error}~and~\eqref{step:bound_I}
are in
$O\bigl( n^{-\lambda - r_0} \log(n)^{\kappa_j} \bigr)$,
adding them to~$R(n)$ preserves its form.
By~\eqref{eq:triangle}, the resulting bound holds for all $n \geq N_0$.
\end{proof}

\begin{remark}
  In special circumstances, such as when dealing with algebraic series or
  diagonals, it is possible to express the coefficients in closed form.
  In particular, when dealing with an algebraic series, by choosing bases of
  solutions of $\diffop$ at the origin and at each singular point
  that are also solutions of the algebraic relation, the
  corresponding connection matrix (in Definition \ref{def:connection_matrix})
  is simply a permutation matrix. In this case, instead of dealing with
  coefficients in $\rholgamma$, we only encounter elements of the
  $\overline\Q$-algebra generated by $\{\Gamma(\alpha)^{-1} : \alpha\in \overline{\Q}\}
  \cup \{\gamma^{(j)}(z) : z\in \overline{\Q}, j \in \N\}$.
\end{remark}

\section{Implementation and Further Examples}
\label{sec:implementation}

We have implemented the algorithm described in this article (up to minor
variations) using the SageMath computer algebra system.
Our implementation is part of the ore\_algebra
package~\cite{KauersJaroschekJohansson2015}, available at
\begin{center}
\url{https://github.com/mkauers/ore_algebra/}
\end{center}
under the GNU~General Public License.
The version described here corresponds to git revision \texttt{47e05a45}%
\footnote{
SWHID:
\href{https://archive.softwareheritage.org/swh:1:rev:47e05a4556c854847f0ed9239fc3e288fde28ab3/}%
{\texttt{swh:1:rev:47e05a4556c854847f0ed9239fc3e288fde28ab3}}.}.
The examples were run under SageMath 9.7.beta2%
\footnote{
SWHID:
\href{https://archive.softwareheritage.org/swh:1:rev:a6e696e91d2f2a3ab91031b3e1fcd795af3e6e62/}%
{\texttt{swh:1:rev:a6e696e91d2f2a3ab91031b3e1fcd795af3e6e62}}.}.
The documentation and test suite of ore\_algebra contain executable versions
of all examples from this paper, sometimes with minor changes.

\begin{footnotesize}
\begin{excont1}
Using this version of ore\_algebra,
Example~\ref{ex:lattice} (page~\pageref{excont1}) can be reproduced through the following commands:
\begin{verbatim}
sage: from ore_algebra import OreAlgebra
sage: from ore_algebra.analytic.singularity_analysis import bound_coefficients
sage: Pol.<z> = PolynomialRing(QQ)
sage: Dop.<Dz> = OreAlgebra(Pol)  # Dz represents the operator d/dz
sage: dop = (z^2*(4*z - 1)*(4*z + 1)*Dz^3 + 2*z*(4*z+1)*(16*z-3)*Dz^2
....:        + 2*(112*z^2 + 14*z - 3)*Dz + 4*(16*z + 3))
sage: bound_coefficients(dop, [1, 2, 6], order=2)
\end{verbatim}
On a standard laptop, the computation takes about 3.5\,s, of which roughly 3\,s
are spent bounding the global error term by evaluation on the big circle.
\end{excont1}
\end{footnotesize}

The implementation builds on pre-existing code in ore\_algebra
for computing the connection matrices of Definition \ref{def:connection_matrix} 
(see~\cite{Mezzarobba2016})
and for computing bounds on tails of logarithmic series solutions of D-finite equations,
as in Proposition~\ref{prop:diffopbound} (see~\cite{mezzarobba2019truncation}).
Except for singularities and local exponents, which are algebraic numbers and
are represented exactly, numeric coefficients are represented as
elements of SageMath's \texttt{ComplexBallField}, based on the Arb
library~\cite{Johansson2017}.
We perform intermediate computations that lead to the coefficients of the
output at a working precision selected by the user, with the occasional
addition of some guard digits for steps where we expect a loss of accuracy, but
do not attempt to provide any guarantees on the radius of the output intervals.
For operations that only affect the error terms, we currently use a fixed,
hardcoded working precision.
Our code also relies on SageMath's
\texttt{AsymptoticRing}~\cite{AsymptoticRing} to represent the asymptotic
expansion it outputs.

The implemented algorithm deviates from the one described here in some minor ways.
Perhaps the most significant difference is that we implement the following
variant of Remark~\ref{rk:expomodZ}:
at step~\eqref{step:foreachbasiselement} of Algorithm~\ref{algo:main},
elements~$y_j$ of the local basis are partitioned according to their value
modulo the integers of the exponent~$\nu_j$, and the computations associated to
of elements of a given class are carried out simultaneously.

Below we discuss some examples that illustrate the behaviour of our
implementation on ``real-life'' P-recursive sequences.
Except where noted, we call the algorithm with $r_0=2$, $n_0 = 50$, and an initial
working precision of 53~bits.
Taking $n_0 = 50$ makes the constants in the error terms slightly smaller than
with the default~$n_0=0$.
There is room for improvement in the performance of the code:
as of this writing, calls to \verb!bound_coefficients! take about 2~to~15
seconds each on a standard laptop, with the vast majority of the time spent
computing the global error term.
All outputs were slightly edited for readability.

\begin{footnotesize}
\begin{example}[Diagonals of symmetric rational functions]
Due to a connection to certain special functions, Baryshnikov et al.~\cite{baryshnikov2018diagonal} studied the \emph{diagonals} of the family of rational functions
$f(z_1, \ldots, z_d) = (1 - (z_1 + ... + z_d) + c \cdot z_1 \ldots z_d)^{-1}$
with $d\in\N$ and $c\in\R$, obtained by expanding $f(\mathbf{z})$ as a power series and taking the terms with monomials $(z_1z_2\cdots z_d)^n$ where all exponents are equal. Taking $d = 4$ and making the substitution $z_1z_2\cdots z_d=z$, the methods of \emph{creative telescoping} imply that the operator~\cite[Equation 11]{baryshnikov2018diagonal}
\begin{align*}
\diffop =& z^2 (c^4 z^4 + 4 c^3 z^3 + 6 c^2 z^2 + 4 c z - 256 z + 1) (3 c z - 1)^2 \, \tfrac{d^3}{dz^3} \\
        & + 3 z (3 c z - 1) (6 c^5 z^5 + 15 c^4 z^4 + 8 c^3 z^3 - 6 c^2 z^2 - 384 c z^2 - 6 c z + 384 z - 1) \, \tfrac{d^2}{dz^2} \\
        & + (c z + 1) (63 c^5 z^5 - 3 c^4 z^4 - 66 c^3 z^3 + 18 c^2 z^2 + 720 c z^2 + 19 c z - 816 z + 1) \, \tfrac{d}{dz} \\
        & + (9 c^6 z^5 - 3 c^5 z^4 - 6 c^4 z^3 + 18 c^3 z^2 - 360 c^2 z^2 + 13 c^2 z - 384 c z + c - 24)
\end{align*}
annihilates the diagonal $f_{\mathrm{diag}}(z)$.
Baryshnikov et al.\ showed that $f_{\mathrm{diag}}(z)$ is ultimately positive
if and only if $c < (d-1)^{d-1}$, with certain interesting phenomena happening
at $c=(d-1)^{d-1}$.
We illustrate this result for $d = 4$ and $c \in \{28, 27, 26\}$.

\paragraph{When $c = 28$:}
The diagonal $f_{\mathrm{diag}}(z)$ has the initial coefficient sequence
$(f_0, f_1, f_2) = (1, -4, -56)$.
Our implementation returns
\begin{align*}
    f_n & \in \phi^{n} n^{-3/2} \cdot \bigl(
    \begin{aligned}[t]
      &([-0.0311212622056357 \pm 10^{-16}] + [-0.0345183803114027 \pm 10^{-16}]\i) \\
      &+ ([0.050269964085834 \pm 10^{-15}] + [-0.0298161277530909 \pm 10^{-16}]\i) \, n^{-1} \bigr)
    \end{aligned} \\
    & \phantom{\in} + \overline{\phi}^{n} n^{-3/2} \cdot \bigl(
    \begin{aligned}[t]
      &([-0.0311212622056357 \pm 10^{-16}] + [0.0345183803114027 \pm 10^{-16}]\i) \\
      &+ ([0.050269964085834 \pm 10^{-15}] + [0.0298161277530909 \pm 10^{-16}]\i) \, n^{-1} \bigr)
    \end{aligned} \\
    & \phantom{\in} + B(6.11\, |\phi|^{n} n^{-7/2}, n \geq 50)
\end{align*}
where $\phi \approx 79.33 + 25.48i$ is algebraic of degree~$4$
and $B(\varepsilon_n, n \geq N_0)$ indicates an term of absolute value bounded
by~$\varepsilon_n$ for all $n \geq N_0$.
In this case, $f_n$ is not ultimately positive.

\paragraph{When $c = 27$:}
In this case $(f_0, f_1, f_2) = (1, -3, 9)$,
and our implementation returns
\begin{align*}
    f_n & \in \phi^n n^{-3/2} \cdot \bigl(
    \begin{aligned}[t]
        & ([0.306608607103967 \pm 10^{-15}] + [0.146433894558384 \pm 10^{-15}]\i) \\
        & + ([-0.26554984277221 \pm 10^{-15}] + [-0.03529869348794 \pm 10^{-15}]\i)n^{-1} \bigr)
    \end{aligned} \\
    & \phantom{\in} + \overline\phi^{n} n^{-3/2} \cdot \bigl(
    \begin{aligned}[t]
        & ([0.306608607103967 \pm 10^{-15}] + [-0.146433894558384 \pm 10^{-15}]\i) \\
        & + ([-0.26554984277221 \pm 10^{-15}] + [0.03529869348794 \pm 10^{-15}]\i)n^{-1} \bigr)
    \end{aligned} \\
    & \phantom{\in} + B(50.1\, \cdot \, 9^{n} n^{-7/2}, n \geq 50)
\end{align*}
where one can check that $\phi = -7 + 4\sqrt{2}\i$.
In this case $f_n$ is also not ultimately positive, however an interesting phenomenon observed in \cite{baryshnikov2018diagonal} is explicitly illustrated here: as $c \to 27$ the exponential growth rate of $|f_n|$ drops from around 81 to 9.

\paragraph{When $c = 26$:}
One has
$(f_0, f_1, f_2) = (1, -2, 76)$,
and our implementation gives
\[
    f_n \in \phi^{n} \cdot \bigl(
    \begin{aligned}[t]
        &[0.0484997667050581 \pm 10^{-16}] n^{-3/2} 
        + [-0.068160009777454 \pm 10^{-15}]n^{-5/2} \\
        &+ B(8.41 \, n^{-7/2}, n \geq 50) \bigr)
    \end{aligned}
\]
for a real algebraic number $\phi \approx 108.10$ of degree~$4$.
From this, we can immediately see that $f_n$ is positive for all $n \geq 50$, verifying the ultimate positivity derived in \cite{baryshnikov2018diagonal} using multivariate methods.
\end{example}
\end{footnotesize}

Next, we give an example that illustrates how our algorithm deals with complex exponents.

\begin{footnotesize}
\begin{example}[Complex exponents]
Consider the power series $f(z) = f_0 + f_1 z + f_2 z^2 + \cdots$ with the initial conditions
$(f_0, f_1, f_2) = \left(1, 2, -1/8 \right)$
satisfying $\diffop f = 0$ where
\[
\diffop = (z-2)^2 \frac{d^2}{dz^2} + z(z-2) \frac{d}{dz} + 1.
\]
Algorithm \ref{algo:main} finds that
\[
    f_n \in 2^{-n} \cdot \bigl(
    \begin{aligned}[t]
      & ([1.1243375066147 \pm 10^{-14}] + [-0.4622196104635 \pm 10^{-14}]\i) {n^{- \alpha \i - 1/2}} \\
      & + ([1.1243375066147 \pm 10^{-14}] + [-0.4622196104635 \pm 10^{-14}]\i) {n^{\alpha \i - 1/2}} \\    
      & + ([-0.4002939247887 \pm 10^{-14}] + [-0.9737048431560 \pm 10^{-14}]\i) {n^{-\alpha \i - 3/2}} \\
      & + ([-0.4002939247887 \pm 10^{-14}] + [0.9737048431560 \pm 10^{-14}]\i) {n^{\alpha \i - 3/2}} \\
      & + B(9 \cdot 10^3 \, n^{-5/2}, n_0 \geq 50) \bigr),
    \end{aligned}
\]
with $\alpha = \frac{\sqrt{3}}{2}$.
We can conclude in particular that the series $(f_n)$ is not ultimately
positive, since for any $\varepsilon>0$, there exist infinitely many $n$ such
that $\arg(n^{\alpha \i}) \in (\pi-\varepsilon, \pi)$.
\end{example}
\end{footnotesize}

The following example from \cite{KauersPillwein2010} illustrates how a priori knowledge about dominant singularities can affect the usefulness of the bound produced.

\begin{footnotesize}
\begin{example}[Difficulty of certifying singularities]\label{ex:cert_sing}
Consider the sequence $(f_n)$ with the initial conditions
$(f_0, f_1) = \left(1, \frac{1}{4}\right)$
satisfying the recurrence equation
\[
(n+3)^2 f_{n+2} - \frac{1}{2}  (n+2) (3n+11) f_{n+1} + \frac{1}{2} (n+4) (n+1) f_n = 0.
\]
The generating function $f(z)$ of $\{f_n\}$ is a solution of the operator
\[
\diffop = \left(\frac{1}{2} z^4 - \frac{3}{2} z^3 + z^2\right) \frac{d^4}{dz^4} + \left(7 z^3 - 16 z^2 + 7 z\right) \frac{d^3}{dz^3} + \left(26 z^2 - 41 z + 9\right) \frac{d^2}{dz^2} + \left(26 z - 22\right) \frac{d}{dz} + 4.
\]

Without any a priori knowledge of dominant singularities of $f$, the algorithm assumes the singular point of $\diffop$ with the smallest modulus apart from 0, which is $z = 1$, to be the dominant singularity of $f$.
Our implementation, with the initial working precision raised to 1000~bits, determines that
\begin{align*}
    f_n \in [\pm 10^{-300}] + [\pm 10^{-300}]\i + B(3 \cdot 10^3 \, (4/7)^n, n \geq 50)
\end{align*}
This estimate does not give much useful information about the asymptotic behaviour of $f_n$, since we do not know if the dominant term is zero or not.
The output suggests however that the corresponding constant might indeed be zero, in other words, that $f$~might be analytic at~$z=1$.
A direct computation shows that indeed $f_n = \frac{2^{-n}}{n+1}$ and $f(z) = \frac{1}{2}\log \frac{1}{1 - z/2}$.

Adding $\anasing = \{0, 1\}$ to the input of Algorithm \ref{algo:main} results in the bound
\[
    f_n \in 2^{-n} \cdot \bigl( [1.0 \pm 10^{-15}] n^{-1}
    + ([1.0 \pm 10^{-15}]) n^{-2}
    + B(66 \, n^{-3} \log n, n \geq 50) \bigr),
\]
which characterizes the dominant asymptotic behaviour of $f_n$.
\end{example}
\end{footnotesize}

\section*{Acknowledgements}

We thank Joris van der Hoeven for sharing a draft version
of~\cite{vanderhoeven:hal-03291372},
and for several useful discussions.
We thank Bruno Salvy for the reference to~\cite[Section~IV.2]{Salvy1991}.
Finally, we thank the two anonymous reviewers whose comments helped clarify this manuscript.

\printbibliography
\end{document}